\documentclass[sigconf]{aamas}

\usepackage{balance} 

\usepackage{amstext}
\usepackage{amsmath}
\usepackage{amsthm}
\usepackage{amsfonts}
\usepackage{bm}
\usepackage{multirow}
\usepackage{multicol}
\usepackage{array}
\usepackage{enumitem}
\usepackage{algorithm}
\usepackage{algpseudocode}

\makeatletter
\gdef\@copyrightpermission{
  \begin{minipage}{0.2\columnwidth}
   \href{https://creativecommons.org/licenses/by/4.0/}{\includegraphics[width=0.90\textwidth]{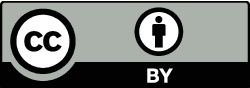}}
  \end{minipage}\hfill
  \begin{minipage}{0.8\columnwidth}
   \href{https://creativecommons.org/licenses/by/4.0/}{This work is licensed under a Creative Commons Attribution International 4.0 License.}
  \end{minipage}
  \vspace{5pt}
}
\makeatother

\setcopyright{ifaamas}
\acmConference[AAMAS '25]{Proc.\@ of the 24th International Conference
on Autonomous Agents and Multiagent Systems (AAMAS 2025)}{May 19 -- 23, 2025}
{Detroit, Michigan, USA}{Y.~Vorobeychik, S.~Das, A.~Nowé  (eds.)}
\copyrightyear{2025}
\acmYear{2025}
\acmDOI{}
\acmPrice{}
\acmISBN{}

\acmSubmissionID{714}

\title[AAMAS-2025 Formatting Instructions]{Single-Agent Planning in a Multi-Agent System: \\A Unified Framework for Type-Based Planners}

\author{Fengming Zhu}
\affiliation{
  \institution{Hong Kong University of Science and Technology}
  \city{Hong Kong SAR}
  \country{China}}
\email{fzhuae@connect.ust.hk}

\author{Fangzhen Lin}
\affiliation{
  \institution{Hong Kong University of Science and Technology}
  \city{Hong Kong SAR}
  \country{China}}
\email{flin@cs.ust.hk}

\begin{abstract}
We consider a general problem where an agent is in a multi-agent environment and must plan for herself without any prior information about her opponents.
At each moment, this pivotal agent is faced with a trade-off between exploiting her currently accumulated information about the other agents and exploring further to improve future (re-)planning.
We propose a theoretic framework that unifies a spectrum of planners for the pivotal agent to address this trade-off.
The planner at one end of this spectrum aims to find exact solutions, while those towards the other end yield approximate solutions as the problem scales up.
Beyond theoretical analysis, we also implement \textbf{13} planners and conduct experiments in a specific domain called \textit{multi-agent route planning} with the number of agents \textbf{up to~50}, to compare their performaces in various scenarios.
One interesting observation comes from a class of planners that we call \textit{safe-agents} and their enhanced variants by incorporating domain-specific knowledge, which is a simple special case under the proposed general framework, but performs sufficiently well in most cases.
Our unified framework, as well as those induced planners, provides new insights on multi-agent decision-making, with potential applications to related areas such as mechanism design.
\end{abstract}


\newcommand{\BibTeX}{\rm B\kern-.05em{\sc i\kern-.025em b}\kern-.08em\TeX}

\begin{document}


\pagestyle{fancy}
\fancyhead{}


\maketitle 


\section{Introduction}
\label{sec:intro}

We consider a general problem where an agent is in a multi-agent system and is supposed to plan for herself without any prior information about her opponents.
A motivating example of such a scenario is in autonomous driving~\cite{chen2024end}, where a fully autonomous vehicle needs to navigate herself to the destination, but inevitably shares the roads with other cars that she may not know much about, ones either driven by human beings or possibly controlled by AI softwares from another company.
Therefore, she must learn the behavioral patterns of the other cars in real-time and act accordingly.

There are three unique features that distinguish this setting from other seemingly related ones.
\textbf{Firstly}, it is clearly different from a system where all the agents are under the control of a central planner, such as programming robot fleets to deliver packages in warehouses~\cite{stern2019multi-overview, stern2019multi, sharon2015conflict, li2021eecbs}, which is also known as \textit{multi-agent path finding} (MAPF) problems.
\textbf{Secondly}, the setting that we are investigating does not ensure pure competitiveness among participants, and therefore, cannot be modelled as zero-sum games (win or lose).
In typical zero-sum games, e.g., Go/chess~\cite{doi:10.1126/science.aar6404, schrittwieser2020mastering, danihelka2022policy, antonoglou2022planning} and poker~\cite{zinkevich2007regret},
a player can assume the other one will do the most harm, and compute a minimax strategy, which coincides with the Nash equilibrium (NE) and later serves as a never-lose strategy regardless of what strategy the opponent eventually follows.
\textbf{Thirdly}, an oracle that computes a sample NE dose not \textit{directly} offer much help, even if all agents are aligned with common interests~\cite{boutilier1996planning, claus1998dynamics}.
Chances are that there may be multiple NEs or the other agents are simply not playing the equilibrium strategies at all.

Conceptually, 
this general class of problems is usually addressed by conducting opponent modelling and opponent-aware replanning at the same time. 
Given no knowledge about the upcoming opponents, the controlled agent will start with an initial belief over possible opponent models (or types). With time going on, actions done by the other players are observed, which can be utilized to update the belief. Meanwhile, the revised belief can be employed for replanning. Essentially, to control the agent is to find a trade-off between exploiting current knowledge and acquiring more information about  opponents to enhance future replanning~\cite{carmel1998explore, carmel1999exploration}.

In principle, it can be formalized as a stochastic game~\cite{shapley1953stochastic, solan2015stochastic} with incomplete information~\cite{harsanyi1967games}, and the incompleteness is due to the lack of knowledge about opponents.
There is a vast literature on how to devise an effective strategy for the controlled agent, mostly based on best responses~\cite{kalai1993rational, boutilier1996planning, claus1998dynamics, carmel1998explore, carmel1999exploration, albrecht2015game, eck2020scalable, schwartz2023bayes, schwartz2023combining, rahman2023general}.
However, one can clearly observe two extremes along this line of work.
At one extreme, researchers apply deliberate formulations on toy artificial problems such as repeated games with one state transiting to itself~\cite{kalai1993rational, boutilier1996planning, claus1998dynamics, carmel1998explore, carmel1999exploration, albrecht2015game}, serving primarily as proof-of-concept examples.
In fact, repeated games cannot fully reflect the complexity underneath as they are just special cases of stochastic games.  
At the other extreme, people jump to highly approximated ones, namely Monte Carlo methods~\cite{albrecht2015game, eck2020scalable, schwartz2023bayes, schwartz2023combining} or complex Neural Network (NN) pipelines~\cite{rahman2023general}, yet the largest known domain remains only $10^{14}$ states with four agents~\cite{schwartz2023bayes, schwartz2023combining, rahman2023general}.
We here ask a significant research question: \textbf{Is there any viable formulation in between, and even beyond?}
More precisely, given such a multi-agent planning problem with $10^6$ states ($\ll 10^{14}$),
can we identify a planner that is more accurate than Monte Carlo ones, 
and is there any planner that can scale to even larger state spaces ($\gg 10^{14}$)?


%
%
%
%
%

To this end,
we first present a set of formulations that each of them \textit{rediscovers and generalizes} an existing one from the literature.
More specifically, to find an optimal plan, one has to compute an exact solution for the underlying infinite-horizon \textit{partially observable Markov decision process} (POMDP)~\cite{smallwood1973optimal, sondik1978optimal, kaelbling1998planning}.
This is only feasible for tiny problem instances, as optimally solving infinite-horizon POMDPs is computationally expensive~\cite{papadimitriou1987complexity, madani1999undecidability}.
As the environment scales up in terms of the number of agents, we have to appeal to approximated formulations.
We then propose a general framework that unifies this entire set of formulations.
This unified framework is designed to encompass a spectrum that interpolates the aforementioned formulations, enabling users to devise new ones for those problems of intermediate scales as well as much larger scales.
\textbf{The theoretic framework is basically to perform layered lookahead search and backup from values of leaf nodes in each layer, wherein different formulations implement different lookahead search.}
The induced formulation at one end of this spectrum aims at solving exact solutions, while the formulations towards the other extreme lead to approximated ones as the problem scales up.
Within our proposed framework, approximations can be made by controlling the depth of tree search, altering node evaluation functions, adopting sampling-based backup, toggling heuristic action selection, etc. The framework effectively connects \textit{POMDP}~\cite{kaelbling1998planning}, \textit{contextual-MDP/RL}~\cite{hallak2015contextual, benjamins2021carl}, \textit{tree search}~\cite{browne2012survey} and even \textit{constraint satisfaction}~\cite{brailsford1999constraint}.



It is highly non-trivial to compare those formulations purely in theory. Therefore, we implement a group of planners\footnote{Code available at
\url{https://github.com/Fernadoo/BestResponsePlanning}.}
resulted from those formulations in a concrete domain, called \textit{multi-agent route planning} (MARP)
This benchmark is challenging in two aspects: 1)~it is no longer repeated games, instead, involves complex state transitions; 2)~it necessitates long-horizon planning, as the rewards are goal-conditioned and hence sparse.
We select this domain on purpose to take the advantage that it can 
generate problem instances (parameterized by the map size and the number of agents) along a fine-grained axis of scales, from around 70 states with two agents to around $4.62 \times 10^{145}$ states with 50 agents, matching our need of studying the capacities of the planners in the spectrum.
As it is impossible to run each proposed planner against all possible hypothetical opponents and their combinations,
we instead test our planners against three representative groups of opponents, including rational ones, malicious ones, and self-playing ones.
Among all the planners, we also draw one's attention to a particular class, what we call \textit{safe-agents} and their enhanced variants, which turns out to be a special case of our general framework but performs sufficiently well in most cases. It is also simple enough to save a great amount of computation.

\section{Related Work}
\label{sec:related}

The problem that we investigate pertains to a board range of theoretic areas, primarily \textit{opponent modelling} and the control theory of \textit{POMDP}, alongside  practical tools, such as efficient \textit{NE solvers}.

\textbf{Opponent Modelling.}
\citet{albrecht2018autonomous} surveyed a number of methodologies under this topic, where our work falls into the category of type-based reasoning.
Types usually refer to predefined opponent models.
For instance, \citet{carmel1998explore, carmel1999exploration} investigate a similar problem by modelling opponent strategies as deterministic finite automata,
while other researchers have approached this by using reactive policy oracles, either given by human experts~\cite{albrecht2015game} or trained by separated pipelines~\cite{schwartz2023bayes, schwartz2023combining}.
An critical follow-up issue is how to integrate opponent modelling with opponent-aware planning. 
For repeated games, much of the literature emphasizes planners that best respond to the current belief over hypothetical opponent models, e.g., $\epsilon$-BR~\cite{kalai1993rational}, IL/JAL~\cite{boutilier1996planning, claus1998dynamics}, CJAL~\cite{banerjee2007reaching}, etc.
For a more general setting, HBA~\cite{albrecht2015game} is proposed as a belief-dependent planning operator that takes long-term returns into account. The authors have also examined some convergence results~\cite{albrecht2019convergence}, and the conditions under which the belief correctly reveals the truth~\cite{albrecht2016belief}.
Recent studies have also explored MCTS methods~\cite{schwartz2023bayes, schwartz2023combining}, and neural network approximations~\cite{rahman2023general}.

\textbf{POMDP.}
We later formalize the underlying problem as a \textit{Contextual MDP} (CMDP)~\cite{hallak2015contextual}, where the set of contexts corresponds to the set of all possible opponent models.
However, CMDPs reduce to POMDPs~\cite{smallwood1973optimal, sondik1978optimal, kaelbling1998planning} when those contexts are not revealed. 
One can therefore resort to existing techniques in the literature of the POMDP theory to alleviate the computational burden~\cite{littman1995learning, zhang2001speeding}.
Moreover, some meta/contextual learning techniques developed by the RL community may also be related \cite{beck2023survey, benjamins2021carl}.


\textbf{NE Solvers.}
Lastly, we want to clarify that our work may \textit{not directly} align with the interests of researchers who are investigating the problem of \textit{Multi-Agent Reinforcement Learning} (MARL)~\cite{lowe2017multi, foerster2018counterfactual, yu2022the, pmlr-v162-fu22d, JMLR:v24:22-0169, samvelyan19smac} or \textit{Multi-Agent Path Finding} (MAPF)~\cite{stern2019multi-overview, stern2019multi, sharon2015conflict, zhang2022multi, li2021eecbs}.
Both lines of work focus on solving a sample NE in a centralized manner, which does not strictly fall into the scope of our discussion as they force all agents to execute the NE plan that is found upfront.
Nevertheless, as we will show that with the help of an oracle that computes an NE very fast in real-time, our proposed framework can effectively repair it while replanning.
We emphasize that this routing domain resembling MAPF is selected as our benchmark solely for the purpose of evaluating and illustrating the capabilities of our proposed planners.

%


%
%
%


\section{Model and Solution Concepts}
\label{sec:model}

\subsection{Stochastic Game}

The whole system where the agents interact is modelled as a \textit{stochastic game} (SG, also known as Markov games)~\cite{shapley1953stochastic, solan2015stochastic}, which 
can be seen as an extension of both \textit{normal-form games} (to dynamic situations with stochastic transitions) and \textit{Markov decision processes} (to strategic situations with multiple agents).
A stochastic game
is a 5-tuple $\langle\mathcal{N},\mathcal{S}, \mathcal{A}, T, R\rangle$ given as follows,
\begin{enumerate}
	\item $\mathcal{N}$ is a finite set of agents.
	\item $\mathcal{S}$ is a finite set of normal form games. We call each possible $S \in \mathcal{S}$ a stage game. In the proceeding discussion, each $S$ is also called an environment state, or just a state.
	\item $\mathcal{A} = \mathcal{A}_1 \times \cdots \times \mathcal{A}_n$ is a set of joint actions, where $\mathcal{A}_i$ is the action set of agent $i$. We write $a_i$ as the action of agent $i$ and $a = (a_i, a_{-i})$ (without any subscript) as the joint action.
	\item $T: \mathcal{S} \times \mathcal{A}_1 \times \cdots \mathcal{A}_n \mapsto \Delta(\mathcal{S})$ defines a stochastic transition among stage games.
	\item $R_i: \mathcal{S} \times \mathcal{A}_1 \times \cdots \mathcal{A}_n \mapsto \mathbb{R}$ denote the utility function of agent $i$ in each stage game.
\end{enumerate}
Each agent $i$ is only aware of her own utilities $R_i$. In fact, knowing other agents' utility functions does not help, as she has no knowledge about the opponents until the game begins and will be possibly exposed to different opponents in each match.
This adds extra difficulty of evaluation as we will cover that later in Section \ref{sec:benchmark}.
Following the terminology in \cite{albrecht2018autonomous}, we also call agent $i$ the modelling agent under our control, and agent $-i$ the modelled agents.

\subsection{The Modelling Agent's Problem}

The modelling agent under such an environment acts in a manner of maximizing her own accumulated discounted utility, i.e.,
$
	\mathbb{E}_{\tau \sim SG}[\sum_{t=0}^{T(\tau)} \gamma^tR_{i,t}],
$
where $\tau$ denotes the trajectory of one single match of the SG, $\gamma < 1$ denotes the discount factor, and $T(\tau)$ denotes the length of the trajectory.
For the sake of simplicity, we assume perfect recall, i.e., the modelling agent can remember the entire history from the beginning of the match up to the current moment, including every past states and joint actions.
Formally, a history segment is given as $h \in \mathcal{H} \triangleq (\mathcal{S} \times \mathcal{A})^*$, i.e. a sequence consisting of alternating states and joint-actions.

Therefore, a candidate solution concept for the modelling agent is a history-depenent stochastic policy $\pi_i: \mathcal{H} \times \mathcal{S} \mapsto \Delta(\mathcal{A}_i)$ that maps from all possible history segments and current states to (randomized) available actions.
From agent $i$'s perspective, she also assumes every other opponent $j$'s potential strategy is drawn from a set of basis policies $\Pi_j$, where each $\pi_j^k \in \Pi_j$ is a stationary Markov policy that maps states to actions,
i.e. $\pi_j^k: \mathcal{S} \mapsto \Delta(\mathcal{A}_j)$.
This is indeed a strong assumption that suggests those opponents do not care about histories, and they are not learning and evolving during the game play.

We then show, from the modelling agent's perspective, conditioned on the strategies of opponents, one can model the rest of the problem as a \textit{Contextual  Markov Decision Process} (CMDP)~\cite{hallak2015contextual}.
The resulted formulation is given as a 3-tuple $\langle\mathcal{C}, \mathcal{A}, \mathcal{M}\rangle$,
\begin{enumerate}
	\item $\mathcal{C}$ is a set of contexts. In our setting, $\mathcal{C} \triangleq \Pi_{-i} = \prod_{j\neq i} \Pi_j $, i.e., the set of contexts are exactly the set of all possible combinations of opponent models.
	\item $\mathcal{A}$ is a set of actions. In our setting, $\mathcal{A} \triangleq \mathcal{A}_i$, as we are approaching the problem from agent $i$'s perspective.
	\item $\mathcal{M}$ is a mapping from contexts to corresponding Markov decision processes.
	An interpretation is, if the modelling agent knows for sure how her opponents would act, then she is basically faced with a single-agent MDP with opponent behaviors subsumed into transition/utility functions as noises.
	The resulted MDP from the perspective of agent $i$, denoted as $\mathcal{M}(\pi_{-i})$,
	is induced as a 5-tuple
	$\langle\mathcal{S}, \mathcal{A}_i, T^{\pi_{-i}}, R^{\pi_{-i}}, \gamma\rangle$:
	\begin{itemize}
		\item $\mathcal{S}, \mathcal{A}_i$ and $\gamma$ inherit from the previous setup,
		\item $T^{\pi_{-i}}(S'| S, a_i) \triangleq \sum_{a_{-i} \in \mathcal{A}_{-i}}T(S'|S,a)\pi_{-i}(a_{-i}|S)$,
		\item $R^{\pi_{-i}}(S, a_i) \triangleq \sum_{a_{-i} \in \mathcal{A}_{-i}}R_i(S,a)\pi_{-i}(a_{-i}|S)$,
	\end{itemize}
\end{enumerate}
If the context is revealed, it is equivalent to augmenting the MDP state by an additional dimension representing the context, and such a CMDP can therefore be factorized into $|\mathcal{C}|$ MDPs and solved one by one via any off-the-shelf solver accordingly.
However, the issue is that in our situation the context is unobservable as it represents the unknown strategies of opponents.
It then reduces to solving the corresponding POMDP~\cite{smallwood1973optimal, sondik1978optimal, kaelbling1998planning}, where a state in the POMDP consists of the environment state and the underlying context (opponent model), and the observation function takes as input a POMDP state, and returns only the environment state exposed to everyone yet keeps the context (opponent model) hidden.

We also assume that the pivotal agent model every other opponent \textit{independently}, and each independent belief is a distribution over potential basis policies. Mathematically, for any belief $b \in \mathcal{B} \triangleq \prod_{j\neq i}\Delta(\Pi_j)$, we have $\pi_{-i}(a_{-i}|s) = \prod_{j\neq i}\pi_j(a_j|s)$ and $b(\pi_{-i}) = \prod_{j\neq i}b(\pi_j)$.
Depending on the set of presumed opponent policies, any belief $b_t$ at time step $t$ can be updated to a successor belief $b_{t+1}$ inductively by a revision operator, denoted as $\xi: \mathcal{B}\times \mathcal{S}\times \mathcal{A} \mapsto \mathcal{B}$. Equivalently speaking, given a prior belief and a sequence consisting of past states, and the actions taken under each corresponding state, a posterior belief can be inferred. Here we formulate it in a Bayesian way, $b_{t+1} = \xi(b_t, S, a)$ is given as,
\begin{equation}
\begin{split}
	b_{t+1}(\pi_j^k) \triangleq \mathbb{P}[\pi_j^k | S, a_j]
	& = \frac{ (\mathbb{P}[a_j|S, \pi_j^k]\cdot \mathbb{P}[\pi_j^k|S])^{1/\beta} }
	{\sum_{\pi_j^l \in \Pi_j} (\mathbb{P}[a_j|S, \pi_j^l]\cdot \mathbb{P}[\pi_j^l|S])^{1/\beta} } \\
	& = \frac{ (\pi_j^k(a_j|S)\cdot b_t(\pi_j^k))^{1/\beta} }
	{\sum_{\pi_j^l \in \Pi_j} (\pi_j^l(a_j|S)\cdot b_t(\pi_j^l))^{1/\beta} } \\
\end{split}
\label{eq:bayes}
\end{equation}
where $\beta$ is a tunable temperature parameter.
If $\beta = 1$ is always the case, it will be exactly the same as how beliefs are computed in the corresponding POMDP. However, altering the value of $\beta$ will open up a broader set of choices for implementations, as we will illustrate in Table \ref{tab:all_planners}.


\subsection{Potential Solution Formulations}
\label{sec:solutions}

\subsubsection{Exact Dynamic Programming}


As we mentioned, the modelling agent is faced with a CMDP with unobservable contexts which subsequently reduces to a POMDP.
Following from the control theory of POMDP~\cite{smallwood1973optimal}, one can actually find out that at each moment a compound \textit{information state} consisting of the environment state and the belief serves as a sufficient statistic, summarizing all the past history.
Hence, the transitions among those information states from the modelling agent's perspective, denoted as $\mathcal{T}: \mathcal{S}\times \mathcal{B} \times \mathcal{A}_i \mapsto  \Delta(\mathcal{S})\times \mathcal{B}$, are resulted from a joint effect of the environment transition and a sampling process of opponents' actions according to the belief. Mathematically,
{\small
\[
\mathcal{T} \Big ( (S', b') \Big| (S, b), a_i  \Big ) \triangleq 
\sum_{\pi_{-i}\in b} b(\pi_{-i})\sum_{a_{-i} \in \mathcal{A}_{-i}} T(S'|S, a) \pi_{-i}(a_{-i}|S)
\]
}
when $T(S'|S, a) > 0 \text{ and } b' = \xi(b, S, a)$, and all the other transitions are invalid and assigned zero probability.
The expected reward from the modelling agent's perspective is induced analogously,
{\small
\[
\mathcal{R} \Big((S, b), a_i \Big ) \triangleq 
\sum_{\pi_{-i}\in b} b(\pi_{-i})\sum_{a_{-i} \in \mathcal{A}_{-i}} R_i(S, a) \pi_{-i}(a_{-i}|S)
\]
}
Then, the Bellman optimality equation can be established over a continuous space,
\begin{equation}
{\small
\begin{split}
	& V_i(S,b) \\
	& = \max_{a_i \in \mathcal{A}_i}
	\bigg\{  \mathcal{R} \Big( (S, b), a_i \Big)
	+ \gamma \sum_{S', b'} \mathcal{T} \Big( (S', b') \Big| (S, b), a_i  \Big)\cdot V_i(S',b') \bigg\} \\
	& = \max_{a_i\in \mathcal{A}_i}\{\mathbb{E}_{ \substack{\pi_{-i} \sim b, \\ a_{-i}\sim \pi_{-i}} }
	[R_i(S,a) +  \gamma \sum_{S'}T(S'|S,a) V_i(S', b')]\}
\end{split}
}
\label{eq:pomdp}	
\end{equation}
when $b' = \xi(b, S, a)$, and $V_i(\cdot,\cdot)$ is the desired optimal value function for the modelling agent.
Again, if $\xi$ is implemented with $\beta = 1$, then solving Equation~(\ref{eq:pomdp}) is equivalent to optimally solving the corresponding infinite-horizon POMDP.
Despite that it is computationally costly in theory as revealed by~\cite{papadimitriou1987complexity, madani1999undecidability},
we can resort to off-the-shelf POMDP solvers like \textit{pomdp-solve}
\footnote{https://pomdp.org/code/} to solve small problem instances in practice, as we will show in Appendix~\ref{apd:case_study}.


As one may notice, this formulation resembles the HBA operator proposed by~\cite{albrecht2015game}, where the authors thereof call it a best-response rule based on Bellman control. We here point out that it is not just any random rule, but also the exact characterization of what the modelling agent is faced with, and the optimal control strategy can be therefore \textit{derived} and computed upfront.


In fact, Equation~(\ref{eq:pomdp}) also leads to potential (contextual-)RL solutions~\cite{benjamins2021carl}.
Subsequently, the most challenging part lies in how to effectively train such a policy that converges to the desired optimum.
As we will show in Appendix~\ref{apd:exp_setting}, there is usually an intermediate plateau phase before convergence, which we conjecture is caused due to the gap between feasible strategies and optimal strategies. 


\subsubsection{Belief-Induced MDP}
\label{sec:belief_induced_mdp}


As an alternative, one can amortize the burden of computing an optimal strategy completely in advance over repeated online replanning during the game play.
We first extend the aforementioned context-to-MDP mapping $\mathcal{M}(\pi_{-i})$ to $\mathcal{M}(b)$, allowing inputs of distributions of contexts,
\begin{itemize}
\item $T^{b}(S'| S, a_i) \triangleq \sum_{\pi_{-i}\in b} b(\pi_{-i})\sum_{a_{-i} \in \mathcal{A}_{-i}}T(S'|S,a)\pi_{-i}(a_{-i}|S)$
\item $R^{b}(S, a_i) \triangleq \sum_{\pi_{-i}\in b} b(\pi_{-i})\sum_{a_{-i} \in \mathcal{A}_{-i}}R_i(S,a)\pi_{-i}(a_{-i}|S)$
\end{itemize}
Therefore, optimally solving the belief-induced MDP $\mathcal{M}(b_t)$ at each time $t$ is equivalent to solving a surrogate target,
\begin{equation}
\begin{split}
& V_i(S,b_t) \\
&  = \max_{a_i\in A_i}\{\mathbb{E}_{ \substack{\pi_{-i} \sim b_t, \\ a_{-i}\sim \pi_{-i}} }
[R_i(S,a) +  \gamma \sum_{S'}T(S'|S,a) V_i(S', b_t)]\}	
\end{split}
\label{eq:belief_mdp}
\end{equation}
where $V_i(\cdot,b_t)$ is the optimal value function characterizing the optimal control of the belief-induced MDP $\mathcal{M}(b_t)$.
Note that the belief on the LHS is the same as that on the RHS, i.e., $b_t$ is fixed as a hyperparameter.
That is to say, the modelling agent is assuming the others will play the stationary mixed strategies throughout the rest of the game according to the current belief.

One may notice that solving such an online surrogate target involves inherent inconsistency issues. That is, the effect on revising the belief by observing future plays of the opponents is ignored while inducing such MDPs, which is exactly the cause that leads to sub-optimal strategies for the modelling agent.

Understandably, one can then come up with two alternatives. If solving a newly induced MDP at each move is computationally feasible, then repeatedly updating the belief and solving the MDP induced by the updated belief will be a better choice.
In contrast, if compiling and solving an MDP is costly or the belief itself is not worthy of iterative update, one might as well only compute the optimal policy of the MDP induced by a cautiously selected initial belief and commit to it until the end of the game. 



By Equation (\ref{eq:belief_mdp}) we generalize what is called \textit{independent learners} (ILs)~\cite{claus1998dynamics} in repeated games, which evaluate the accumulated return of the modelling agents with the opponents subsumed in the environment as stationary noises and best respond accordingly, to the ones that act under the same principle in stochastic games. 

\subsubsection{Belief-mixed MDP (QMDP)}

One may want to figure out the best response against each opponent model in the first place and then mix them afterwards according to the real-time belief, i.e.,
the action choice at each time step $t$ is hence given as
\begin{equation}
\begin{split}
a_t^* 
& \in \arg\max_{a_i\in \mathcal{A}_i}  \sum_{\pi_{-i} \in \Pi_{-i}} b_t(\pi_{-i})\cdot
	Q_{\mathcal{M}(\pi_{-i})}(S, a_i) \\
& = \arg\max_{a_i\in \mathcal{A}_i} \sum_{(S, \pi_{-i}) \in \mathcal{S}\times \Pi_{-i}} \mathbb{P}[(S, \pi_{-i})] \cdot Q((S, \pi_{-i}), a_i)  \\
\end{split}
\label{eq:qmdp}
\end{equation}
where $Q_{\mathcal{M}(\pi_{-i})}$ is the optimal Q-function of $\mathcal{M}(\pi_{-i})$ that can be solved upfront, but there are in total $|\Pi_{-i}|$ MDPs to solve, which grows exponentially in terms of the number of agents.
In fact,  the second line is exactly the QMDP formulation from the POMDP community~\cite{littman1995learning},
which pretends that the modelling agent can observe the underlying opponent model, and the Q-function $Q((S, \pi_{-i}), a_i)$ is obtained by optimally solving the underlying hypothetical MDP.
The equality holds because the opponents are assumed not evolving their strategies, and therefore, those factorized MDPs, namely all such $\mathcal{M}(\pi_{-i}), \forall \pi_{-i} \in \Pi_{-i}$, are independent of each other.




Moreover, Equation (\ref{eq:qmdp}) extends 
\textit{joint action learners} (JALs)~\cite{claus1998dynamics} in repeated games to the ones that act under the same principle in stochastic games. A JAL learns Q-values with respect to all possible opponent actions, and then assesses the best expected return by mixing the Q-values over the observed action distribution of the opponent.
Clearly, ours generalizes opponent actions (in single-state scenarios) to presumed opponent policies (in multi-state scenarios), on which a corresponding Q-function is derived.

\textit{We postpone the pseudocode for these three planners to Appendix~\ref{apd:planner_detail}, and some case studies to Appendix~\ref{apd:case_study}.}

\section{Unification}

\label{sec:unification}

In this section, we propose a unified planning framework conceptually based on tree search.
The general framework has all the aforementioned formulations embedded and even sheds lights on new ones.
Given the state $S_t$ and the belief $b_t$ at each step $t$, we do \textit{three layers of lookahead search} rooted at the node $(S_t, b_t)$, denoted as $(S^0_t, b^0_t)$ inside the tree (with $t$ omitted below),

%
%
\begin{enumerate}[left=0pt, itemsep=0pt]
	\item Belief-updated lookahead for the first $n$ levels, i.e., for $l\in[0, n)$:
	\[
		V_i(S^{l},b^l)
		\gets \max_{a_i\in \mathcal{A}_i}
		\sum_{ \substack{\pi_{-i} \in b^{l}, \\ a_{-i} \in \pi_{-i}} }
		[R_i(S^{l},a)
		+  \gamma \sum_{S^{l+1}}T(S^{l+1}|S^l,a) V_i(S^{l+1}, b^{l+1})],
	\label{eq:belief_bellman}
	\]
	where $b^{l+1} = \xi(b^l, S^l, a)$,
	
	\item Belief-fixed lookahead for the next $m$ levels, for $l\in [n,n+m)$:
	\[
		V_i(S^{l},b^{n})
		\gets \max_{a_i\in \mathcal{A}_i}
		\sum_{ \substack{\pi_{-i} \in b^{n}, \\ a_{-i} \in \pi_{-i}} }
		[R_i(S^{l},a)
		+  \gamma \sum_{S^{l+1}}T(S^{l+1}|S^l,a)V_i(S^{l+1}, b^{n})],
	\label{eq:van_bellman}
	\]
	\item Heuristic evaluation for the leaf nodes in level $(n+m)$:
	\[
		V_i(S^{n+m},b^{n}) \gets \textsc{Eval}_i(S^{n+m},b^{n}),
	\]
	where $\textsc{Eval}_i(\cdot)$ can be any heuristic value function that estimates a reasonable future return for the modelling agent $i$.
\end{enumerate}
Note that, for long-horizon planning problems, especially goal-conditioned ones, the planner is supposed to lookahead in depth to ensure non-trivial backup from leaf nodes, which implies that either $(n+m)$ should be large enough or $\textsc{Eval}_i$ should be sufficiently informative.
In fact, the possession of such a heuristic enables one to evaluate any given situation by directly enquiring $\textsc{Eval}_i(S_t,b_t)$, without doing any lookahead search.
\textit{Nevertheless, the lookahead search by the first two layers is highly necessary as it serves as an online policy improvement operator.}
We postpone the proof of this rather intuitive statement to Appendix~\ref{apd:converge}.

By various options of $(n, m, \textsc{Eval}_i)$,
we here reproduce the aforementioned three formulations, and introduce a new forth one, 
\begin{enumerate}[label=\texttt{F\arabic*.}]
	\item If $n = \infty$, then it is equivalent to exactly solving the infinite-horizon POMDP \textbf{as Equation~(\ref{eq:pomdp})}, and there is no need to go to layer~(2) and (3), i.e., $m=0$ and $\textsc{Eval}_i(\cdot) = any$. No replanning will be needed.
	\item If $n=0, m=\infty, \textsc{Eval}_i(\cdot) = any$, then it is equivalent to solving the belief-induced MDP \textbf{as Equation~(\ref{eq:belief_mdp})}. One should note that $\mathcal{S}$ is a finite set. Therefore, layer (2) search can instead be implemented as dynamic programming to handle value backup on repeated states. 
	\item If $n=0$, $m=0$, and
	\[
	\textsc{Eval}_i(S,b) =
	\max_{a_i\in \mathcal{A}_i} \sum_{\pi_{-i}} b(\pi_{-i})\cdot Q_{\mathcal{M}(\pi_{-i})}(S, a_i),
	\]
	then it is equivalent to the QMDP approach \textbf{as Equation~(\ref{eq:qmdp})}.
	\item If $0<n<\infty, m=\infty$ and $\textsc{Eval}_i(\cdot) = any$, then it is equivalent to solving a (finite) $n$-horizon POMDP with terminal states evaluated by respective MDPs. Considerably, this will perform better then \texttt{F2}, but replanning will be needed. 
\end{enumerate}
Note that once $n =\infty$ or $m = \infty$, it does not matter which $\textsc{Eval}_i$ is adopted, as the backup operators in the first two layers are both $\gamma$-contractions which eventually lead to unique fixed points.
The proof is postponed to Appendix~\ref{apd:converge}.
Although one can resort to off-the-shelf MDP solvers, formulations involving optimally solving MDPs will soon become computationally infeasible as the number of agents grows.
The framework then guides one to additional scalable implementations,
\begin{enumerate}[label=\texttt{F\arabic*.}, start=5]
	\item $0<n < \infty$ and $0<m < \infty$, 
	then it is equivalent to solving a finite-horizon POMDP with terminal states evaluated as finite-horizon MDPs with terminal states evaluated by a heuristic.
	Chances are that belief update itself is not that costly. Then one may reallocate all the computational budget of $m$ to $n$, i.e., set $(n,m) \gets (n+m, 0)$, since the complexity of doing lookahead search is the same for the first two layers.	
\end{enumerate}
For both cases in \texttt{F5}, the choice of $\textsc{Eval}_i$ is no longer trivial.
It has to serve as an oracle that returns a heuristic value that is informative enough in a flash,
\textit{which is why we draw one's attention to contextual-RL policies in Section~\ref{sec:solutions} as well as efficient NE solvers later in Section~\ref{sec:benchmark}.}
Once such an oracle is ready, \texttt{F5} is all about improving (or sometimes, repairing) the oracle in real-time by online lookahead search.
In fact, \texttt{F5} leads to a full-width \textsc{ExpectiMax} tree search paradigm as Algorithm~\ref{alg:uts} in Appendix~\ref{apd:planner_detail},
while in the literature there is indeed some alternative framework like \textsc{Max}$^n$ tree search~\cite{stahl1993evolution, samothrakis2011fast}. However, the latter one aims at converging to NEs, which imposes much stronger assumptions on the opponents.



One should notice that in order to perform exact backup and compute $\max\{\cdot\}$, the agent has to exhaustively enumerate all joint actions over belief distributions, which is of complexity $\bigr(|\mathcal{A}_i|\times \prod_{j\neq i}|\Pi_{j}|\times \prod_{j\neq i}|\mathcal{A}_j|\bigl)^{n+m}$ and hence expensive when the problem further scales up. 
The framework is then extended to alleviate the heavy computation, but unavoidably compromises the accuracy, by the following two implementations.


\begin{enumerate}[label=\texttt{F\arabic*.}, start=6]
	\item Sample actions over belief distributions and use sample mean to approximate the exact backup.
	\item Use bandit-based lookahead search to approximate $\max\{\cdot\}$.
\end{enumerate} 
For bandit-based lookahead search to work, the planner has to utilize a node selection function that well balances exploration and exploitation, to eventually minimize accumulated regret~\cite{browne2012survey}.
Besides the most widely used UCT formula~\cite{kocsis2006bandit}, it is also proposed to use a certain prior policy to guide the selection, resulting the pUCT formula~\cite{rosin2011multi, silver2017mastering, schrittwieser2020mastering, schwartz2023bayes, schwartz2023combining}.
%
%
Note a the prior policy is usually not easy to acquire, which leads to the following discussion
on how to make use of NE strategies, to extract such priors as well as value estimates.
%
An ultimate version integrating \texttt{F6} and \texttt{F7} leads to an opponent-aware MCTS  planner as Algorithm~\ref{alg:mcts} in Appendix~\ref{apd:planner_detail}.
It resembles the POMCP implementation~\cite{silver2010monte},
while we avoid using random Monte-Carlo rollouts as value estimates, since they usually cannot backpropagate useful information when the problem is goal-conditioned or the rewards are sparse.
\textit{With the help of efficient NE solvers, we scale the implementation up to suit 50-agent long-horizon planning problems, which to the best of our knowledge is the first opponent-modelling planner of this capability.}

To conclude, this unified framework draws a spectrum of planners.
On top of it, one can predict the performance of a planner devised within this framework, as the deeper (and more deliberately) it searches, the better it performs.
For example,
with the same choice of $\textsc{Eval}_i$, one can easily predict a performance ranking among them:
$\texttt{F1} \succcurlyeq \texttt{F4} \succcurlyeq \texttt{F2} \sim \texttt{F5} \succcurlyeq \texttt{F6} \sim \texttt{F7}$.
One can later see such a trend as expected in the experimental results in Table~\ref{tab:performance}.





\begin{table*}[!ht]
\footnotesize
\caption{Detailed description of the implemented planners under the unified framework.}
\vspace{-2mm}
\centering
\begin{tabular}{@{}c|cccccccc@{}}
\toprule
\textbf{Planners}    & \textbf{n}   & \textbf{m} & $\mathbf\beta$ & \textbf{Collision penalty} & \textbf{\textsc{Eval}$_i$} & \textbf{Backup}     & \textbf{Lookahead} & \textbf{Replanning} \\ \midrule
$A^*$-Agent          &              &            &                &                            &                                             &                     &                    &                          \\
Safe-Agent           &              & 1          &                & $\infty$                   & single-agent $A^*$                          & exact               & full-width         & \checkmark                \\
EnhancedSafe-Agent   &              & 1          & 1 $\to$ 0      & $\infty$                   & single-agent $A^*$                          & exact               & full-width         & \checkmark                \\ \midrule
MDPAgentFixedBelief  & 0            & $\infty$   &                & $< \infty$                 &                                          &    exact           &                    &                          \\
MDPAgentUpdateBelief & 0            & $\infty$   &                & $< \infty$                 &                                          &    exact            &                    & \checkmark                \\ \midrule
RLAgentFixedBelief   & $\infty$     & 0          &                & $< \infty$                 &                                &                     &                    &                          \\
RLAgentUpdateBelief  & $\infty$     & 0          &                & $< \infty$                 &                                &                     &                    &                          \\
UniformTSAgentRL     & $(0,\infty)$ &            &                & $< \infty$                 & contextual-RL                               & exact               & full-width         & \checkmark                \\ \midrule
CBSAgentFixedBelief  & 0            & 0          &                & $< \infty$                 & NE by EECBS                                 &                     &                    &                          \\
CBSAgentUpdateBelief & 0            & 0          &                & $< \infty$                 & NE by EECBS                                 &                     &                    & \checkmark                \\
UniformTSAgentCBS    & $(0,\infty)$ &            &                & $< \infty$                 & NE by EECBS                                 & \{exact, sampling\} & full-width         & \checkmark                \\
MCTSAgentCBSuct      &              & 0          &                & $< \infty$                 & NE by EECBS                                 & sampling            & bandit-based       & \checkmark                \\
MCTSAgentCBSpuct     &              & 0          &                & $< \infty$                 & NE by EECBS                                 & sampling            & bandit-based       & \checkmark                \\ \bottomrule
\end{tabular}
\label{tab:all_planners}
\end{table*}

\section{Experiments}

\label{sec:benchmark}
\subsection{Benchmark}

\begin{table*}[t]
\footnotesize
\caption{Statistics in each scenario. The numbers in the ``lower bounds'' row are means (standard deviations).}
\vspace{-2mm}
\centering
\begin{tabular}{@{}r|cc|cc|cc|cc|cc@{}}
\toprule
                  & \multicolumn{2}{c|}{\textbf{Small2a (8x8 map)}}                          & \multicolumn{2}{c|}{\textbf{Square2a (12x12 map)}}                       & \multicolumn{2}{c|}{\textbf{Square4a (12x12 map)}}                       & \multicolumn{2}{c|}{\textbf{Medium20a (18x18 map)}}                         & \multicolumn{2}{c}{\textbf{Large50a (32x32 map)}}                          \\ \midrule
                  & \begin{tabular}[c]{@{}c@{}}Rational\\ Malicious\end{tabular} & Self       & \begin{tabular}[c]{@{}c@{}}Rational\\ Malicious\end{tabular} & Self       & \begin{tabular}[c]{@{}c@{}}Rational\\ Malicious\end{tabular} & Self       & \begin{tabular}[c]{@{}c@{}}Rational\\ Malicious\end{tabular}  & Self         & \begin{tabular}[c]{@{}c@{}}Rational\\ Malicious\end{tabular}  & Self        \\ \midrule
\#(states)        & \multicolumn{2}{c|}{930}                                                  & \multicolumn{2}{c|}{7310}                                                 & \multicolumn{2}{c|}{$5.47\times 10^7$}                           & \multicolumn{2}{c|}{$6.44\times 10^{46}$}                             & \multicolumn{2}{c}{$4.62\times 10^{145}$}                            \\
\#(joint-actions) & \multicolumn{2}{c|}{25}                                                   & \multicolumn{2}{c|}{25}                                                   & \multicolumn{2}{c|}{625}                                                  & \multicolumn{2}{c|}{$5^{20} \sim 9.54\times 10^{13}$} & \multicolumn{2}{c}{$5^{50} \sim 8.88\times 10^{34}$} \\
\#(op-types each) & \multicolumn{2}{c|}{30}                                                   & \multicolumn{2}{c|}{85}                                                   & \multicolumn{2}{c|}{85}                                                   & \multicolumn{2}{c|}{218}                                                     & \multicolumn{2}{c}{818}                                                     \\ \midrule
Lower bounds      & 4.16(2.14)                                                   & 4.26(1.54) & 6.78(3.34)                                                   & 6.79(2.54) & 6.61(3.35)                                                   & 6.85(1.77) & 11.12(5.56)                                                   & 11.32(1.19)  & 23.2(11.25)                                                   & 23.33(1.54) \\
Upper bounds      & \multicolumn{2}{c|}{32}                                                   & \multicolumn{2}{c|}{48}                                                   & \multicolumn{2}{c|}{48}                                                   & \multicolumn{2}{c|}{144}                                                      & \multicolumn{2}{c}{256}                                                     \\ \bottomrule
\end{tabular}
\label{tab:testbed}
\vspace{-1mm}
\end{table*}

\subsubsection{Multi-Agent Route Planning}
In the domain of \textit{multi-agent route planning} (MARP), a group of agents are placed in a grid world with possible obstacles. Each of them is designated a goal. 
We here name the benchmark as MAPP in order to drag one's attention away from the particular research area called \textit{multi-agent path finding} (MAPF)~\cite{stern2019multi-overview, stern2019multi}. 
In MAPF, this fleet of robots are supposed to find a set of paths to reach their goals from the initial positions without any collision among them.
Usually, the set of paths is computed upfront by an efficient centralized planner~\cite{sharon2015conflict, li2021eecbs} and then robots are forced to strictly execute them.
In other words, all robots are under the control of a central planner.
By contrast, in our MARP setting, only one agent is under our control,
without knowing others' goals and strategies.
The agent will be positively rewarded when she arrives at her goal, and will be penalized if she hits others or gets hit.
The consequences of actions are deterministic and commonly known. The environment state consists of all agents' locations and is completely observable to everyone.
We borrow this domain for the reason that it provides us with problem instances of customizable scales.
Our experiments are conducted on configurations
up to 32x32 maps with 50 agents. 
Compared with contemporary work~\cite{carmel1998explore, albrecht2015game, eck2020scalable},
This benchmark is challenging for two reasons: (1)~it is no longer repeated games, instead, involves a large number of states and joint-actions, as shown in Table~\ref{tab:testbed}; (2)~deliberate long-horizon planning is definitely needed, as the rewards are goal-conditioned and hence sparse.

\subsubsection{Nash Equilibrium}
An NE typically reveals a predictable future outcome. 
In this routing domain, given a set of initial positions and goals,
if a set of collision-free paths is found optimal up to certain metrics, e.g., minimizing the sum of lengths, it can serve as an NE as no agent will unitarily deviate in the sense of finding a shorter path without colliding to any others.
In general, it is an NP-hard problem~\cite{yu2013structure}.
However, efficient planners can solve it very fast in practice, by making good use of constraint propagation, e.g., conflict-based search~\cite{sharon2015conflict}, its bounded sub-optimal variant~\cite{li2021eecbs}, and even some universal planner~\cite{zhu2024computinguniversalplanspartially}.


Despite our argument that NE as a solution concept does not \textit{directly} offer much help,
an agent can still compute one by sampling hypothetical opponents from her belief, and make a smart use of it: (1) either directly execute the NE strategy expecting that the others will do the same, or (2) transform it into the $\textsc{Eval}_i$ heuristic as we mentioned, and therefore, can plug it into our tree search framework to repair it in an online manner.
\textit{A heuristic extracted from NEs is considerably more informative than random Monte Carlo rollouts.}
Detailed procedures will be revealed in Appendix~\ref{apd:exp_setting}.

For other domains, one can always appeal to MARL methods for NEs, e.g., MADDPG~\cite{lowe2017multi}, QMIX~\cite{rashid2020monotonic}, IPPO/MAPPO~\cite{yu2022the}, etc.

\subsubsection{A Special Case: Safe Agents}
Given this domain, we now present one simple-yet-powerful agent which is a myopic special case of the unified tree search framework, called \textit{safe-agents}. 
An action $a_i \in \mathcal{A}_i$ in state $S$ is unsafe if there exists another agent $j$ and an action $a_j \in \mathcal{A}_j$ that will drive agent $i$ and $j$ into collision in the successor state.
For each step, a \textit{safe-agent} first rules out all unsafe actions and takes the best one among the rest. By ``the best action'' here, we mean that the agent simply finds the action through which the shortest distance (ignoring other agents) to her goal can be minimized. 
In fact, one can easily construct a counter-example to get a \textit{safe-agent} stuck at an intermediate position forever.
For example, if one of the opponents stops right in the way of a \textit{safe-agent}'s only shortest path to her goal, then this \textit{safe-agent} would rather stop forever than move to a nearby safe position since that leads to a longer path.
We can therefore make \textit{safe-agents} slightly cleverer by incorporating domain-specific knowledge: no matter how random an opponent is, she will no longer move once she reaches her goal.
This leads to the \textit{enhanced safe-agents}, who will simply treat a nearby agent that has stayed still for a sufficiently long time as an obstacle, and replan a shortest path to get around.

In the language of our unified framework, \textit{safe-agents} are induced by ignoring belief update, doing one-depth full-width lookahead search, using constraints to rule out unsafe actions, and evaluating the leaf nodes by values of the onward shortest paths.
The \textit{enhanced} version is further obtained by descending the belief temperature $\beta$ from 1 to 0 when the aforementioned situations are detected, making the belief distribution a hard-max one. In this sense, these two \textit{safe-agent} variants are both online improved versions of $A^*$ agents (ignoring opponents) by one depth of tree search.

\subsubsection{Setup}
We now instantiate each theoretic module to the corresponding domain-specific implementation.

\textit{Belief Initialization and Update.}
Since the specific goal of each opponent is not revealed to the modelling agent, the modelling agent therefore associates each of the others a \textit{uniform} initial belief over all empty cells.
She also assumes everyone else is a goal-directed agent up to some randomness. More precisely, for each opponent $j$, for each empty cell $g_k$ in the map layout, a plan will be hypothesized as a shortest path tree from all possible positions to $g_k$, denoted as $\pi_j^k$ and $\mathbb{P}[\pi_j^k] = 1 / \#(empty\ cells)$.
 The modelling agent $i$ assumes agent $j$ will play $\pi_j^k$ with probability $(1-\epsilon)\cdot~\mathbb{P}[\pi_j^k]$, and go to a random adjacent position with probability $\epsilon$. 
 As for belief update, we simply use Equation (\ref{eq:bayes}) with $\beta=1$ except for the \textit{enhanced safe-agents}. Intuitively, when $\epsilon \to 0$, the term $\pi_j^k(a_j|s)\cdot~b_t(\pi_j^k)$ will be zero if agent $j$ does not follow the course of shortest-path actions towards $g_k$.
 Therefore, the belief update,
 to some extent, actually makes soft inference about opponents' true goals.

\textit{Solvers.}
We use 
\textit{mdptoolbox}\footnote{https://pymdptoolbox.readthedocs.io/en/latest/api/mdptoolbox.html} as the MDP solver, 
EECBS
\footnote{https://github.com/Jiaoyang-Li/EECBS}~\cite{li2021eecbs} as the constraint-based approximate NE solver, which returns joint-plans whose lengths are no more than a user-specified factor away from the optimum, and PPO~\cite{schulman2017proximal} implemented by \textit{stable-baseline3}\footnote{https://stable-baselines3.readthedocs.io/en/master/}~\cite{stable-baselines3} as the contextual-RL algorithm.


\textit{Hardware.}
The RL experiments are conducted on a Linux machine with one NVIDIA 2080Ti GPU. The rest are done on a Mac mini (M1 CPU, 16GB memory) with multiprocessing.

\begin{table*}[t]
\footnotesize
\caption{Detailed comparison across all planners. The numbers are means (standard deviations). Those numbers in bold fonts are the best ones in each scenario. Slots filled with ``\textbf{/}'' means the corresponding planners are computationally infeasible at that problem scale. We also copy the leftmost three columns of ``\textbf{Small2a}'' and ``\textbf{Square2a}'' to the lower table for better  comparison.} 
\centering
\begin{tabular}{@{}>{\raggedright\arraybackslash}p{1.2cm}>{\raggedright\arraybackslash}p{1cm}|>{\raggedright\arraybackslash}p{1.4cm}>{\raggedright\arraybackslash}p{1.2cm}>{\raggedright\arraybackslash}p{1.8cm}|>{\raggedright\arraybackslash}p{1.5cm}>{\raggedright\arraybackslash}p{1.5cm}>{\raggedright\arraybackslash}p{1.5cm}>{\raggedright\arraybackslash}p{1.5cm}>{\raggedright\arraybackslash}p{2cm}@{}}
\toprule
                   &           & \textbf{Astar} & \textbf{Safe}         & \textbf{EnhancedSafe} & \textbf{MDPFixed}   & \textbf{MDPUpdate}  & \textbf{RLFixed} & \textbf{RLUpdate} & \textbf{UnifTSRL} \\ \midrule
\textbf{Small2a}     & Rational  & 7.25(9.18)     & 7.33(8.90)             & 4.95(3.67)            & 6.62(7.85)          & 4.76(2.91)          & 5.92(6.91)         & 5.74(6.56)          & 5.06(4.51)             \\
\textbf{}          & Malicious & 12.33(13.11)   & 5.18(3.04)            & 5.18(3.04)            & 5.00(2.84)           & \textbf{4.96(2.79)} & 9.59(11.14)        & 9.47(11.05)         & 6.45(6.91)             \\
                   & Self      & 9.19(10.92)    & 9.96(10.89)           & 5.98(5.95)            & 9.35(10.62)         & \textbf{4.96(2.94)} & 6.44(7.54)         & 6.18(7.11)          & 5.56(5.38)             \\ \midrule
\textbf{Square2a}  & Rational  & 9.42(10.82)    & 9.60(10.85)            & 7.14(3.87)            & 8.90(9.49)           & \textbf{/}          & 9.44(10.82)        & 9.31(10.55)         & 8.19(7.84)             \\
\textbf{}          & Malicious & 17.40(18.64)    & 7.75(4.41)            & 7.75(4.41)            & \textbf{7.28(3.58)} & \textbf{/}          & 17.11(18.45)       & 17.10(18.45)         & 10.10(11.14)            \\
                   & Self      & 11.18(13.02)   & 11.60(13.12)           & 7.70(5.83)             & 10.91(12.38)        & \textbf{/}          & 12.08(14.13)       & 11.78(13.76)        & 8.92(9.08)             \\
\bottomrule
\end{tabular}

\footnotesize
\centering
\begin{tabular}{@{}>{\raggedright\arraybackslash}p{1.2cm}>{\raggedright\arraybackslash}p{1cm}|>{\raggedright\arraybackslash}p{1.4cm}>{\raggedright\arraybackslash}p{1.2cm}>{\raggedright\arraybackslash}p{1.8cm}|>{\raggedright\arraybackslash}p{1.5cm}>{\raggedright\arraybackslash}p{1.5cm}>{\raggedright\arraybackslash}p{1.5cm}>{\raggedright\arraybackslash}p{1.5cm}>{\raggedright\arraybackslash}p{2cm}@{}}
\toprule
                   &           & \textbf{Astar} & \textbf{Safe}         & \textbf{EnhancedSafe} & \textbf{CBSFixed} & \textbf{CBSUpdate} & \textbf{UnifTSCBS} & \textbf{MCTSCBSuct} & \textbf{MCTSCBSpuct} \\ \midrule
\textbf{Small2a}     & Rational  & 7.25(9.18)     & 7.33(8.90)             & 4.95(3.67)            & 6.85(8.62)        & 6.09(7.33)         & \textbf{4.74(3.19)}   & 5.48(5.07)          & 5.47(5.24)           \\
\textbf{}          & Malicious & 12.33(13.11)   & 5.18(3.04)            & 5.18(3.04)            & 12.30(13.08)       & 11.81(12.79)       & 5.32(4.08)            & 7.15(7.8)           & 6.93(7.53)           \\
                   & Self      & 9.19(10.92)    & 9.96(10.89)           & 5.98(5.95)            & 8.46(10.22)       & 6.77(8.09)         & 5.20(4.30)              & 5.48(4.21)          & 5.32(3.86)           \\ \midrule
\textbf{Square2a}  & Rational  & 9.42(10.82)    & 9.60(10.85)            & 7.14(3.87)            & 9.48(10.89)       & 8.74(9.41)         & \textbf{7.13(4.14)}   & 8.41(7.00)           & 8.27(6.95)           \\
\textbf{}          & Malicious & 17.40(18.64)    & 7.75(4.41)            & 7.75(4.41)            & 17.28(18.57)      & 16.77(18.22)       & 8.41(7.63)            & 12.77(13.89)        & 11.92(13.00)          \\
                   & Self      & 11.18(13.02)   & 11.60(13.12)           & 7.70(5.83)             & 10.94(12.71)      & 9.10(9.77)          & \textbf{7.05(3.16)}   & 7.98(5.17)          & 7.61(3.96)           \\ \midrule
\textbf{Square4a}  & Rational  & 13.40(15.92)    & 13.53(15.41)          & 8.44(6.60)             & 13.04(15.56)      & 11.50(13.72)        & \textbf{8.26(7.42)}   & 12.51(12.00)         & 11.71(11.47)          \\
\textbf{}          & Malicious & 20.89(20.43)   & $\textbf{11.43(10.68)}$ & 11.76(11.33)          & 20.06(20.03)      & 19.71(19.87)       & 13.65(15.29)          & 19.58(17.88)        & 18.54(17.71)          \\
                   & Self      & 27.42(20.78)   & 27.45(19.30)           & 11.59(10.50)           & 25.57(20.69)      & 19.14(18.92)       & \textbf{9.56(9.55)}   & 15.93(13.99)        & 14.04(12.80)           \\ \midrule
\textbf{Medium20a} & Rational  & 88.45(66.72)   & 73.92(62.80)           & \textbf{35.52(40.65)} & 86.14(67.10)       & 71.56(66.62)       & \textbf{/}            & 59.15(52.10)         & 56.04(52.59)          \\
\textbf{}          & Malicious & 96.27(64.83)   & \textbf{40.26(48.63)} & 40.62(49.03)          & 91.89(65.88)      & 90.82(65.98)       & \textbf{/}            & 64.97(56.53)        & 64.02(57.07)          \\
                   & Self      & 144.00(0.00)     & \textbf{49.27(28.73)} & 67.38(51.92)          & 144.00(0.00)        & 144.00(0.00)         & \textbf{/}            & /                   & /                     \\ \midrule
\textbf{Large50a}  & Rational  & 182.01(111.03) & 111.80(95.58)          & \textbf{74.60(78.27)}  & 187.32(108.28)    & 163.06(114.24)     & \textbf{/}            & 120.48(91.51)      & 119.26(96.18) \\
\textbf{}          & Malicious & 193.73(105.03) & \textbf{79.42(92.44)} & 79.84(92.48)          & 182.22(109.31)    & 188.22(106.83)      & \textbf{/}            & 145.88(106.18)      & 132.87(104.97)        \\
                   & Self      & 256.00(0.00)     & \textbf{76.65(5.30)}   & 209.15(82.39)         & /                 & /                  & \textbf{/}            & /                   & /                     \\ \bottomrule
\end{tabular}
\label{tab:performance}
\end{table*}

\subsection{Evaluation}

\label{sec:empirical}

In this section, we conduct a comprehensive study over \textbf{13} planners resulted from our unified framework. 
Table~\ref{tab:all_planners} shows how each implemented planner is induced by the unified framework.
For each class of ``FixedBelief'' and its opposite ``UpdateBelief'', we mean whether the planner will update the belief after observing opponents' actions.
For example, ``MDPAgentFixedBelief'' means the planner only solves the MDP induced by the initial belief and does no replanning, while ``MDPAgentFixedBelief'' means the planner will update the belief at each step and solve a new MDP. ``RL'' agents  solve the underlying POMDP and thus no replanning is needed.

As we are investigating the setting where no prior information about opponents' goals or strategies are revealed to the modelling agents, the golden standard is to run each proposed planner against all possible hypothetical opponents and their combinations, which is clearly impossible.
To get around, we prepare three representative sets of opponents against which every planner will be tested for an average performance.

%
%

\begin{enumerate}
	\item ``Rational types'': a mixture of opponents that are, to some extent, goal-directed, mainly including three types:
	\textit{ShortestPathAgent}, \textit{RandomAgent(p)} and \textit{SafeAgent}. A \textit{ShortestPathAgent} is an agent that goes towards her own goal ignoring other agents. A \textit{RandomAgent(p)} is a modified version who does random actions with probability $p$, while replans a shortest path to her goal with probability $(1-p)$. 
	
	\item ``Malicious types'': a mixture of opponents that sometimes intend to do harm to the modelling agent, what we call \textit{ChasingAgent(p)}. A \textit{ChasingAgent(p)} is an agent that plans a shortest path to the modelling agent with probability $p$, while replans a shortest path to her own goal otherwise. 
	
	\item ``Self-play'': all agents run the same planner as every one is an autonomous agent modelling others simultaneously. 
\end{enumerate}
Experiments against ``rational'' opponents aim at simulating the situations where belief modelling more or less aligns with the underlying truth, while experiments against ``malicious'' ones examine whether belief-dependent planning will be broken down as  belief modelling is significantly attacked by the chasing behavior.
Empirical results under ``self-play'' scenarios try to figure out the gap between the outcome and the potential (but unreachable) ex-post NE (of the SG realized by specific goals).

Table~\ref{tab:testbed} presents some statistics for each testing scenario. For example, the one named \textbf{Medium20a} refers to a configuration with 20 agents initially randomly spawned on maps of size 18x18 with a few obstacles.
Each opponent will be associated with 218 hypothetical policies by the modelling agent.
The scenario contains around $6.44\times 10^{46}$ states and $9.54\times 10^{13}$ joint-actions under each state.
The lower bounds for ``rational/malicious'' cases are computed as the average lengths of single-agent shortest paths (ignoring all collisions) for the modelling agent, while the lower bounds for ``self-play'' cases are computed as the average lengths per agent of the multi-agent optimal joint paths (avoiding any collision). By definition, these lower bounds are impossible to reach and hence only for reference.
The upper bounds are the maximum number of steps allowed for route planning. Also, once any collision happens, the length of the route will be set to this upper bound.


We attach the overall comparison among those 13 planners in Table \ref{tab:performance}.
The numbers are the average path lengths penalized by collisions.
Additional details about raw path lengths and collision ratios are attached in Appendix~\ref{apd:exp_setting}.
We intend to present Table~\ref{tab:performance} as a practical handbook for one to select a suitable planner, or devise their own ones in a similar way, given any problem in a certain scale, and therefore, answer the research question in Section~\ref{sec:intro}.

One can imagine Table~\ref{tab:performance} as an upper-triangular matrix, where each column corresponds to a particular planner, and each row shows the 
performances of all planners in a particular testing scenario with the best numbers in bold fonts. 
\textit{One should compare the numbers in each row, while their absolute values may not be meaningful as once a collision happens the penalty will be huge.}

\textbf{As an overview}, with the testing scenario scaling up, one can notice that the most capable planners gradually become computationally infeasible as they do deliberate computation, starting from ``{MDPUpdate}'', then ``{MDPFixed}'', finally ``{UnifTSCBS}''.
Understandably, when exhaustive enumeration of joint-actions becomes infeasible, the only options left are {MCTS} planners and \textit{safe-agent} variants. 
Despite that MCTS is in principle an anytime algorithm, if the increase of computational budget is not able to catch up with the inherent exponential complexity, it will definitely lead to compromised performance. 

\textbf{Look closer}, Table~\ref{tab:performance} indeed echoes the potential of tree search that serves as an online policy improvement operator, even while modelling a large number of opponents, e.g., by comparing vanilla RL/CBS agents against their improved tree search versions.
Additionally, MCTS planners with pUCT formula guided by NE policy priors beat those with vanilla UCT formula, but the gap is not significantly large, mainly because policy priors out of NE strategies presume others will do the same, which indeed causes some inconsistency.
One should also notice that CBS agents are never the best ones, which supports our argument that NE offers little help if one aims at computing effective strategies against priorly unknown opponents.
Contextual-RL here rather serves as a proof-of-concept example. Although online tree search further improves the learned contextual-policy, it is still challenging to invent a learning diagram that can be immediately applied to any configuration once trained.

\textbf{Unexpectedly,}  \textit{safe-agents} and its \textit{enhanced} versions perform surprisingly well, especially in cases of larger scales and against malicious opponents.
For large-scale instances, due to the limited computational budget, tree search agents are not capable of doing deliberate lookahead search to significantly outperform \textit{safe-agents}, while the latter ones easily beat all the rest in terms of fast replanning due to their simple algorithmic structure.
For situations against malicious opponents where belief modelling does not match the underlying truth, \textit{safe-agents} conservatively rule out all unsafe action and therefore guarantee a lower chance of collision. 
Alternatively, one highly interpretable angle is to see \textit{safe-agents} as \textsc{Minimax} tree search that offers a fairly good worst-case performance, especially when your opponents turn out to be harmful ones.

\textbf{For case-by-case recommendations}, Table~\ref{tab:performance} roughly shows the computational limit of each planner. One should always identify the scale of the problem at first and find the most suitable planner accordingly. For example, given a problem of $10^6$ states with three agents, we would definitely vote for two-depth full-width tree search with an NE oracle against MCTS. However, in a huge scale where value estimate oracles like RL/NE become out of reach, the only choice left is to devise rule-based planners such as \textit{safe-agents}, which at least offer certain conservative guarantees.

\textit{For other planners that may not be suitable for a multi-run evaluation, we attach a few case studies about them in Appendix~\ref{apd:case_study}.}

\section{Conclusion}
\label{sec:conclusion}

To conclude, we formalize the problem of controlling one single agent against multiple opponents that are priorly unknown. A spectrum of formulations are drawn, all of which are further unified under a tree search perspective. We underpin the investigation by offering a challenging benchmark, namely multi-agent route planning. We also show by this general framework how to customize domain-specific planners such as \textit{safe-agents}. To offer a practical handbook of proper selection of the induced planners, we have empirically tested each of them against three representative groups of opponents. One interesting observation from our experiments is that those conservative and myopic \textit{safe-agents} perform sufficiently well in most cases, especially when belief modelling does not match the underlying truth or deliberate replanning is computationally infeasible.
As for \textbf{future work}, we point out a few valuable directions: 
\begin{enumerate}
	\item How to design a formal language to describe domain knowledge that can be embedded into those general planners \cite{gao2020embedding}?
	\item How to model dependencies among agents which may potentially decompose the complexity of modelling joint-transitions and joint-utilities \cite{ma2024efficient}, and what if such dependencies are dynamic, i.e., agents may come and go?
	\item What if there are infinitely many opponents, is it justifiable to model the problem as a mean-field game \cite{lasry2007mean}?
	\item
	Other applicable domains like mechanism design~\cite{tang2017reinforcement} and even negotiation \cite{kraus1997negotiation, beer1999negotiation, jennings2001automated},
	 which also turn out to involve strategic behaviors that explore and exploit opponents during repeated interaction. 
\end{enumerate}




\begin{acks}
This work was supported in part by a generous research grant from Xiaoi Robot Technology Limited.
We thank Yuxin Pan, Chenglin Wang, Haozhe Wang, and Yangfan Wu for their valuable feedback during the revision of this paper. We also appreciate the anonymous reviewers for their insightful comments.
\end{acks}



\bibliographystyle{ACM-Reference-Format} 
\bibliography{sample}


\begin{thebibliography}{64}


\ifx \showCODEN    \undefined \def \showCODEN     #1{\unskip}     \fi
\ifx \showDOI      \undefined \def \showDOI       #1{#1}\fi
\ifx \showISBNx    \undefined \def \showISBNx     #1{\unskip}     \fi
\ifx \showISBNxiii \undefined \def \showISBNxiii  #1{\unskip}     \fi
\ifx \showISSN     \undefined \def \showISSN      #1{\unskip}     \fi
\ifx \showLCCN     \undefined \def \showLCCN      #1{\unskip}     \fi
\ifx \shownote     \undefined \def \shownote      #1{#1}          \fi
\ifx \showarticletitle \undefined \def \showarticletitle #1{#1}   \fi
\ifx \showURL      \undefined \def \showURL       {\relax}        \fi
\providecommand\bibfield[2]{#2}
\providecommand\bibinfo[2]{#2}
\providecommand\natexlab[1]{#1}
\providecommand\showeprint[2][]{arXiv:#2}

\bibitem[\protect\citeauthoryear{Albrecht, Crandall, and Ramamoorthy}{Albrecht
  et~al\mbox{.}}{2016}]%
        {albrecht2016belief}
\bibfield{author}{\bibinfo{person}{Stefano~V Albrecht},
  \bibinfo{person}{Jacob~W Crandall}, {and} \bibinfo{person}{Subramanian
  Ramamoorthy}.} \bibinfo{year}{2016}\natexlab{}.
\newblock \showarticletitle{Belief and truth in hypothesised behaviours}.
\newblock \bibinfo{journal}{\emph{Artificial Intelligence}}
  \bibinfo{volume}{235} (\bibinfo{year}{2016}), \bibinfo{pages}{63--94}.
\newblock


\bibitem[\protect\citeauthoryear{Albrecht and Ramamoorthy}{Albrecht and
  Ramamoorthy}{2015}]%
        {albrecht2015game}
\bibfield{author}{\bibinfo{person}{Stefano~V Albrecht} {and}
  \bibinfo{person}{Subramanian Ramamoorthy}.} \bibinfo{year}{2015}\natexlab{}.
\newblock \showarticletitle{A game-theoretic model and best-response learning
  method for ad hoc coordination in multiagent systems}.
\newblock \bibinfo{journal}{\emph{arXiv preprint arXiv:1506.01170}}
  (\bibinfo{year}{2015}).
\newblock


\bibitem[\protect\citeauthoryear{Albrecht and Ramamoorthy}{Albrecht and
  Ramamoorthy}{2019}]%
        {albrecht2019convergence}
\bibfield{author}{\bibinfo{person}{Stefano~V Albrecht} {and}
  \bibinfo{person}{Subramanian Ramamoorthy}.} \bibinfo{year}{2019}\natexlab{}.
\newblock \showarticletitle{On convergence and optimality of best-response
  learning with policy types in multiagent systems}.
\newblock \bibinfo{journal}{\emph{arXiv preprint arXiv:1907.06995}}
  (\bibinfo{year}{2019}).
\newblock


\bibitem[\protect\citeauthoryear{Albrecht and Stone}{Albrecht and
  Stone}{2018}]%
        {albrecht2018autonomous}
\bibfield{author}{\bibinfo{person}{Stefano~V Albrecht} {and}
  \bibinfo{person}{Peter Stone}.} \bibinfo{year}{2018}\natexlab{}.
\newblock \showarticletitle{Autonomous agents modelling other agents: A
  comprehensive survey and open problems}.
\newblock \bibinfo{journal}{\emph{Artificial Intelligence}}
  \bibinfo{volume}{258} (\bibinfo{year}{2018}), \bibinfo{pages}{66--95}.
\newblock


\bibitem[\protect\citeauthoryear{Antonoglou, Schrittwieser, Ozair, Hubert, and
  Silver}{Antonoglou et~al\mbox{.}}{2022}]%
        {antonoglou2022planning}
\bibfield{author}{\bibinfo{person}{Ioannis Antonoglou}, \bibinfo{person}{Julian
  Schrittwieser}, \bibinfo{person}{Sherjil Ozair}, \bibinfo{person}{Thomas~K
  Hubert}, {and} \bibinfo{person}{David Silver}.}
  \bibinfo{year}{2022}\natexlab{}.
\newblock \showarticletitle{Planning in Stochastic Environments with a Learned
  Model}. In \bibinfo{booktitle}{\emph{International Conference on Learning
  Representations}}.
\newblock
\urldef\tempurl%
\url{https://openreview.net/forum?id=X6D9bAHhBQ1}
\showURL{%
\tempurl}


\bibitem[\protect\citeauthoryear{Banerjee and Sen}{Banerjee and Sen}{2007}]%
        {banerjee2007reaching}
\bibfield{author}{\bibinfo{person}{Dipyaman Banerjee} {and}
  \bibinfo{person}{Sandip Sen}.} \bibinfo{year}{2007}\natexlab{}.
\newblock \showarticletitle{Reaching pareto-optimality in prisoner's dilemma
  using conditional joint action learning}.
\newblock \bibinfo{journal}{\emph{Autonomous Agents and Multi-Agent Systems}}
  \bibinfo{volume}{15} (\bibinfo{year}{2007}), \bibinfo{pages}{91--108}.
\newblock


\bibitem[\protect\citeauthoryear{Beck, Vuorio, Liu, Xiong, Zintgraf, Finn, and
  Whiteson}{Beck et~al\mbox{.}}{2023}]%
        {beck2023survey}
\bibfield{author}{\bibinfo{person}{Jacob Beck}, \bibinfo{person}{Risto Vuorio},
  \bibinfo{person}{Evan~Zheran Liu}, \bibinfo{person}{Zheng Xiong},
  \bibinfo{person}{Luisa Zintgraf}, \bibinfo{person}{Chelsea Finn}, {and}
  \bibinfo{person}{Shimon Whiteson}.} \bibinfo{year}{2023}\natexlab{}.
\newblock \showarticletitle{A survey of meta-reinforcement learning}.
\newblock \bibinfo{journal}{\emph{arXiv preprint arXiv:2301.08028}}
  (\bibinfo{year}{2023}).
\newblock


\bibitem[\protect\citeauthoryear{Beer, d'Inverno, Luck, Jennings, Preist, and
  Schroeder}{Beer et~al\mbox{.}}{1999}]%
        {beer1999negotiation}
\bibfield{author}{\bibinfo{person}{Martin Beer}, \bibinfo{person}{Mark
  d'Inverno}, \bibinfo{person}{Michael Luck}, \bibinfo{person}{Nick Jennings},
  \bibinfo{person}{Chris Preist}, {and} \bibinfo{person}{Michael Schroeder}.}
  \bibinfo{year}{1999}\natexlab{}.
\newblock \showarticletitle{Negotiation in multi-agent systems}.
\newblock \bibinfo{journal}{\emph{The Knowledge Engineering Review}}
  \bibinfo{volume}{14}, \bibinfo{number}{3} (\bibinfo{year}{1999}),
  \bibinfo{pages}{285--289}.
\newblock


\bibitem[\protect\citeauthoryear{Benjamins, Eimer, Schubert, Biedenkapp,
  Rosenhahn, Hutter, and Lindauer}{Benjamins et~al\mbox{.}}{2021}]%
        {benjamins2021carl}
\bibfield{author}{\bibinfo{person}{Carolin Benjamins}, \bibinfo{person}{Theresa
  Eimer}, \bibinfo{person}{Frederik Schubert}, \bibinfo{person}{Andr{\'e}
  Biedenkapp}, \bibinfo{person}{Bodo Rosenhahn}, \bibinfo{person}{Frank
  Hutter}, {and} \bibinfo{person}{Marius Lindauer}.}
  \bibinfo{year}{2021}\natexlab{}.
\newblock \showarticletitle{Carl: A benchmark for contextual and adaptive
  reinforcement learning}.
\newblock \bibinfo{journal}{\emph{arXiv preprint arXiv:2110.02102}}
  (\bibinfo{year}{2021}).
\newblock


\bibitem[\protect\citeauthoryear{Boutilier}{Boutilier}{1996}]%
        {boutilier1996planning}
\bibfield{author}{\bibinfo{person}{Craig Boutilier}.}
  \bibinfo{year}{1996}\natexlab{}.
\newblock \showarticletitle{Planning, learning and coordination in multiagent
  decision processes}. In \bibinfo{booktitle}{\emph{TARK}},
  Vol.~\bibinfo{volume}{96}. Citeseer, \bibinfo{pages}{195--210}.
\newblock


\bibitem[\protect\citeauthoryear{Brailsford, Potts, and Smith}{Brailsford
  et~al\mbox{.}}{1999}]%
        {brailsford1999constraint}
\bibfield{author}{\bibinfo{person}{Sally~C Brailsford},
  \bibinfo{person}{Chris~N Potts}, {and} \bibinfo{person}{Barbara~M Smith}.}
  \bibinfo{year}{1999}\natexlab{}.
\newblock \showarticletitle{Constraint satisfaction problems: Algorithms and
  applications}.
\newblock \bibinfo{journal}{\emph{European journal of operational research}}
  \bibinfo{volume}{119}, \bibinfo{number}{3} (\bibinfo{year}{1999}),
  \bibinfo{pages}{557--581}.
\newblock


\bibitem[\protect\citeauthoryear{Browne, Powley, Whitehouse, Lucas, Cowling,
  Rohlfshagen, Tavener, Perez, Samothrakis, and Colton}{Browne
  et~al\mbox{.}}{2012}]%
        {browne2012survey}
\bibfield{author}{\bibinfo{person}{Cameron~B Browne}, \bibinfo{person}{Edward
  Powley}, \bibinfo{person}{Daniel Whitehouse}, \bibinfo{person}{Simon~M
  Lucas}, \bibinfo{person}{Peter~I Cowling}, \bibinfo{person}{Philipp
  Rohlfshagen}, \bibinfo{person}{Stephen Tavener}, \bibinfo{person}{Diego
  Perez}, \bibinfo{person}{Spyridon Samothrakis}, {and} \bibinfo{person}{Simon
  Colton}.} \bibinfo{year}{2012}\natexlab{}.
\newblock \showarticletitle{A survey of monte carlo tree search methods}.
\newblock \bibinfo{journal}{\emph{IEEE Transactions on Computational
  Intelligence and AI in games}} \bibinfo{volume}{4}, \bibinfo{number}{1}
  (\bibinfo{year}{2012}), \bibinfo{pages}{1--43}.
\newblock


\bibitem[\protect\citeauthoryear{Carmel and Markovitch}{Carmel and
  Markovitch}{1998}]%
        {carmel1998explore}
\bibfield{author}{\bibinfo{person}{David Carmel} {and} \bibinfo{person}{Shaul
  Markovitch}.} \bibinfo{year}{1998}\natexlab{}.
\newblock \showarticletitle{How to explore your opponent's strategy (almost)
  optimally}. In \bibinfo{booktitle}{\emph{Proceedings International Conference
  on Multi Agent Systems (Cat. No. 98EX160)}}. IEEE, \bibinfo{pages}{64--71}.
\newblock


\bibitem[\protect\citeauthoryear{Carmel and Markovitch}{Carmel and
  Markovitch}{1999}]%
        {carmel1999exploration}
\bibfield{author}{\bibinfo{person}{David Carmel} {and} \bibinfo{person}{Shaul
  Markovitch}.} \bibinfo{year}{1999}\natexlab{}.
\newblock \showarticletitle{Exploration strategies for model-based learning in
  multi-agent systems: Exploration strategies}.
\newblock \bibinfo{journal}{\emph{Autonomous Agents and Multi-agent systems}}
  \bibinfo{volume}{2} (\bibinfo{year}{1999}), \bibinfo{pages}{141--172}.
\newblock


\bibitem[\protect\citeauthoryear{Chen, Wu, Chitta, Jaeger, Geiger, and Li}{Chen
  et~al\mbox{.}}{2024}]%
        {chen2024end}
\bibfield{author}{\bibinfo{person}{Li Chen}, \bibinfo{person}{Penghao Wu},
  \bibinfo{person}{Kashyap Chitta}, \bibinfo{person}{Bernhard Jaeger},
  \bibinfo{person}{Andreas Geiger}, {and} \bibinfo{person}{Hongyang Li}.}
  \bibinfo{year}{2024}\natexlab{}.
\newblock \showarticletitle{End-to-end autonomous driving: Challenges and
  frontiers}.
\newblock \bibinfo{journal}{\emph{IEEE Transactions on Pattern Analysis and
  Machine Intelligence}} (\bibinfo{year}{2024}).
\newblock


\bibitem[\protect\citeauthoryear{Claus and Boutilier}{Claus and
  Boutilier}{1998}]%
        {claus1998dynamics}
\bibfield{author}{\bibinfo{person}{Caroline Claus} {and} \bibinfo{person}{Craig
  Boutilier}.} \bibinfo{year}{1998}\natexlab{}.
\newblock \showarticletitle{The dynamics of reinforcement learning in
  cooperative multiagent systems}.
\newblock \bibinfo{journal}{\emph{AAAI/IAAI}} \bibinfo{volume}{1998},
  \bibinfo{number}{746-752} (\bibinfo{year}{1998}), \bibinfo{pages}{2}.
\newblock


\bibitem[\protect\citeauthoryear{Danihelka, Guez, Schrittwieser, and
  Silver}{Danihelka et~al\mbox{.}}{2022}]%
        {danihelka2022policy}
\bibfield{author}{\bibinfo{person}{Ivo Danihelka}, \bibinfo{person}{Arthur
  Guez}, \bibinfo{person}{Julian Schrittwieser}, {and} \bibinfo{person}{David
  Silver}.} \bibinfo{year}{2022}\natexlab{}.
\newblock \showarticletitle{Policy improvement by planning with Gumbel}. In
  \bibinfo{booktitle}{\emph{International Conference on Learning
  Representations}}.
\newblock
\urldef\tempurl%
\url{https://openreview.net/forum?id=bERaNdoegnO}
\showURL{%
\tempurl}


\bibitem[\protect\citeauthoryear{Eck, Shah, Doshi, and Soh}{Eck
  et~al\mbox{.}}{2020}]%
        {eck2020scalable}
\bibfield{author}{\bibinfo{person}{Adam Eck}, \bibinfo{person}{Maulik Shah},
  \bibinfo{person}{Prashant Doshi}, {and} \bibinfo{person}{Leen-Kiat Soh}.}
  \bibinfo{year}{2020}\natexlab{}.
\newblock \showarticletitle{Scalable decision-theoretic planning in open and
  typed multiagent systems}. In \bibinfo{booktitle}{\emph{Proceedings of the
  AAAI Conference on Artificial Intelligence}}, Vol.~\bibinfo{volume}{34}.
  \bibinfo{pages}{7127--7134}.
\newblock


\bibitem[\protect\citeauthoryear{Foerster, Farquhar, Afouras, Nardelli, and
  Whiteson}{Foerster et~al\mbox{.}}{2018}]%
        {foerster2018counterfactual}
\bibfield{author}{\bibinfo{person}{Jakob Foerster}, \bibinfo{person}{Gregory
  Farquhar}, \bibinfo{person}{Triantafyllos Afouras}, \bibinfo{person}{Nantas
  Nardelli}, {and} \bibinfo{person}{Shimon Whiteson}.}
  \bibinfo{year}{2018}\natexlab{}.
\newblock \showarticletitle{Counterfactual multi-agent policy gradients}. In
  \bibinfo{booktitle}{\emph{Proceedings of the AAAI conference on artificial
  intelligence}}, Vol.~\bibinfo{volume}{32}.
\newblock


\bibitem[\protect\citeauthoryear{Fu, Yu, Xu, Yang, and Wu}{Fu
  et~al\mbox{.}}{2022}]%
        {pmlr-v162-fu22d}
\bibfield{author}{\bibinfo{person}{Wei Fu}, \bibinfo{person}{Chao Yu},
  \bibinfo{person}{Zelai Xu}, \bibinfo{person}{Jiaqi Yang}, {and}
  \bibinfo{person}{Yi Wu}.} \bibinfo{year}{2022}\natexlab{}.
\newblock \showarticletitle{Revisiting Some Common Practices in Cooperative
  Multi-Agent Reinforcement Learning}. In \bibinfo{booktitle}{\emph{Proceedings
  of the 39th International Conference on Machine Learning}}
  \emph{(\bibinfo{series}{Proceedings of Machine Learning Research},
  Vol.~\bibinfo{volume}{162})}, \bibfield{editor}{\bibinfo{person}{Kamalika
  Chaudhuri}, \bibinfo{person}{Stefanie Jegelka}, \bibinfo{person}{Le~Song},
  \bibinfo{person}{Csaba Szepesvari}, \bibinfo{person}{Gang Niu}, {and}
  \bibinfo{person}{Sivan Sabato}} (Eds.). \bibinfo{publisher}{PMLR},
  \bibinfo{pages}{6863--6877}.
\newblock
\urldef\tempurl%
\url{https://proceedings.mlr.press/v162/fu22d.html}
\showURL{%
\tempurl}


\bibitem[\protect\citeauthoryear{Gao, Lin, Zhou, Zhang, Wu, and Zhang}{Gao
  et~al\mbox{.}}{2020}]%
        {gao2020embedding}
\bibfield{author}{\bibinfo{person}{Zihang Gao}, \bibinfo{person}{Fangzhen Lin},
  \bibinfo{person}{Yi Zhou}, \bibinfo{person}{Hao Zhang},
  \bibinfo{person}{Kaishun Wu}, {and} \bibinfo{person}{Haodi Zhang}.}
  \bibinfo{year}{2020}\natexlab{}.
\newblock \showarticletitle{Embedding high-level knowledge into dqns to learn
  faster and more safely}. In \bibinfo{booktitle}{\emph{Proceedings of the AAAI
  Conference on Artificial Intelligence}}, Vol.~\bibinfo{volume}{34}.
  \bibinfo{pages}{13608--13609}.
\newblock


\bibitem[\protect\citeauthoryear{Hallak, Di~Castro, and Mannor}{Hallak
  et~al\mbox{.}}{2015}]%
        {hallak2015contextual}
\bibfield{author}{\bibinfo{person}{Assaf Hallak}, \bibinfo{person}{Dotan
  Di~Castro}, {and} \bibinfo{person}{Shie Mannor}.}
  \bibinfo{year}{2015}\natexlab{}.
\newblock \showarticletitle{Contextual markov decision processes}.
\newblock \bibinfo{journal}{\emph{arXiv preprint arXiv:1502.02259}}
  (\bibinfo{year}{2015}).
\newblock


\bibitem[\protect\citeauthoryear{Harsanyi}{Harsanyi}{1967}]%
        {harsanyi1967games}
\bibfield{author}{\bibinfo{person}{John~C Harsanyi}.}
  \bibinfo{year}{1967}\natexlab{}.
\newblock \showarticletitle{Games with incomplete information played by
  ``Bayesian'' players, I--III Part I. The basic model}.
\newblock \bibinfo{journal}{\emph{Management science}} \bibinfo{volume}{14},
  \bibinfo{number}{3} (\bibinfo{year}{1967}), \bibinfo{pages}{159--182}.
\newblock


\bibitem[\protect\citeauthoryear{Jennings, Faratin, Lomuscio, Parsons, Sierra,
  and Wooldridge}{Jennings et~al\mbox{.}}{2001}]%
        {jennings2001automated}
\bibfield{author}{\bibinfo{person}{Nicholas~R Jennings},
  \bibinfo{person}{Peyman Faratin}, \bibinfo{person}{Alessio~R Lomuscio},
  \bibinfo{person}{Simon Parsons}, \bibinfo{person}{Carles Sierra}, {and}
  \bibinfo{person}{Michael Wooldridge}.} \bibinfo{year}{2001}\natexlab{}.
\newblock \showarticletitle{Automated negotiation: prospects, methods and
  challenges}.
\newblock \bibinfo{journal}{\emph{International Journal of Group Decision and
  Negotiation}} \bibinfo{volume}{10}, \bibinfo{number}{2}
  (\bibinfo{year}{2001}), \bibinfo{pages}{199--215}.
\newblock


\bibitem[\protect\citeauthoryear{Kaelbling, Littman, and Cassandra}{Kaelbling
  et~al\mbox{.}}{1998}]%
        {kaelbling1998planning}
\bibfield{author}{\bibinfo{person}{Leslie~Pack Kaelbling},
  \bibinfo{person}{Michael~L Littman}, {and} \bibinfo{person}{Anthony~R
  Cassandra}.} \bibinfo{year}{1998}\natexlab{}.
\newblock \showarticletitle{Planning and acting in partially observable
  stochastic domains}.
\newblock \bibinfo{journal}{\emph{Artificial intelligence}}
  \bibinfo{volume}{101}, \bibinfo{number}{1-2} (\bibinfo{year}{1998}),
  \bibinfo{pages}{99--134}.
\newblock


\bibitem[\protect\citeauthoryear{Kalai and Lehrer}{Kalai and Lehrer}{1993}]%
        {kalai1993rational}
\bibfield{author}{\bibinfo{person}{Ehud Kalai} {and} \bibinfo{person}{Ehud
  Lehrer}.} \bibinfo{year}{1993}\natexlab{}.
\newblock \showarticletitle{Rational learning leads to Nash equilibrium}.
\newblock \bibinfo{journal}{\emph{Econometrica: Journal of the Econometric
  Society}} (\bibinfo{year}{1993}), \bibinfo{pages}{1019--1045}.
\newblock


\bibitem[\protect\citeauthoryear{Kocsis and Szepesv{\'a}ri}{Kocsis and
  Szepesv{\'a}ri}{2006}]%
        {kocsis2006bandit}
\bibfield{author}{\bibinfo{person}{Levente Kocsis} {and} \bibinfo{person}{Csaba
  Szepesv{\'a}ri}.} \bibinfo{year}{2006}\natexlab{}.
\newblock \showarticletitle{Bandit based monte-carlo planning}. In
  \bibinfo{booktitle}{\emph{European conference on machine learning}}.
  Springer, \bibinfo{pages}{282--293}.
\newblock


\bibitem[\protect\citeauthoryear{Kraus}{Kraus}{1997}]%
        {kraus1997negotiation}
\bibfield{author}{\bibinfo{person}{Sarit Kraus}.}
  \bibinfo{year}{1997}\natexlab{}.
\newblock \showarticletitle{Negotiation and cooperation in multi-agent
  environments}.
\newblock \bibinfo{journal}{\emph{Artificial intelligence}}
  \bibinfo{volume}{94}, \bibinfo{number}{1-2} (\bibinfo{year}{1997}),
  \bibinfo{pages}{79--97}.
\newblock


\bibitem[\protect\citeauthoryear{Lasry and Lions}{Lasry and Lions}{2007}]%
        {lasry2007mean}
\bibfield{author}{\bibinfo{person}{Jean-Michel Lasry} {and}
  \bibinfo{person}{Pierre-Louis Lions}.} \bibinfo{year}{2007}\natexlab{}.
\newblock \showarticletitle{Mean field games}.
\newblock \bibinfo{journal}{\emph{Japanese journal of mathematics}}
  \bibinfo{volume}{2}, \bibinfo{number}{1} (\bibinfo{year}{2007}),
  \bibinfo{pages}{229--260}.
\newblock


\bibitem[\protect\citeauthoryear{Li, Ruml, and Koenig}{Li
  et~al\mbox{.}}{2021}]%
        {li2021eecbs}
\bibfield{author}{\bibinfo{person}{Jiaoyang Li}, \bibinfo{person}{Wheeler
  Ruml}, {and} \bibinfo{person}{Sven Koenig}.} \bibinfo{year}{2021}\natexlab{}.
\newblock \showarticletitle{Eecbs: A bounded-suboptimal search for multi-agent
  path finding}. In \bibinfo{booktitle}{\emph{Proceedings of the AAAI
  conference on artificial intelligence}}, Vol.~\bibinfo{volume}{35}.
  \bibinfo{pages}{12353--12362}.
\newblock


\bibitem[\protect\citeauthoryear{Littman, Cassandra, and Kaelbling}{Littman
  et~al\mbox{.}}{1995}]%
        {littman1995learning}
\bibfield{author}{\bibinfo{person}{Michael~L Littman},
  \bibinfo{person}{Anthony~R Cassandra}, {and} \bibinfo{person}{Leslie~Pack
  Kaelbling}.} \bibinfo{year}{1995}\natexlab{}.
\newblock \showarticletitle{Learning policies for partially observable
  environments: Scaling up}.
\newblock In \bibinfo{booktitle}{\emph{Machine Learning Proceedings 1995}}.
  \bibinfo{publisher}{Elsevier}, \bibinfo{pages}{362--370}.
\newblock


\bibitem[\protect\citeauthoryear{Lowe, Wu, Tamar, Harb, Pieter~Abbeel, and
  Mordatch}{Lowe et~al\mbox{.}}{2017}]%
        {lowe2017multi}
\bibfield{author}{\bibinfo{person}{Ryan Lowe}, \bibinfo{person}{Yi~I Wu},
  \bibinfo{person}{Aviv Tamar}, \bibinfo{person}{Jean Harb},
  \bibinfo{person}{OpenAI Pieter~Abbeel}, {and} \bibinfo{person}{Igor
  Mordatch}.} \bibinfo{year}{2017}\natexlab{}.
\newblock \showarticletitle{Multi-agent actor-critic for mixed
  cooperative-competitive environments}.
\newblock \bibinfo{journal}{\emph{Advances in neural information processing
  systems}}  \bibinfo{volume}{30} (\bibinfo{year}{2017}).
\newblock


\bibitem[\protect\citeauthoryear{Ma, Li, Du, Dong, and Yang}{Ma
  et~al\mbox{.}}{2024}]%
        {ma2024efficient}
\bibfield{author}{\bibinfo{person}{Chengdong Ma}, \bibinfo{person}{Aming Li},
  \bibinfo{person}{Yali Du}, \bibinfo{person}{Hao Dong}, {and}
  \bibinfo{person}{Yaodong Yang}.} \bibinfo{year}{2024}\natexlab{}.
\newblock \showarticletitle{Efficient and scalable reinforcement learning for
  large-scale network control}.
\newblock \bibinfo{journal}{\emph{Nature Machine Intelligence}}
  (\bibinfo{year}{2024}), \bibinfo{pages}{1--15}.
\newblock


\bibitem[\protect\citeauthoryear{Madani, Hanks, and Condon}{Madani
  et~al\mbox{.}}{1999}]%
        {madani1999undecidability}
\bibfield{author}{\bibinfo{person}{Omid Madani}, \bibinfo{person}{Steve Hanks},
  {and} \bibinfo{person}{Anne Condon}.} \bibinfo{year}{1999}\natexlab{}.
\newblock \showarticletitle{On the undecidability of probabilistic planning and
  infinite-horizon partially observable Markov decision problems}. In
  \bibinfo{booktitle}{\emph{Proceedings of the sixteenth national conference on
  Artificial intelligence and the eleventh Innovative applications of
  artificial intelligence conference innovative applications of artificial
  intelligence}}. \bibinfo{pages}{541--548}.
\newblock


\bibitem[\protect\citeauthoryear{Papadimitriou and Tsitsiklis}{Papadimitriou
  and Tsitsiklis}{1987}]%
        {papadimitriou1987complexity}
\bibfield{author}{\bibinfo{person}{Christos~H Papadimitriou} {and}
  \bibinfo{person}{John~N Tsitsiklis}.} \bibinfo{year}{1987}\natexlab{}.
\newblock \showarticletitle{The complexity of Markov decision processes}.
\newblock \bibinfo{journal}{\emph{Mathematics of operations research}}
  \bibinfo{volume}{12}, \bibinfo{number}{3} (\bibinfo{year}{1987}),
  \bibinfo{pages}{441--450}.
\newblock


\bibitem[\protect\citeauthoryear{Raffin, Hill, Gleave, Kanervisto, Ernestus,
  and Dormann}{Raffin et~al\mbox{.}}{2021}]%
        {stable-baselines3}
\bibfield{author}{\bibinfo{person}{Antonin Raffin}, \bibinfo{person}{Ashley
  Hill}, \bibinfo{person}{Adam Gleave}, \bibinfo{person}{Anssi Kanervisto},
  \bibinfo{person}{Maximilian Ernestus}, {and} \bibinfo{person}{Noah Dormann}.}
  \bibinfo{year}{2021}\natexlab{}.
\newblock \showarticletitle{Stable-Baselines3: Reliable Reinforcement Learning
  Implementations}.
\newblock \bibinfo{journal}{\emph{Journal of Machine Learning Research}}
  \bibinfo{volume}{22}, \bibinfo{number}{268} (\bibinfo{year}{2021}),
  \bibinfo{pages}{1--8}.
\newblock
\urldef\tempurl%
\url{http://jmlr.org/papers/v22/20-1364.html}
\showURL{%
\tempurl}


\bibitem[\protect\citeauthoryear{Rahman, Carlucho, H{\"o}pner, and
  Albrecht}{Rahman et~al\mbox{.}}{2023}]%
        {rahman2023general}
\bibfield{author}{\bibinfo{person}{Arrasy Rahman}, \bibinfo{person}{Ignacio
  Carlucho}, \bibinfo{person}{Niklas H{\"o}pner}, {and}
  \bibinfo{person}{Stefano~V Albrecht}.} \bibinfo{year}{2023}\natexlab{}.
\newblock \showarticletitle{A general learning framework for open ad hoc
  teamwork using graph-based policy learning}.
\newblock \bibinfo{journal}{\emph{Journal of Machine Learning Research}}
  \bibinfo{volume}{24}, \bibinfo{number}{298} (\bibinfo{year}{2023}),
  \bibinfo{pages}{1--74}.
\newblock


\bibitem[\protect\citeauthoryear{Rashid, Samvelyan, De~Witt, Farquhar,
  Foerster, and Whiteson}{Rashid et~al\mbox{.}}{2020}]%
        {rashid2020monotonic}
\bibfield{author}{\bibinfo{person}{Tabish Rashid}, \bibinfo{person}{Mikayel
  Samvelyan}, \bibinfo{person}{Christian~Schroeder De~Witt},
  \bibinfo{person}{Gregory Farquhar}, \bibinfo{person}{Jakob Foerster}, {and}
  \bibinfo{person}{Shimon Whiteson}.} \bibinfo{year}{2020}\natexlab{}.
\newblock \showarticletitle{Monotonic value function factorisation for deep
  multi-agent reinforcement learning}.
\newblock \bibinfo{journal}{\emph{Journal of Machine Learning Research}}
  \bibinfo{volume}{21}, \bibinfo{number}{178} (\bibinfo{year}{2020}),
  \bibinfo{pages}{1--51}.
\newblock


\bibitem[\protect\citeauthoryear{Rosin}{Rosin}{2011}]%
        {rosin2011multi}
\bibfield{author}{\bibinfo{person}{Christopher~D Rosin}.}
  \bibinfo{year}{2011}\natexlab{}.
\newblock \showarticletitle{Multi-armed bandits with episode context}.
\newblock \bibinfo{journal}{\emph{Annals of Mathematics and Artificial
  Intelligence}} \bibinfo{volume}{61}, \bibinfo{number}{3}
  (\bibinfo{year}{2011}), \bibinfo{pages}{203--230}.
\newblock


\bibitem[\protect\citeauthoryear{Samothrakis, Robles, and Lucas}{Samothrakis
  et~al\mbox{.}}{2011}]%
        {samothrakis2011fast}
\bibfield{author}{\bibinfo{person}{Spyridon Samothrakis},
  \bibinfo{person}{David Robles}, {and} \bibinfo{person}{Simon Lucas}.}
  \bibinfo{year}{2011}\natexlab{}.
\newblock \showarticletitle{Fast approximate max-n monte carlo tree search for
  ms pac-man}.
\newblock \bibinfo{journal}{\emph{IEEE Transactions on Computational
  Intelligence and AI in Games}} \bibinfo{volume}{3}, \bibinfo{number}{2}
  (\bibinfo{year}{2011}), \bibinfo{pages}{142--154}.
\newblock


\bibitem[\protect\citeauthoryear{Samvelyan, Rashid, de~Witt, Farquhar,
  Nardelli, Rudner, Hung, Torr, Foerster, and Whiteson}{Samvelyan
  et~al\mbox{.}}{2019}]%
        {samvelyan19smac}
\bibfield{author}{\bibinfo{person}{Mikayel Samvelyan}, \bibinfo{person}{Tabish
  Rashid}, \bibinfo{person}{Christian~Schroeder de Witt},
  \bibinfo{person}{Gregory Farquhar}, \bibinfo{person}{Nantas Nardelli},
  \bibinfo{person}{Tim G.~J. Rudner}, \bibinfo{person}{Chia-Man Hung},
  \bibinfo{person}{Philiph H.~S. Torr}, \bibinfo{person}{Jakob Foerster}, {and}
  \bibinfo{person}{Shimon Whiteson}.} \bibinfo{year}{2019}\natexlab{}.
\newblock \showarticletitle{{The} {StarCraft} {Multi}-{Agent} {Challenge}}.
\newblock \bibinfo{journal}{\emph{CoRR}}  \bibinfo{volume}{abs/1902.04043}
  (\bibinfo{year}{2019}).
\newblock


\bibitem[\protect\citeauthoryear{Schrittwieser, Antonoglou, Hubert, Simonyan,
  Sifre, Schmitt, Guez, Lockhart, Hassabis, Graepel,
  et~al\mbox{.}}{Schrittwieser et~al\mbox{.}}{2020}]%
        {schrittwieser2020mastering}
\bibfield{author}{\bibinfo{person}{Julian Schrittwieser},
  \bibinfo{person}{Ioannis Antonoglou}, \bibinfo{person}{Thomas Hubert},
  \bibinfo{person}{Karen Simonyan}, \bibinfo{person}{Laurent Sifre},
  \bibinfo{person}{Simon Schmitt}, \bibinfo{person}{Arthur Guez},
  \bibinfo{person}{Edward Lockhart}, \bibinfo{person}{Demis Hassabis},
  \bibinfo{person}{Thore Graepel}, {et~al\mbox{.}}}
  \bibinfo{year}{2020}\natexlab{}.
\newblock \showarticletitle{Mastering atari, go, chess and shogi by planning
  with a learned model}.
\newblock \bibinfo{journal}{\emph{Nature}} \bibinfo{volume}{588},
  \bibinfo{number}{7839} (\bibinfo{year}{2020}), \bibinfo{pages}{604--609}.
\newblock


\bibitem[\protect\citeauthoryear{Schulman, Wolski, Dhariwal, Radford, and
  Klimov}{Schulman et~al\mbox{.}}{2017}]%
        {schulman2017proximal}
\bibfield{author}{\bibinfo{person}{John Schulman}, \bibinfo{person}{Filip
  Wolski}, \bibinfo{person}{Prafulla Dhariwal}, \bibinfo{person}{Alec Radford},
  {and} \bibinfo{person}{Oleg Klimov}.} \bibinfo{year}{2017}\natexlab{}.
\newblock \showarticletitle{Proximal policy optimization algorithms}.
\newblock \bibinfo{journal}{\emph{arXiv preprint arXiv:1707.06347}}
  (\bibinfo{year}{2017}).
\newblock


\bibitem[\protect\citeauthoryear{Schwartz and Kurniawati}{Schwartz and
  Kurniawati}{2023}]%
        {schwartz2023bayes}
\bibfield{author}{\bibinfo{person}{Jonathon Schwartz} {and}
  \bibinfo{person}{Hanna Kurniawati}.} \bibinfo{year}{2023}\natexlab{}.
\newblock \showarticletitle{Bayes-Adaptive Monte-Carlo Planning for Type-Based
  Reasoning in Large Partially Observable, Multi-Agent Environments}. In
  \bibinfo{booktitle}{\emph{Proceedings of the 2023 International Conference on
  Autonomous Agents and Multiagent Systems}}. \bibinfo{pages}{2355--2357}.
\newblock


\bibitem[\protect\citeauthoryear{Schwartz, Kurniawati, and Hutter}{Schwartz
  et~al\mbox{.}}{2023}]%
        {schwartz2023combining}
\bibfield{author}{\bibinfo{person}{Jonathon Schwartz}, \bibinfo{person}{Hanna
  Kurniawati}, {and} \bibinfo{person}{Marcus Hutter}.}
  \bibinfo{year}{2023}\natexlab{}.
\newblock \showarticletitle{Combining a Meta-Policy and Monte-Carlo Planning
  for Scalable Type-Based Reasoning in Partially Observable Environments}.
\newblock \bibinfo{journal}{\emph{arXiv preprint arXiv:2306.06067}}
  (\bibinfo{year}{2023}).
\newblock


\bibitem[\protect\citeauthoryear{Shapley}{Shapley}{1953}]%
        {shapley1953stochastic}
\bibfield{author}{\bibinfo{person}{Lloyd~S Shapley}.}
  \bibinfo{year}{1953}\natexlab{}.
\newblock \showarticletitle{Stochastic games}.
\newblock \bibinfo{journal}{\emph{Proceedings of the national academy of
  sciences}} \bibinfo{volume}{39}, \bibinfo{number}{10} (\bibinfo{year}{1953}),
  \bibinfo{pages}{1095--1100}.
\newblock


\bibitem[\protect\citeauthoryear{Sharon, Stern, Felner, and Sturtevant}{Sharon
  et~al\mbox{.}}{2015}]%
        {sharon2015conflict}
\bibfield{author}{\bibinfo{person}{Guni Sharon}, \bibinfo{person}{Roni Stern},
  \bibinfo{person}{Ariel Felner}, {and} \bibinfo{person}{Nathan~R Sturtevant}.}
  \bibinfo{year}{2015}\natexlab{}.
\newblock \showarticletitle{Conflict-based search for optimal multi-agent
  pathfinding}.
\newblock \bibinfo{journal}{\emph{Artificial Intelligence}}
  \bibinfo{volume}{219} (\bibinfo{year}{2015}), \bibinfo{pages}{40--66}.
\newblock


\bibitem[\protect\citeauthoryear{Silver, Hubert, Schrittwieser, Antonoglou,
  Lai, Guez, Lanctot, Sifre, Kumaran, Graepel, Lillicrap, Simonyan, and
  Hassabis}{Silver et~al\mbox{.}}{2018}]%
        {doi:10.1126/science.aar6404}
\bibfield{author}{\bibinfo{person}{David Silver}, \bibinfo{person}{Thomas
  Hubert}, \bibinfo{person}{Julian Schrittwieser}, \bibinfo{person}{Ioannis
  Antonoglou}, \bibinfo{person}{Matthew Lai}, \bibinfo{person}{Arthur Guez},
  \bibinfo{person}{Marc Lanctot}, \bibinfo{person}{Laurent Sifre},
  \bibinfo{person}{Dharshan Kumaran}, \bibinfo{person}{Thore Graepel},
  \bibinfo{person}{Timothy Lillicrap}, \bibinfo{person}{Karen Simonyan}, {and}
  \bibinfo{person}{Demis Hassabis}.} \bibinfo{year}{2018}\natexlab{}.
\newblock \showarticletitle{A general reinforcement learning algorithm that
  masters chess, shogi, and Go through self-play}.
\newblock \bibinfo{journal}{\emph{Science}} \bibinfo{volume}{362},
  \bibinfo{number}{6419} (\bibinfo{year}{2018}), \bibinfo{pages}{1140--1144}.
\newblock
\urldef\tempurl%
\url{https://doi.org/10.1126/science.aar6404}
\showDOI{\tempurl}
\showeprint{https://www.science.org/doi/pdf/10.1126/science.aar6404}


\bibitem[\protect\citeauthoryear{Silver, Schrittwieser, Simonyan, Antonoglou,
  Huang, Guez, Hubert, Baker, Lai, Bolton, et~al\mbox{.}}{Silver
  et~al\mbox{.}}{2017}]%
        {silver2017mastering}
\bibfield{author}{\bibinfo{person}{David Silver}, \bibinfo{person}{Julian
  Schrittwieser}, \bibinfo{person}{Karen Simonyan}, \bibinfo{person}{Ioannis
  Antonoglou}, \bibinfo{person}{Aja Huang}, \bibinfo{person}{Arthur Guez},
  \bibinfo{person}{Thomas Hubert}, \bibinfo{person}{Lucas Baker},
  \bibinfo{person}{Matthew Lai}, \bibinfo{person}{Adrian Bolton},
  {et~al\mbox{.}}} \bibinfo{year}{2017}\natexlab{}.
\newblock \showarticletitle{Mastering the game of go without human knowledge}.
\newblock \bibinfo{journal}{\emph{nature}} \bibinfo{volume}{550},
  \bibinfo{number}{7676} (\bibinfo{year}{2017}), \bibinfo{pages}{354--359}.
\newblock


\bibitem[\protect\citeauthoryear{Silver and Veness}{Silver and Veness}{2010}]%
        {silver2010monte}
\bibfield{author}{\bibinfo{person}{David Silver} {and} \bibinfo{person}{Joel
  Veness}.} \bibinfo{year}{2010}\natexlab{}.
\newblock \showarticletitle{Monte-Carlo planning in large POMDPs}.
\newblock \bibinfo{journal}{\emph{Advances in neural information processing
  systems}}  \bibinfo{volume}{23} (\bibinfo{year}{2010}).
\newblock


\bibitem[\protect\citeauthoryear{Smallwood and Sondik}{Smallwood and
  Sondik}{1973}]%
        {smallwood1973optimal}
\bibfield{author}{\bibinfo{person}{Richard~D Smallwood} {and}
  \bibinfo{person}{Edward~J Sondik}.} \bibinfo{year}{1973}\natexlab{}.
\newblock \showarticletitle{The optimal control of partially observable Markov
  processes over a finite horizon}.
\newblock \bibinfo{journal}{\emph{Operations research}} \bibinfo{volume}{21},
  \bibinfo{number}{5} (\bibinfo{year}{1973}), \bibinfo{pages}{1071--1088}.
\newblock


\bibitem[\protect\citeauthoryear{Solan and Vieille}{Solan and Vieille}{2015}]%
        {solan2015stochastic}
\bibfield{author}{\bibinfo{person}{Eilon Solan} {and} \bibinfo{person}{Nicolas
  Vieille}.} \bibinfo{year}{2015}\natexlab{}.
\newblock \showarticletitle{Stochastic games}.
\newblock \bibinfo{journal}{\emph{Proceedings of the National Academy of
  Sciences}} \bibinfo{volume}{112}, \bibinfo{number}{45}
  (\bibinfo{year}{2015}), \bibinfo{pages}{13743--13746}.
\newblock


\bibitem[\protect\citeauthoryear{Sondik}{Sondik}{1978}]%
        {sondik1978optimal}
\bibfield{author}{\bibinfo{person}{Edward~J Sondik}.}
  \bibinfo{year}{1978}\natexlab{}.
\newblock \showarticletitle{The optimal control of partially observable Markov
  processes over the infinite horizon: Discounted costs}.
\newblock \bibinfo{journal}{\emph{Operations research}} \bibinfo{volume}{26},
  \bibinfo{number}{2} (\bibinfo{year}{1978}), \bibinfo{pages}{282--304}.
\newblock


\bibitem[\protect\citeauthoryear{Stahl}{Stahl}{1993}]%
        {stahl1993evolution}
\bibfield{author}{\bibinfo{person}{Dale~O Stahl}.}
  \bibinfo{year}{1993}\natexlab{}.
\newblock \showarticletitle{Evolution of smartn players}.
\newblock \bibinfo{journal}{\emph{Games and Economic Behavior}}
  \bibinfo{volume}{5}, \bibinfo{number}{4} (\bibinfo{year}{1993}),
  \bibinfo{pages}{604--617}.
\newblock


\bibitem[\protect\citeauthoryear{Stern}{Stern}{2019}]%
        {stern2019multi-overview}
\bibfield{author}{\bibinfo{person}{Roni Stern}.}
  \bibinfo{year}{2019}\natexlab{}.
\newblock \showarticletitle{Multi-agent path finding--an overview}.
\newblock \bibinfo{journal}{\emph{Artificial Intelligence}}
  (\bibinfo{year}{2019}), \bibinfo{pages}{96--115}.
\newblock


\bibitem[\protect\citeauthoryear{Stern, Sturtevant, Felner, Koenig, Ma, Walker,
  Li, Atzmon, Cohen, Kumar, et~al\mbox{.}}{Stern et~al\mbox{.}}{2019}]%
        {stern2019multi}
\bibfield{author}{\bibinfo{person}{Roni Stern}, \bibinfo{person}{Nathan~R
  Sturtevant}, \bibinfo{person}{Ariel Felner}, \bibinfo{person}{Sven Koenig},
  \bibinfo{person}{Hang Ma}, \bibinfo{person}{Thayne~T Walker},
  \bibinfo{person}{Jiaoyang Li}, \bibinfo{person}{Dor Atzmon},
  \bibinfo{person}{Liron Cohen}, \bibinfo{person}{TK~Satish Kumar},
  {et~al\mbox{.}}} \bibinfo{year}{2019}\natexlab{}.
\newblock \showarticletitle{Multi-agent pathfinding: Definitions, variants, and
  benchmarks}. In \bibinfo{booktitle}{\emph{Twelfth Annual Symposium on
  Combinatorial Search}}.
\newblock


\bibitem[\protect\citeauthoryear{Tang}{Tang}{2017}]%
        {tang2017reinforcement}
\bibfield{author}{\bibinfo{person}{Pingzhong Tang}.}
  \bibinfo{year}{2017}\natexlab{}.
\newblock \showarticletitle{Reinforcement mechanism design}. In
  \bibinfo{booktitle}{\emph{Proceedings of the 26th International Joint
  Conference on Artificial Intelligence}}. \bibinfo{pages}{5146--5150}.
\newblock


\bibitem[\protect\citeauthoryear{Yu, Velu, Vinitsky, Gao, Wang, Bayen, and
  Wu}{Yu et~al\mbox{.}}{2022}]%
        {yu2022the}
\bibfield{author}{\bibinfo{person}{Chao Yu}, \bibinfo{person}{Akash Velu},
  \bibinfo{person}{Eugene Vinitsky}, \bibinfo{person}{Jiaxuan Gao},
  \bibinfo{person}{Yu Wang}, \bibinfo{person}{Alexandre Bayen}, {and}
  \bibinfo{person}{Yi Wu}.} \bibinfo{year}{2022}\natexlab{}.
\newblock \showarticletitle{The Surprising Effectiveness of {PPO} in
  Cooperative Multi-Agent Games}. In \bibinfo{booktitle}{\emph{Thirty-sixth
  Conference on Neural Information Processing Systems Datasets and Benchmarks
  Track}}.
\newblock
\urldef\tempurl%
\url{https://openreview.net/forum?id=YVXaxB6L2Pl}
\showURL{%
\tempurl}


\bibitem[\protect\citeauthoryear{Yu and LaValle}{Yu and LaValle}{2013}]%
        {yu2013structure}
\bibfield{author}{\bibinfo{person}{Jingjin Yu} {and} \bibinfo{person}{Steven
  LaValle}.} \bibinfo{year}{2013}\natexlab{}.
\newblock \showarticletitle{Structure and intractability of optimal multi-robot
  path planning on graphs}. In \bibinfo{booktitle}{\emph{Proceedings of the
  AAAI Conference on Artificial Intelligence}}, Vol.~\bibinfo{volume}{27}.
  \bibinfo{pages}{1443--1449}.
\newblock


\bibitem[\protect\citeauthoryear{Zhang, Li, Surynek, Kumar, and Koenig}{Zhang
  et~al\mbox{.}}{2022}]%
        {zhang2022multi}
\bibfield{author}{\bibinfo{person}{Han Zhang}, \bibinfo{person}{Jiaoyang Li},
  \bibinfo{person}{Pavel Surynek}, \bibinfo{person}{TK~Satish Kumar}, {and}
  \bibinfo{person}{Sven Koenig}.} \bibinfo{year}{2022}\natexlab{}.
\newblock \showarticletitle{Multi-agent path finding with mutex propagation}.
\newblock \bibinfo{journal}{\emph{Artificial Intelligence}}
  \bibinfo{volume}{311} (\bibinfo{year}{2022}), \bibinfo{pages}{103766}.
\newblock


\bibitem[\protect\citeauthoryear{Zhang and Zhang}{Zhang and Zhang}{2001}]%
        {zhang2001speeding}
\bibfield{author}{\bibinfo{person}{Nevin~Lianwen Zhang} {and}
  \bibinfo{person}{Weihong Zhang}.} \bibinfo{year}{2001}\natexlab{}.
\newblock \showarticletitle{Speeding up the convergence of value iteration in
  partially observable Markov decision processes}.
\newblock \bibinfo{journal}{\emph{Journal of Artificial Intelligence Research}}
   \bibinfo{volume}{14} (\bibinfo{year}{2001}), \bibinfo{pages}{29--51}.
\newblock


\bibitem[\protect\citeauthoryear{Zhou, Wan, Wang, Wen, Wu, Wen, Yang, Yu, Wang,
  and Zhang}{Zhou et~al\mbox{.}}{2023}]%
        {JMLR:v24:22-0169}
\bibfield{author}{\bibinfo{person}{Ming Zhou}, \bibinfo{person}{Ziyu Wan},
  \bibinfo{person}{Hanjing Wang}, \bibinfo{person}{Muning Wen},
  \bibinfo{person}{Runzhe Wu}, \bibinfo{person}{Ying Wen},
  \bibinfo{person}{Yaodong Yang}, \bibinfo{person}{Yong Yu},
  \bibinfo{person}{Jun Wang}, {and} \bibinfo{person}{Weinan Zhang}.}
  \bibinfo{year}{2023}\natexlab{}.
\newblock \showarticletitle{MALib: A Parallel Framework for Population-based
  Multi-agent Reinforcement Learning}.
\newblock \bibinfo{journal}{\emph{Journal of Machine Learning Research}}
  \bibinfo{volume}{24}, \bibinfo{number}{150} (\bibinfo{year}{2023}),
  \bibinfo{pages}{1--12}.
\newblock
\urldef\tempurl%
\url{http://jmlr.org/papers/v24/22-0169.html}
\showURL{%
\tempurl}


\bibitem[\protect\citeauthoryear{Zhu and Lin}{Zhu and Lin}{2024}]%
        {zhu2024computinguniversalplanspartially}
\bibfield{author}{\bibinfo{person}{Fengming Zhu} {and}
  \bibinfo{person}{Fangzhen Lin}.} \bibinfo{year}{2024}\natexlab{}.
\newblock \bibinfo{title}{On Computing Universal Plans for Partially Observable
  Multi-Agent Path Finding}.
\newblock
\newblock
\showeprint[arxiv]{2305.16203}~[cs.MA]
\urldef\tempurl%
\url{https://arxiv.org/abs/2305.16203}
\showURL{%
\tempurl}


\bibitem[\protect\citeauthoryear{Zinkevich, Johanson, Bowling, and
  Piccione}{Zinkevich et~al\mbox{.}}{2007}]%
        {zinkevich2007regret}
\bibfield{author}{\bibinfo{person}{Martin Zinkevich}, \bibinfo{person}{Michael
  Johanson}, \bibinfo{person}{Michael Bowling}, {and} \bibinfo{person}{Carmelo
  Piccione}.} \bibinfo{year}{2007}\natexlab{}.
\newblock \showarticletitle{Regret minimization in games with incomplete
  information}.
\newblock \bibinfo{journal}{\emph{Advances in neural information processing
  systems}}  \bibinfo{volume}{20} (\bibinfo{year}{2007}).
\newblock


\end{thebibliography}


\balance

\clearpage
\onecolumn
\appendix


\section{Theoretic Analysis of The Unified Framework}

\label{apd:converge}

For the belief-fixed lookahead search presented in Section~\ref{sec:unification}.(\ref{eq:van_bellman}),
since the belief is fixed, we can make the notations simpler.
Let $v(S) \triangleq V_i(S,b)$, and let $\mathcal{V}$ be the space of all possible such value functions,
and $\Gamma: \mathcal{V} \mapsto \mathcal{V}$ be the operator that does the job of this backup equation.
\begin{theorem}
	The backup operator $\Gamma$ in Section~\ref{sec:unification}.(\ref{eq:van_bellman}) is a $\gamma$-contraction. Mathematically, for $u,v \in \mathcal{V}$, we have
	\[
	\|\Gamma(u) - \Gamma(v) \|_\infty \leq \gamma \|u - v\|_\infty
	\]
\label{thm:van_bellman_contraction}
\end{theorem}

\begin{proof}
First, we fist convert it from the expectation notation back to the summation notation,
\[
\begin{split}
v(S)
& = \max_{a_i\in \mathcal{A}_i}\{\mathbb{E}_{\pi_{-i} \sim b, a_{-i}\in \mathcal{A}_{-i}}[R_i(s,a) +  \gamma \sum_{S'}T(S'|S,a) v(S')]\} \\
& = \max_{a_i\in \mathcal{A}_i}\{\sum_{\pi_{-i}\in b} b(\pi_{-i})\sum_{a_{-i} \in \mathcal{A}_{-i}}R_i(S,a)\pi_{-i}(a_{-i}|S) + \gamma \sum_{\pi_{-i}\in b} b(\pi_{-i})\sum_{a_{-i} \in \mathcal{A}_{-i}}T(S'|S,a)\pi_{-i}(a_{-i}|S) v(S')\} \\
& = \max_{a_i\in \mathcal{A}_i} \{R^{b}(S, a_i) + \gamma T^{b}(S'| S, a_i) v'(S) \} \\
\end{split}
\]
The last line is to simplify the equation, by our belief-induced reward and transition functions defined in Section \ref{sec:belief_induced_mdp}. With slight abuses of notations, we write $\Gamma_v(S)$ for $\Gamma(v)(S)$. In the following, we prove our target in two sub-cases,
\begin{enumerate}
\item For $S \in \mathcal{S}$, such that $\Gamma_u(S) > \Gamma_v(S)$.
We choose $a_i^* \in \arg\max_{a_i\in\mathcal{A}_i}\{R^b(S,a_i) + \gamma\sum_{S'} T^b(S'|S,a_i)u(S')\}$, then
\[
\begin{split}
|\Gamma_u(S) - \Gamma_v(S)|
& = \Gamma_u(S) - \Gamma_v(S) \\
& = R^b(S,a_i^*) + \gamma\sum_{S'} T^b(S'|S,a_i^*)u(S')
- \max_{a_i\in A_i}\{R^b(S,a_i) + \gamma\sum_{S'} T^b(S'|S,a_i)v(S')\} \\
& \leq R^b(S,a_i^*) + \gamma\sum_{S'} T^b(S'|S,a_i^*)u(S')
- [R^b(S,a_i^*) + \gamma\sum_{S'} T^b(S'|S,a_i^*)v(S')] \\
& \leq \gamma\sum_{S'} T^b(S'|S,a_i^*) [u(S') - v(S')] \\
& \leq \gamma\sum_{S'} T^b(S'|S,a_i^*) |u(S') - v(S')| \\
& \leq \gamma\sum_{S'} T^b(S'|S,a_i^*) \|u - v\|_\infty
\end{split}
\]

We then show that $\sum_{S'} T^b(S'|S,a_i) = 1$ given any $a_i \in \mathcal{A}_i$,
i.e., $T^b$ is indeed a valid (stochastic) transition function,
\[
\begin{split}
\sum_{S'\in \mathcal{S}} T^b(S'|S,a_i^*)
& = \sum_{S'\in \mathcal{S}}\sum_{\pi_{-i}\in b} b(\pi_{-i})\sum_{a_{-i} \in \mathcal{A}_{-i}}T(S'|S,a)\pi_{-i}(a_{-i}|S) \\
& = \sum_{S'\in \mathcal{S}} \sum_{a_{-i} \in \mathcal{A}_{-i}} \sum_{\pi_{-i} \in b}
T(S'|S,a)\pi_{-i}(a_{-i}|s)b_{-i}(\pi_{-i}) \\
& = \sum_{\pi_{-i} \in b}  b(\pi_{-i})
\sum_{a_{-i} \in \mathcal{A}_{-i}} \pi_{-i}(a_{-i}|s)
\sum_{S'\in \mathcal{S}} T(S'|S,a) \\
& = \sum_{\pi_{-i} \in b}  b(\pi_{-i})
\sum_{a_{-i} \in \mathcal{A}_{-i}} \pi_{-i}(a_{-i}|s) \\
& = \sum_{\pi_{-i} \in b}  b(\pi_{-i}) \\
& = 1
\end{split}
\]

\item For $S \in \mathcal{S}$, such that $\Gamma_u(S) < \Gamma_v(S)$. Similarly, we can also have
\[
|\Gamma_u(S) - \Gamma_v(S)|
= \Gamma_v(S) - \Gamma_u(S)
\leq \gamma |v(S') - u(S')|
\leq \gamma \|u - v\|_\infty
\]
\end{enumerate}
As a result, for all $S\in \mathcal{S}$, we have $|\Gamma_u(S) - \Gamma_v(S)| \leq \gamma \|u - v\|_\infty$. Hence,
$
\|\Gamma(u) - \Gamma(v)\|_\infty \leq \gamma \|u - v\|_\infty
$.
\end{proof}

By the above theorem, we can conclude that by the second-level belief-fixed lookahead search $V_i(S^{n},b^{n})$ converges to the optimal value function $v^*(S^n)$ of the induced MDP $\mathcal{M}(b^n)$, as $m$ approaches infinity.
For the first-level belief-updated lookahead search presented in Section \ref{sec:unification}.(1), we can also have a similar property.
\begin{theorem}
The backup operator in Section~\ref{sec:unification}.(\ref{eq:belief_bellman}) is a $\gamma$-contraction.
\end{theorem}
	
\begin{proof}
\[
V_i(S,b)
= \max_{a_i \in \mathcal{A}_i}\bigg\{ \mathcal{R}\Big( (S, b), a_i \Big)
+ \gamma \sum_{S',b'} \mathcal{T} \Big( (S', b') \Big| (S, b), a_i  \Big)\cdot V_i(S',b') \bigg\}
\]
Now that $V_i: \mathcal{S} \times \mathcal{B} \mapsto \mathbb{R}$ is a function over continuous variables, then $\|V_i\|_\infty \triangleq  \sup_{s,b} |V_i(S,b)|$. 
In fact, $V_i$ is piece-wise linear and convex, as it reveals the value function of the underlying POMDP \cite{smallwood1973optimal}, then ``$\sup$'' simply becomes ``$\max$'' and the rest of the proof will naturally proceed as that for Theorem \ref{thm:van_bellman_contraction}, as showing $\sum_{S',b'} \mathcal{T} ( (S', b') | (S, b), a_i  ) = 1$ for any given $a_i \in \mathcal{A}_i$ is also straightforward,
\[
\begin{split}
\sum_{S',b'} \mathcal{T} \Big( (S', b') \Big| (S, b), a_i  \Big)
& = \sum_{b' \in \mathcal{B}}\sum_{S' \in \mathcal{S}}\sum_{\pi_{-i}\in b} b(\pi_{-i})\sum_{a_{-i} \in \mathcal{A}_{-i}} T(S'|S, a) \pi_{-i}(a_{-i}|S) \\
& = \sum_{S' \in \mathcal{S}}\sum_{\pi_{-i}\in b} b(\pi_{-i})\sum_{a_{-i} \in \mathcal{A}_{-i}} T(S'|S, a) \pi_{-i}(a_{-i}|S) \\
& = 1
\end{split}
\]
The second last equality holds because the belief update process is deterministic, hence the unique successor belief, while the last equality is proving the same target as what we did in the proof of Theorem \ref{thm:van_bellman_contraction}.

\end{proof}

Therefore, the above theorem shows $V_i(S^0,b^0)$ will converge to the solution of Equation~(\ref{eq:pomdp}), as $n$ approaches infinity, i.e. the optimal value function $V_i^*(S^0,b^0)$ of the underlying POMDP.
More illustratively, it can be informally shown by Figure \ref{fig:conv},
where $\|EB\| = \gamma^m \|CB\|$, $\|AG\| = \gamma^m \|AC\|$, and $\|AD\| =\gamma^n \|AB\|$, $\|AF\| =\gamma^n \|AE\|$, $\|AH\| =\gamma^n \|AG\|$, therefore, by simple geometry,
$FD \parallel EB$ and $HF \parallel GE \parallel AB$.
Projecting our planning procedures to the diagram,
we start from C, go to E, and end up with F.
For several other alternatives,
\begin{enumerate}
	\item $C\rightarrow B \rightarrow D$ means one optimally solves the belief induced MDP first and then backup the value $n$ levels with updated beliefs. 
	\item $C\rightarrow G \rightarrow H$ means one optimally solves the finite-$(n+m)$-horizon POMDP with terminating states evaluated by $\textsc{Eval}_i$.
\end{enumerate}
Consequently, $\|FD\|$ and $\|HF\|$ are the respective distances of these two alternative solutions to that given by what we have proposed.

\begin{figure}[!ht]
	\centering
	\includegraphics[height=60mm]{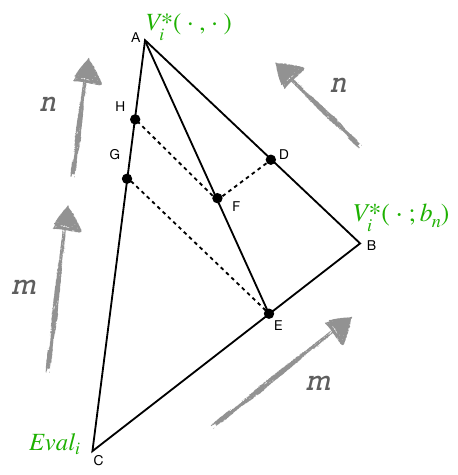}
	\caption{The convergence dynamics.}
	\label{fig:conv}
\end{figure}

%
%
%

\section{Planners in Pseudocode}
\label{apd:planner_detail}

By convention, we use $SG$ to denote an instantiated multi-agent environment, $SG.reset()$ launches a new episode and returns the initial state, and $SG.step(a_i, a_{-i})$ proceeds the environment by the given joint actions and returns the successor state. An unbounded \textbf{while} loop is used to represent a running episode, and will terminate automatically if $SG$ reaches an end state.

\begin{algorithm}[!ht]
\caption{Planning via belief-fixed MDPs}
\begin{algorithmic}[1]
%
%
\Function{Plan-Belief-Fixed}{$b^0$}
	\State $M \gets \mathcal{M}(b^0)$
	\Comment Induce an MDP
	\State $\pi_i \gets \textsc{solve}(M)$
	\State $S \gets SG.reset()$
	\While{True} \Comment Game loop
		\State $a_i \gets \pi_i(\cdot|S)$
		\Comment No replanning
		\State $S\gets SG.step(a_i, a_{-i})$
	\EndWhile
\EndFunction
\end{algorithmic}
\label{alg:mdp_fixed}
\end{algorithm}

\begin{algorithm}[!ht]
\caption{Planning via belief-updated MDPs}
\begin{algorithmic}[1]
\Function{Plan-Belief-updated}{$b^0$}
	\State $b \gets b^0$
	\State $M \gets \mathcal{M}(b)$
	\Comment Induce an MDP
	\State $\pi_i \gets \textsc{solve}(M)$
	\State $S \gets SG.reset()$
	\While{True} \Comment Game loop
		\State $a_i \gets \pi_i(\cdot|S)$
		\State $S \gets SG.step(a_i, a_{-i})$
		\State $b \gets \xi(b, S, a)$
		\Comment $a$ is a shorthand for $(a_i, a_{-i})$
		\State $M \gets \mathcal{M}(b)$
		\Comment Replanning on the revised MDP
		\State $\pi_i \gets \textsc{solve}(M)$
	\EndWhile
\EndFunction
\end{algorithmic}
\label{alg:mdp_updated}
\end{algorithm}

\begin{algorithm}[!ht]
\caption{Planning via QMDPs}
\label{alg:qmdp}
\begin{algorithmic}[1]
\Function{Plan-QMDP}{$b^0$}
	\For{each $\pi_{-i} \in b^0$}
		\State $Q_{\pi_{-i}} \gets \textsc{Solve}(\mathcal{M}(\pi_{-i}))$
		\Comment Get the Q values instead of the policy
	\EndFor
	\State $b\gets b^0$
	\State $S \gets SG.reset()$
	\While{True} \Comment Game loop
		\State $a_i \in \arg\max_{a\in \mathcal{A}_i} \sum_{\pi_{-i} \in b}Q_{\pi_{-i}}(S, a) \cdot b(\pi_{-i})$
		\Comment No need to replan
		\State $S\gets SG.step(a_i, a_{-i})$
		\State $b \gets \xi(b, S, a)$
	\EndWhile
\EndFunction
\end{algorithmic}
\end{algorithm}

\begin{algorithm}[!ht]
\caption{Planning via ContextualRL}
\label{alg:crl}
\begin{algorithmic}[1]

\Function{Plan-ContextualRL}{$b0$}
	\State $SG' \gets \textsc{WrapAsSamplingEnv}(SG)$
	\State $\pi^* \gets \textsc{AnyLearner}(SG')$
	\State $b \gets b^0$
	\State $S \gets SG.reset()$
	\While{True} \Comment Game loop
		\State $a_i \sim \pi^*(\cdot|S,b)$
		\Comment $\pi$ might be a stochastic policy
		\State $S\gets SG.step(a_i, a_{-i})$
		\State $b \gets \xi(b, S, a)$
	\EndWhile
\EndFunction
\end{algorithmic}
\end{algorithm}

We first present the pseudocode of the planners we mentioned in Section \ref{sec:solutions}, as Algorithm \ref{alg:mdp_fixed}, \ref{alg:mdp_updated}, \ref{alg:qmdp}, \ref{alg:crl}, respectively.
We did not attach the pseudocode for the POMDP planner, as it simply operates this way: one compiles the problem into a POMDP instance, gives it to a POMDP solver, and enquires the returned policy at each step without any replanning.

As illustrated in Figure \ref{fig:ts}, the general framework is implemented as an \textsc{ExpectiMax} tree.
The red triangle nodes are called ``MAX'' nodes, representing the states of the modelling agent $i$, while the orange diamond nodes are called ``EXP'' nodes, representing the hypothetical states\footnote{Also usually known as after-states.} that follow from the states given agent $i$'s committed action. The green diamond nodes are not explicitly implemented as they are just conceptual ones to show that each $a_{-i}$ is drawn probabilistically from the potential policies (or types) based on the belief. Algorithm \ref{alg:ts} shows the skeleton of how to use tree search as an online (re-)planner, and Algorithm \ref{alg:uts} shows the detailed procedure of how this \textsc{ExpectiMax} backup works.

\begin{figure}[!ht]
	\centering
	\includegraphics[height=55mm]{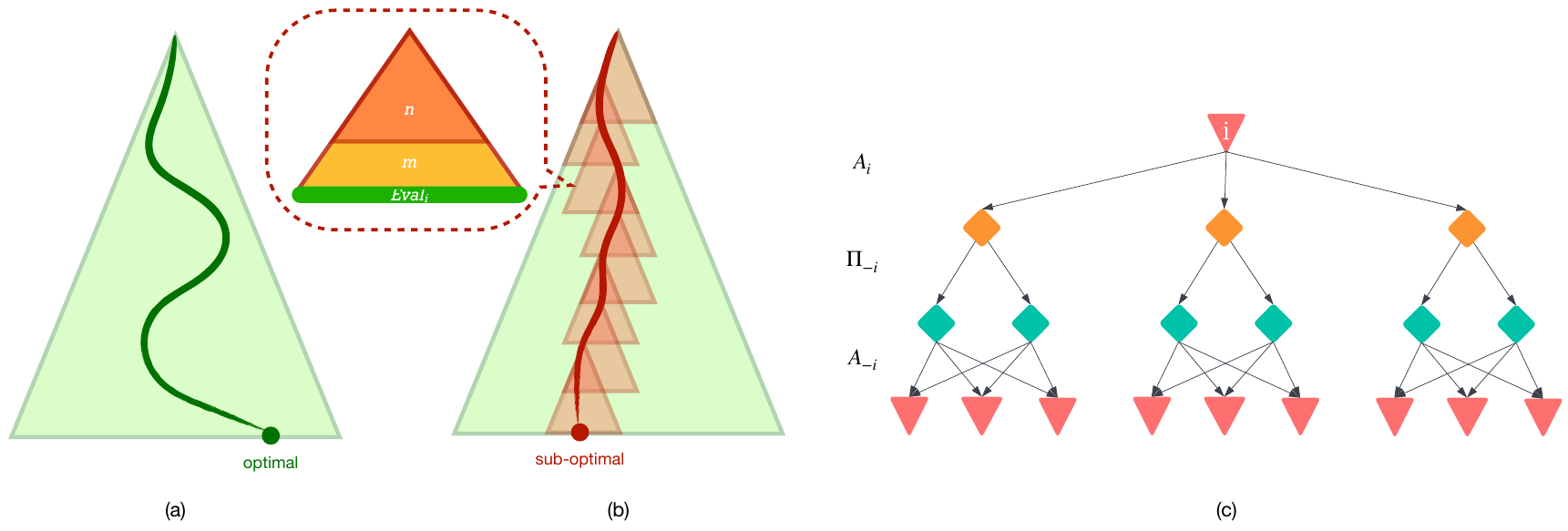}
	\caption{The exact optimal plan~(a), a potential approximated online plan with repeated replanning by layered tree search~(b), and a closer look at the tree diagram for one depth of the lookahead search~(c).}
	\label{fig:ts}
\end{figure}

\begin{algorithm}[!ht]
\caption{Planning via look-ahead tree search}
\label{alg:ts}
\begin{algorithmic}[1]

\Function{Plan-TS}{$b^0, compute\_budget$}
	\State $b \gets b^0$
	\State $S \gets SG.reset()$
	\While{True} \Comment Game loop
		\State $a_i \gets \textsc{TreeSearch.BestResponse}(S, b, compute\_budget)$
		\State $S \gets SG.step(a_i, a_{-i})$
		\State $b \gets \xi(b, S, a)$
	\EndWhile
\EndFunction
\end{algorithmic}
\label{alg:ts}
\end{algorithm}

The ultimate version of MCTS-like planner is described in Algorithm \ref{alg:mcts}.
Note that, in line 19, it has to utilize a exploration-exploitation-balanced choice function for node selection.
Given a tree node $t$, by $t.v$ we denote the accumulated return from this node state onwards, and $t.N$ the number of visits to this node.
The most widely used formula is the UCT formula \cite{kocsis2006bandit},
\begin{equation}
	t_c \in \arg\max_{t_c \in t.children} \frac{t_c.v}{t_c.N} + c\cdot \sqrt{\frac{\ln(t.N)}{t_c.N}}
\end{equation}
where $c$ is a constant controlling the weight of exploration and exploitation, with its best empirical value $\sqrt{2}$.
It is also proposed to use certain prior policies to guide the choices, resulting the pUCT formula \cite{rosin2011multi, silver2017mastering, schrittwieser2020mastering},
\begin{equation}
	t_c \in \arg\max_{t_c \in t.children} \frac{t_c.v}{t_c.N} + t_c.policy\_prior \cdot \sqrt{\frac{\ln(t.N)}{t_c.N}}\Bigl(c_1 + \ln \bigl(\frac{t.N + c_2}{c_2}  \bigr)  \Bigr)
\end{equation}
where $c_1$ and $c_2$ are two constants controlling the influence of the prior policy, with their commonly adopted empirical values $c_1=1.25$ and $c_2=19625$.
In principle, the prior policies are harder to acquire than the value estimations of the leaf nodes.
However, we here mention two ways to obtain both in practice,
\begin{enumerate}
	\item Via the approximate NE strategies computed by \textit{constraint satisfaction} solvers. One can sample the opponent types from the belief distribution for multiple runs, compute an approximate (ex-post) NE w.r.t. the sampled types at each run, and eventually obtain an average policy. The value estimate can be computed by the average utility of those approximate NE strategies. We will elaborate, in our MARP domain, how to convert CBS plans to these two components in Appendix \ref{apd:exp_setting}.
	\item Via policy predictions by the actor in any actor-critic RL algorithm.
	Any actor-critic RL algorithm like PPO \cite{schulman2017proximal} has an actor network to output certain logits, and a value network to predict a rough value of a given state.
	The action sampling distribution parameterized by the logits can serve as the policy prior, while the value prediction can directly be the desired value estimate.
\end{enumerate}

\begin{algorithm}[htpb]
\caption{Opponent-Modelling Uniform Tree Search (Exact Backup)}
\label{alg:uts}
\begin{algorithmic}[1]
\Function{BestResponse}{$S, b, total\_depth$}
	\State Initialize the root node as 
	\[
	t_{root} \gets
	(type=\texttt{`max'}, state=S, height=0, belief=b, a_{prev}=null, children=[], reward = 0, v=null)
	\]
	\For{$a_i \in \mathcal{A}_i$}
		\State $t_c \gets \textsc{NewChildNode}(\texttt{`exp'}, t_{root}, a_i)$
		\State $t_{root}.children.append(t_c)$
	\EndFor
	\State \textsc{MaxVal}$(t_{root}, 0, total\_depth)$
	\State $best\_child \gets \arg\max_{t_c \in t_{root}.chilren} t_c.v$
	\State \Return $best\_child.a_{prev}$
	
\EndFunction

%

\State
\Function{NewChildNode}{$type, parent, a$}
	\If{$type == \texttt{`exp'}$}
		\State $S' \gets parent.state$
		\State $h' \gets parent.height$
		\State $b' \gets parent.belief$
		\State
		$t_{new} \gets (type=\texttt{`exp'},
			   state=S', height=h', belief=b', a_{prev}=a,
			   children=[], reward=0, v=null)$
		
	\ElsIf{$type == \texttt{`max'}$}
		\State $a' \gets concatenate(parent.a_{prev}, a)$
		\Comment To compose a joint action
		\State $S', r' \gets transit\_and\_reward(parent.state, a')$
		\State $h' \gets parent.height + 1$
		\State $b' \gets belief\_update(parent.belief, parent.state, a')$
		\State
		$t_{new} \gets (type=\texttt{`max'},
			   state=S', height=h', belief=b', a_{prev}=a',
			   children=[], reward=r', v=null)$
	
	\EndIf
	\State $parent.children.append(t_{new})$
	\State \Return $t_{new}$
\EndFunction

\State
\Function{ExpVal}{$exp\_node, height, total\_depth$}
  	\For{$a_{-i} \in \mathcal{A}_{-i}$}
  		\State $child_{-i} \gets \textsc{NewChildNode}(\texttt{`max'}, exp\_node, a_{-i})$
  		\State $exp\_node.children.append(child_{-i})$
  	\EndFor
  	\State $exp\_node.v \gets$
  	$\sum_{a_{-i} \in \mathcal{A}_{-i}}
  			\sum_{\pi_{-i} \in b}b(\pi_{-i}) \pi_{-i}(a_{-i}|exp\_node.state)\cdot
  			[child_{-i}.reward + \gamma \cdot \textsc{MaxVal}(child_{-i}, height + 1, total\_depth)]$
  	\State \Return $t.v$
\EndFunction

\State
\Function{MaxVal}{$max\_node, height, total\_depth$}
  	\If{$height == total\_depth$}
  		\State $max\_node.v \gets \textsc{Eval}_i(max\_node.state, max\_node.belief)$
  	\Else
  		\For{$a_{i} \in \mathcal{A}_{i}$}
  			\State $child_{i} \gets \textsc{NewChildNode}(\texttt{`max'}, max\_node, a_i)$
  			\State $max\_node.children.append(child_{i})$
  		\EndFor
  		\State $max\_node.v \gets \max_{t_c\in t.children} \textsc{ExpVal}(max\_node, height, total\_depth)$
  	\EndIf
  	\State \Return $max\_node.v$
\EndFunction
\end{algorithmic}
\label{alg:uts}
\end{algorithm}

\begin{algorithm}[htbp]
\caption{Opponent-Modelling Monte Carlo Tree Search }
\label{alg:mcts}
\begin{algorithmic}[1]
\Function{BestResponse}{$S, b, time\_limit$}
	\State Initialize the root node as 
	\[
	t_{root} \gets
	(type=\texttt{`max'}, state=S, height=0, belief=b, a_{prev}=null, children=[], reward = 0, v=null, N=0)
	\]
	\While{not exceeding $time\_limit$}
		\State $t_{candidate} \gets \textsc{Select}(t_{root})$
		\State $t_{new} \gets \textsc{Expand}(t_{candidate})$
		\State $\textsc{Evaluate}(t_{new})$
		\Comment The node to evaluate must be a MAX node
		\State $\textsc{Backup}(t_{new})$
		\Comment Backup the value given by $\textsc{Eval}_i(t_{new}.state, t_{new}.belief)$
	\EndWhile
	\State $best\_child \gets \arg\max_{t_c \in t_{root}.children} t_c.N$
	\Comment select the action according to the Categorial distribution parameterized by $N$'s
	\State \Return $best\_child.a_{prev}$
	
\EndFunction

%

%
%

\State
\Function{Select}{$node$}
	\While{ True}
		\If{$node.type == \texttt{`max'}$}
			\If{$node$ is not fully expanded}
				\State \Return $node$ 
			\Else
				\State $node \gets \textsc{EEBalancedChoice}(node)$
				\Comment Balance exploration and exploitation
			\EndIf
		\ElsIf{$node.type == \texttt{`exp'}$}
			\State Sample $\pi_{-i} \sim node.belief$
				   for enough times to obtain a mean policy $\tilde\pi_{-i}$
			\State Sample $a_{-i} \sim \tilde\pi_{-i}(\cdot|node.state)$
			\If{$a_{-i}$ is not tried yet}
			\Comment Encounter a new MAX node that is not evaluated
				\State $child \gets \textsc{NewChildNode}(\texttt{`max'}, node, a_{-i})$
				\Comment Additionally set $child.N \gets 0$
				\State \Return $child$
			\Else
				\State $node \gets node.children[-i]$
				\Comment ``[-i]'' means the index corresponding to $a_{-i}$
			\EndIf
		\EndIf
	\EndWhile
	\State \Return node
\EndFunction

\State
\Function{Expand}{$node$}
	\If{$node$ is not yet evaluated}
	\Comment This can be done by checking whether $node.v$ is still $null$
		\State \Return $node$ 
	\EndIf
	\State $a_i \gets$ an untried but available action from $\mathcal{A}_i(node.State)$
	\State $child \gets \textsc{NewChildNode}(\texttt{`exp'}, node, a_i)$
	\State Sample $\pi_{-i} \sim child.belief$
				   for enough times to obtain a mean policy $\tilde\pi_{-i}$
	\State Sample $a_{-i} \sim \tilde\pi_{-i}(\cdot|child.state)$
	\State \Return $\textsc{NewChildNode}(\texttt{`max'}, child, a_{-i})$ 
\EndFunction


\State
\Function{Backup}{$node$}
	\State $G \gets node.v$
	\While{$node$ is not $null$}
		\If{$node.type == \texttt{`max'}$}
		\Comment Only compute discounted rewards at MAX nodes
			\State $G \gets node.reward + \gamma \cdot G$
		\EndIf
		\State $node \gets node.parent$
		\State $node.v \gets node.v + G$
		\State $node.N \gets node.N + 1$
		
	\EndWhile
\EndFunction
\end{algorithmic}
\label{alg:mcts}
\end{algorithm}

\clearpage

\section{More Experimental Details}
\label{apd:exp_setting}

\subsection{Convert CBS plans to NE Strategies}

Given a grid map and a set of initial positions and goals,
MAPF solvers, such as CBS~\cite{sharon2015conflict} and EECBS~\cite{li2021eecbs},
compute a set of collision-free paths. As we mentioned, if this set of collision-free paths is optimal up to certain metrics, e.g., minimizing the sum of lengths, it serves as an NE as no agent will deviate in the sense of finding a shorter path without colliding to any others. We here use EECBS as it supports users to specify an error bound $\epsilon$ and returns a bounded sub-optimal solution of total length no more than $(1+\epsilon)$ times the optimum. By merely sacrificing a bounded amount of solution quality, EECBS can speed up the solving process drastically, e.g., only needs an amount of time of the order of 10ms to solve instances of 32x32 maps with 50 agents. In our experiments, $\epsilon = 0.2$.

By such a solver as an oracle that can compute a sample NE very fast in real-time, we are then able to extract value estimates and policy priors for the usage in tree search algorithms.
Given a tree node $(S,b)$, where $S$ represent the current locations of all agents and $b$ is a distribution over all possible goals of the opponents that is inferred by the modelling agent, the modelling agent will go through the following procedure for multiple rounds,
\begin{enumerate}
	\item Samples a set of opponents' goals from $b$ and calls EECBS to compute a set of collision-free path from the current locations to the sampled goals.
	\item Extract her own path, which is a sequence of actions leading to her goal without colliding to others.
	\item Suppose the path is of length $l$ and the first action is $a^0$, the value  will be estimated as $\gamma^l \times R_i(goal)$ and $e_{a^0}$ will be the policy prior, where $e_i$ is the unit vector with the $i$-th element being 1.
\end{enumerate}
Finally, the value estimate and policy prior for this node $(S,b)$ will be computed as the mean of these values and policy priors collected above.

\subsection{Contextual-RL}
As mentioned in Section~\ref{sec:model}, Equation~(\ref{eq:pomdp}) also leads to potential (contextual-)RL solutions~\cite{benjamins2021carl}.
To this end, one needs to cast the given multi-agent stochastic game as a single-agent learning environment from the perspective of the modelling agent, by (i) first initializing each episode by sampling opponent strategies according to the initial belief, and (ii) then proceeding the environment by the given action of the modelling agent and the sampled actions of the opponents, updating the belief, and returning it to the modelling agent.

Consequently, the most challenging part lies in how to efficiently train such a policy that converges to the desired optimum.
We wrap our environment as a gym-like one, and then use PPO \cite{schulman2017proximal} implemented by \textit{stable-baseline3}\footnote{https://stable-baselines3.readthedocs.io/en/master/} \cite{stable-baselines3} to train our modelling agent. We also tried other alternatives like DQN and A2C, but they do not end up with acceptably good returns.
Figure \ref{fig:crl} shows sample experiments of the training phase, for the two configurations ``Small\_2a'' and ``Square\_2a'', respectively.
In both figures, there is clearly an intermediate plateau before convergence. It is usually the case that in a certain early phase the modelling agent does find a feasible plan to reach the goal without any collision, and in the later phase she eventually manages to find a much shorter plan, hence a much improved return.

\begin{figure}[!ht]
	\centering
	\includegraphics[height=80mm]{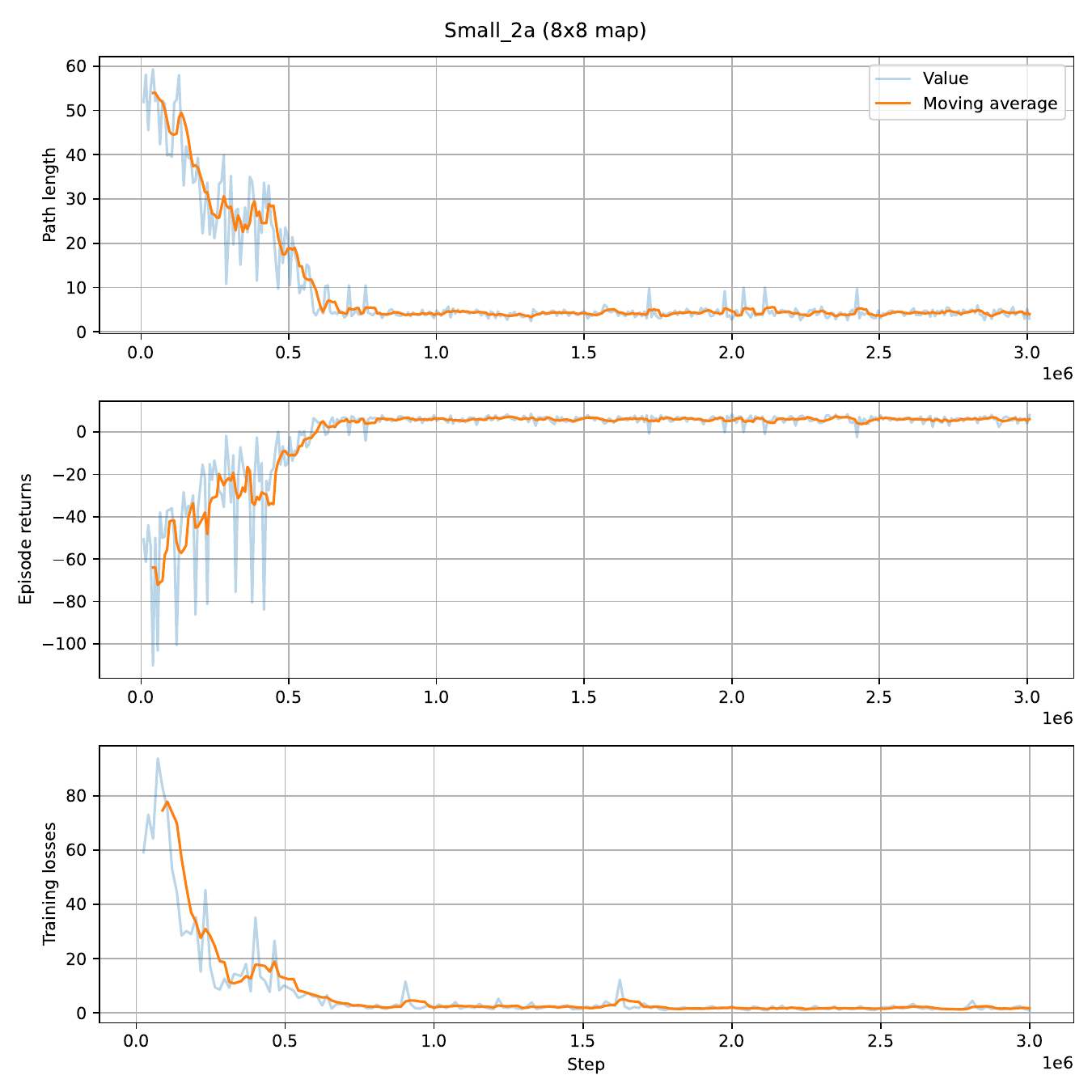}
	\includegraphics[height=80mm]{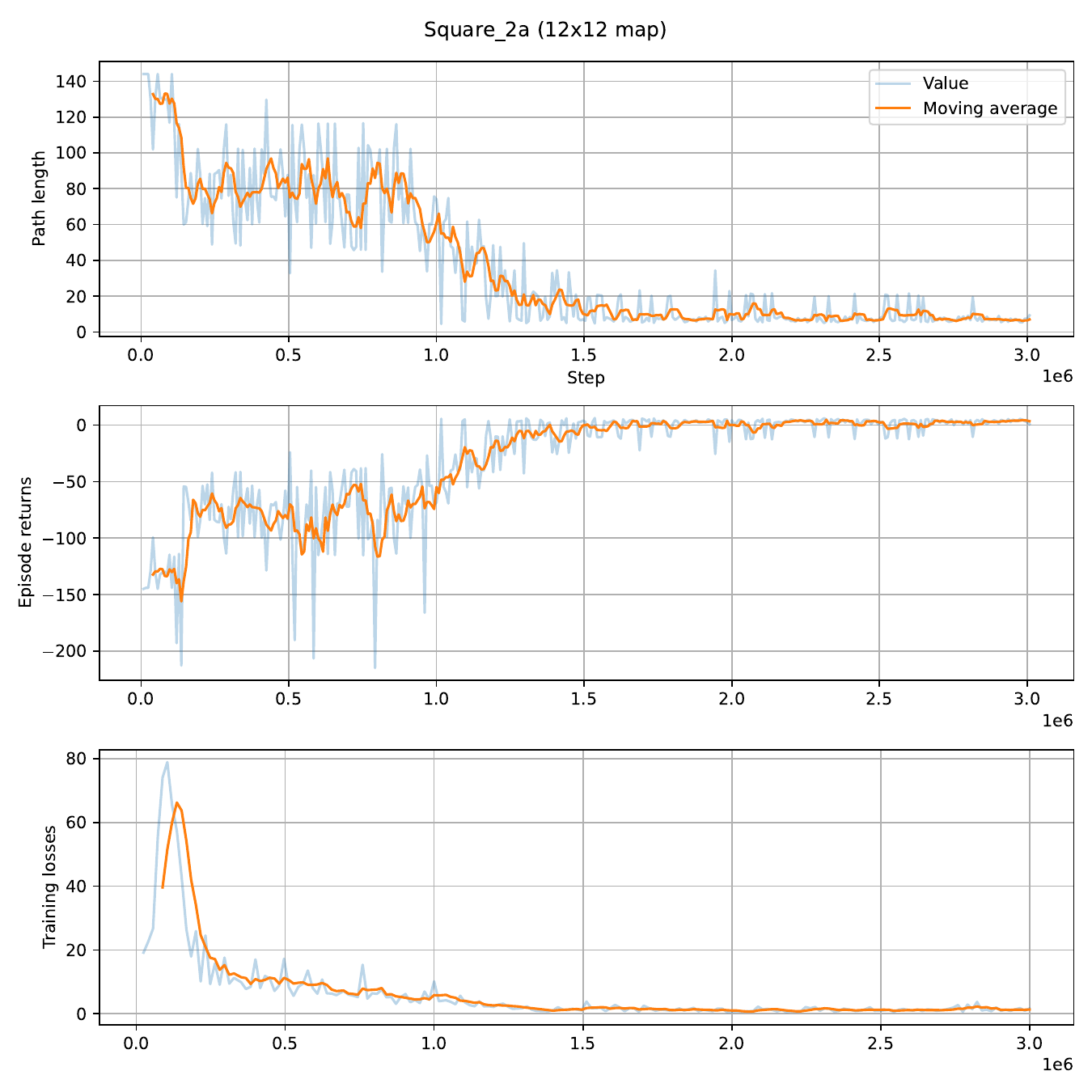}
	\caption{Statistics of the RL training samples.}
	\label{fig:crl}
\end{figure}

\subsection{Detailed Results Additional to Table~\ref{tab:performance}}

Table~\ref{tab:params} shows the detailed parameters for the planners tested in the experiments shown in Table~\ref{tab:performance}.
\begin{enumerate}
	\item ``\#(runs)'' means how many rounds we test each planner to calculate an average.
	\item ``$\epsilon$'' means the ration of randomness associated to the opponents that is assumed by the modelling agent.
	\item ``depth'' means the depth of lookahead search in the corresponding full-width tree search planner.
	\item ``eval\_samples'' means the number of calls to EECBS while evaluating a tree node.
	\item ``backup\_samples'' means the number of samples to perform sampling-based backup.
	\item ``max\_iter'' means the number of simulations performed by the corresponding MCTS planner.
	\item ``select\_samples'' means the number of samples performed at each ``EXP'' node.
\end{enumerate}

\begin{table}[!htpb]
\begin{tabular}{@{}l|lllllll@{}}
\toprule
maps      & \#(runs) & $\epsilon$ & depth & eval\_samples & backup\_samples & max\_iter & select\_samples \\ \midrule
Small2a     & 500      & 7E-04                   & 2     & 10           & exact              & 30      & 50             \\
Square2a  & 1000     & 2E-04                   & 2     & 10           & exact              & 50      & 50             \\
Square4a  & 1500     & 2E-04                   & 1     & 5            & 10             & 60      & 50             \\
Medium20a & 1000     & 8E-05                   & /     & 5            & /              & 80      & 80             \\
Large50a    & 500      & 2E-05                   & /     & 2            & /              & 100     & 125   \\ \bottomrule       
\end{tabular}
\caption{Detailed Parameters for Table \ref{tab:performance}.}
\label{tab:params}
\end{table}

In table~\ref{tab:performance}, we have shown the path lengths penalized by collisions.
Here by Figure~\ref{fig:detail_samll}~-~\ref{fig:detail_large}, we show the raw path lengths as well as collision ratios.
As one can see,
\begin{enumerate}
	\item \textit{Safe-agents} and their \textit{Enhanced} versions lead to possibly longer raw paths but lower chance of collisions, as they may easily get stuck but can avoid most of the collisions.
	\item Compared to each planner by a vanilla oracle, the improved version by tree search usually leads to slightly longer paths but significantly lower collision ratios.
	\item For \textbf{malicious} opponents, it is sometimes better if the modelling agent sticks to the initial uniform belief and does not update it, as in this case belief modelling is ``severely attacked'' by their chasing behavior. 
\end{enumerate}

\begin{figure}[!ht]
	\flushleft
	\includegraphics[height=35mm]{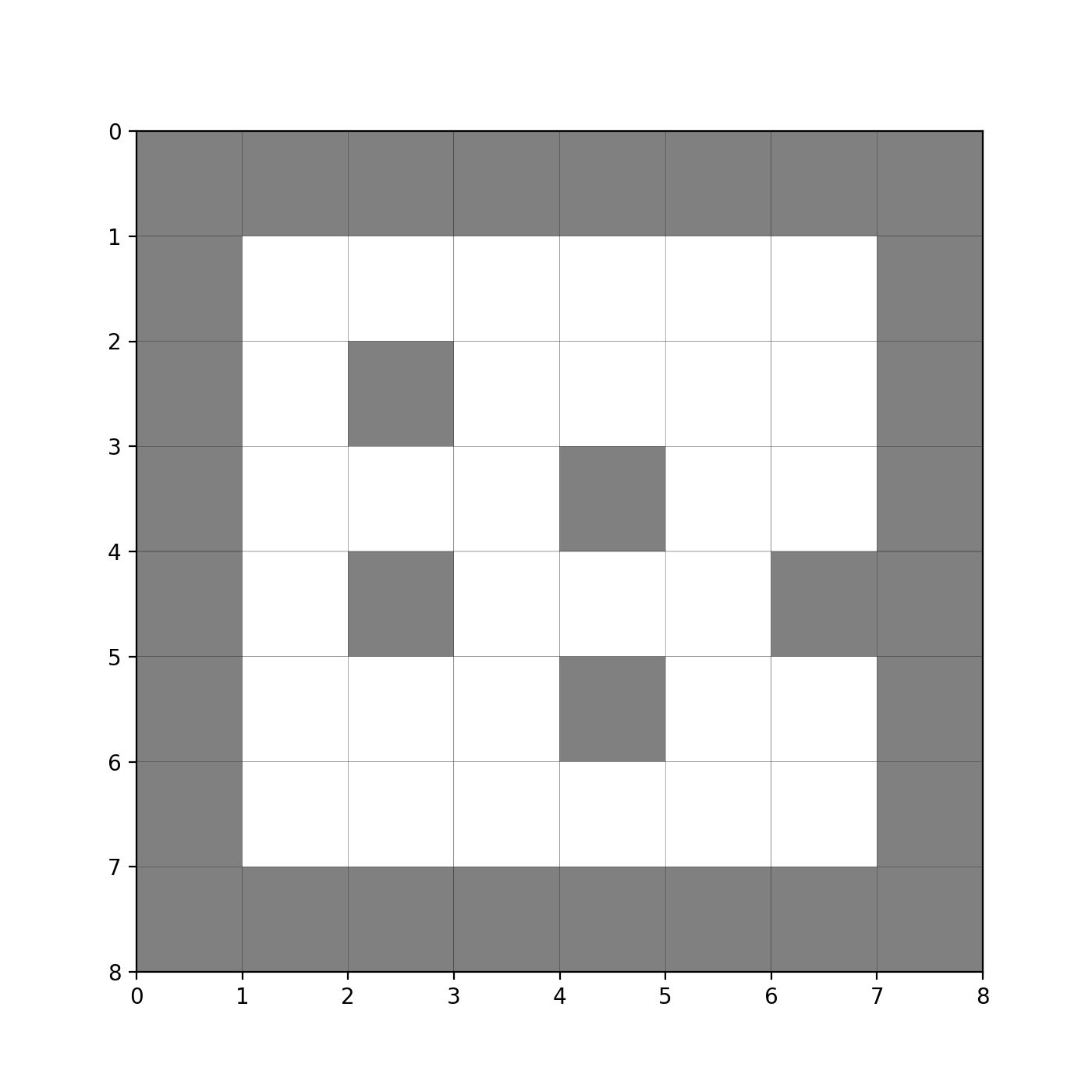}
	\hspace{-3mm}
	\includegraphics[height=35mm]{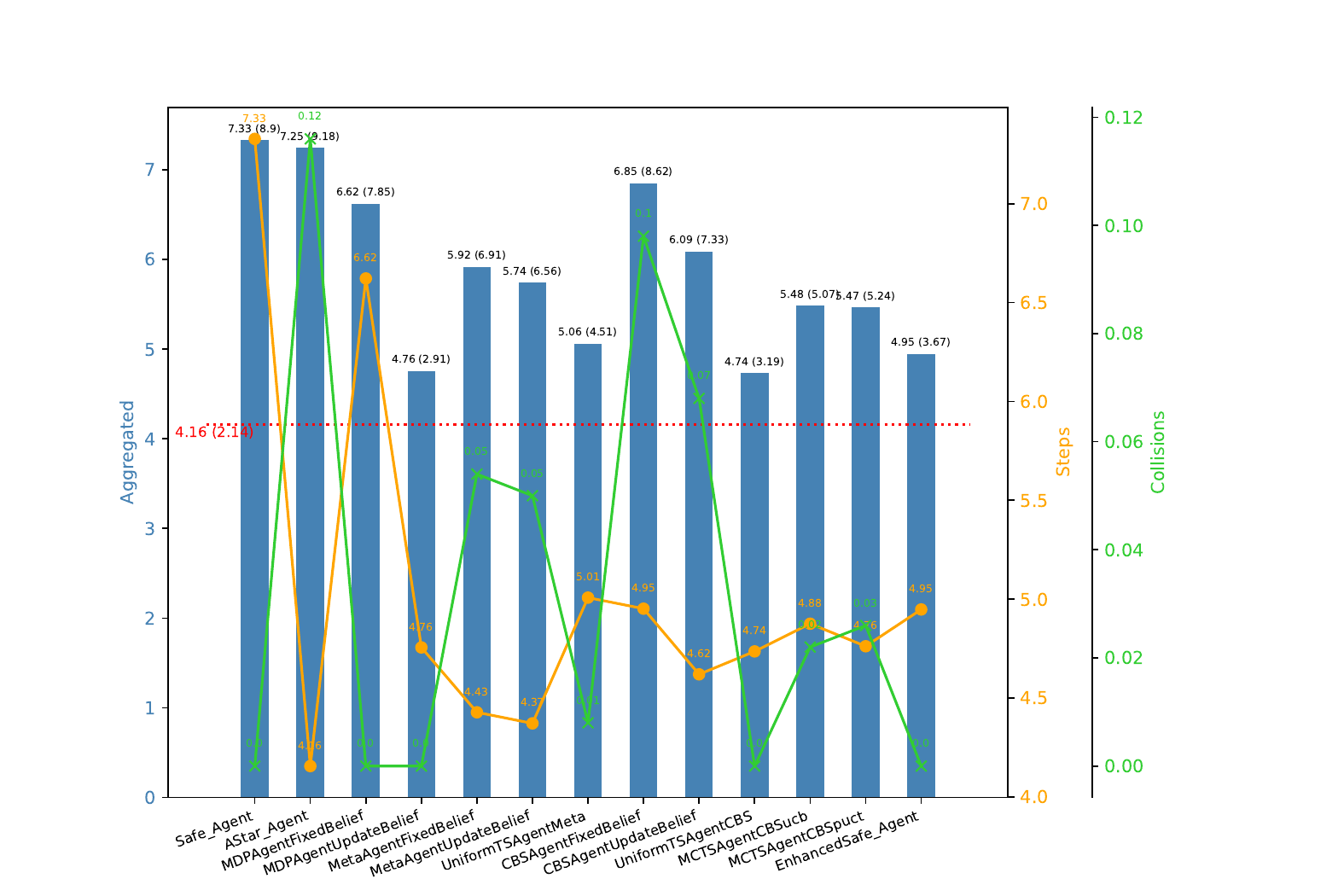}
	\hspace{-8mm}
	\includegraphics[height=35mm]{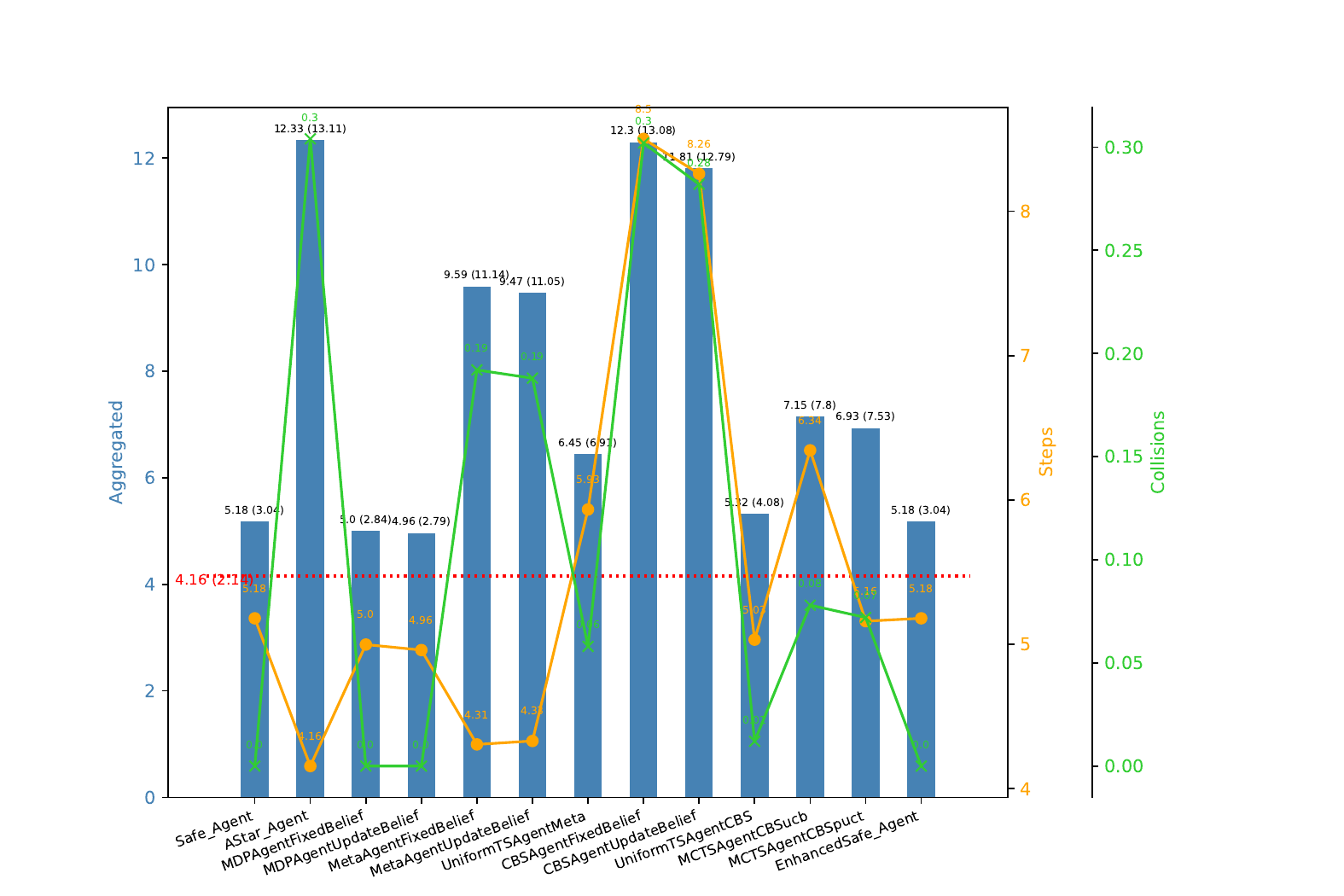}
	\hspace{-8mm}
	\includegraphics[height=35mm]{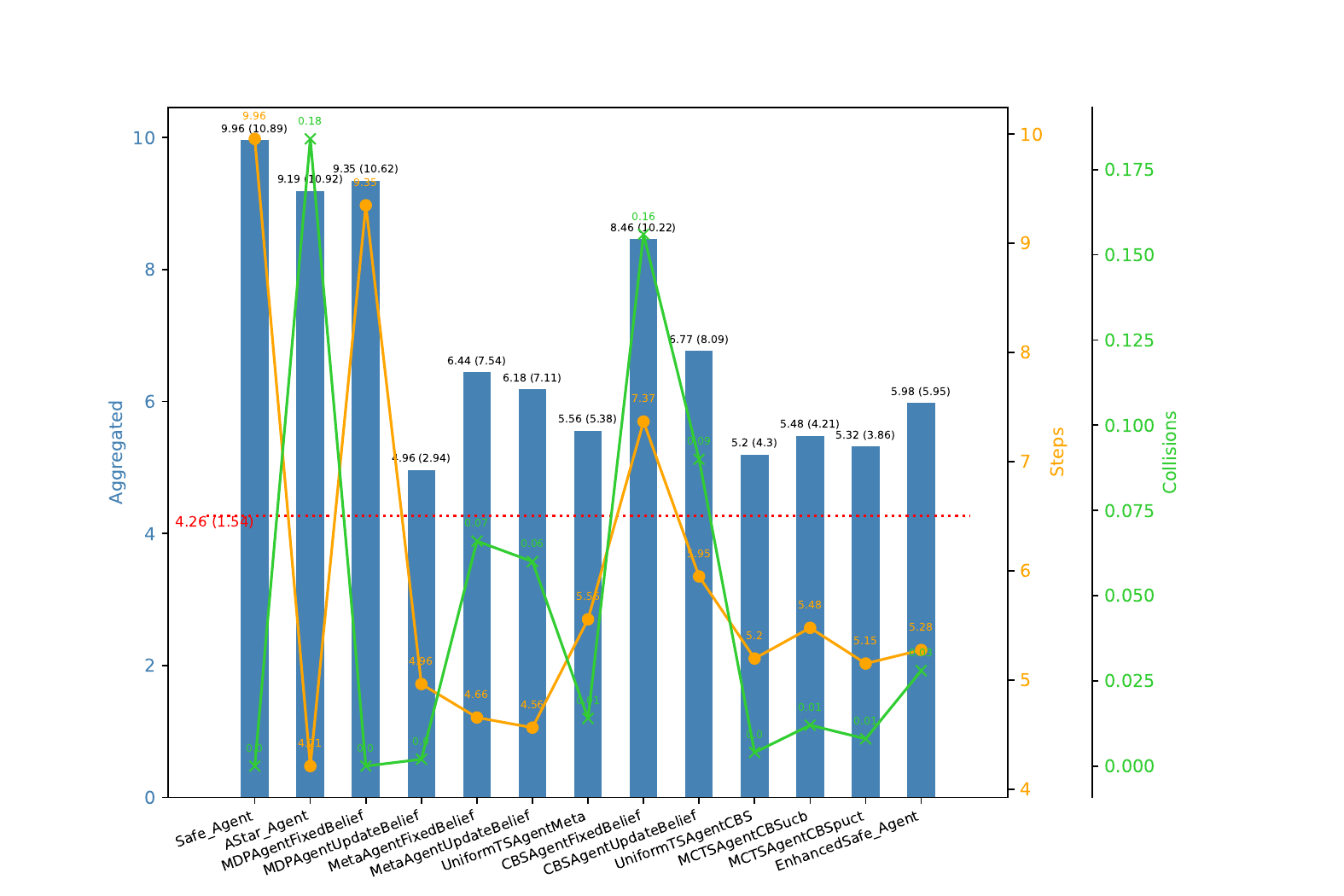}
	\caption{Detailed experiments for ``Small2a'' configurations.}
	\label{fig:detail_samll}
\end{figure}

\begin{figure}[!ht]
	\flushleft
	\includegraphics[height=35mm]{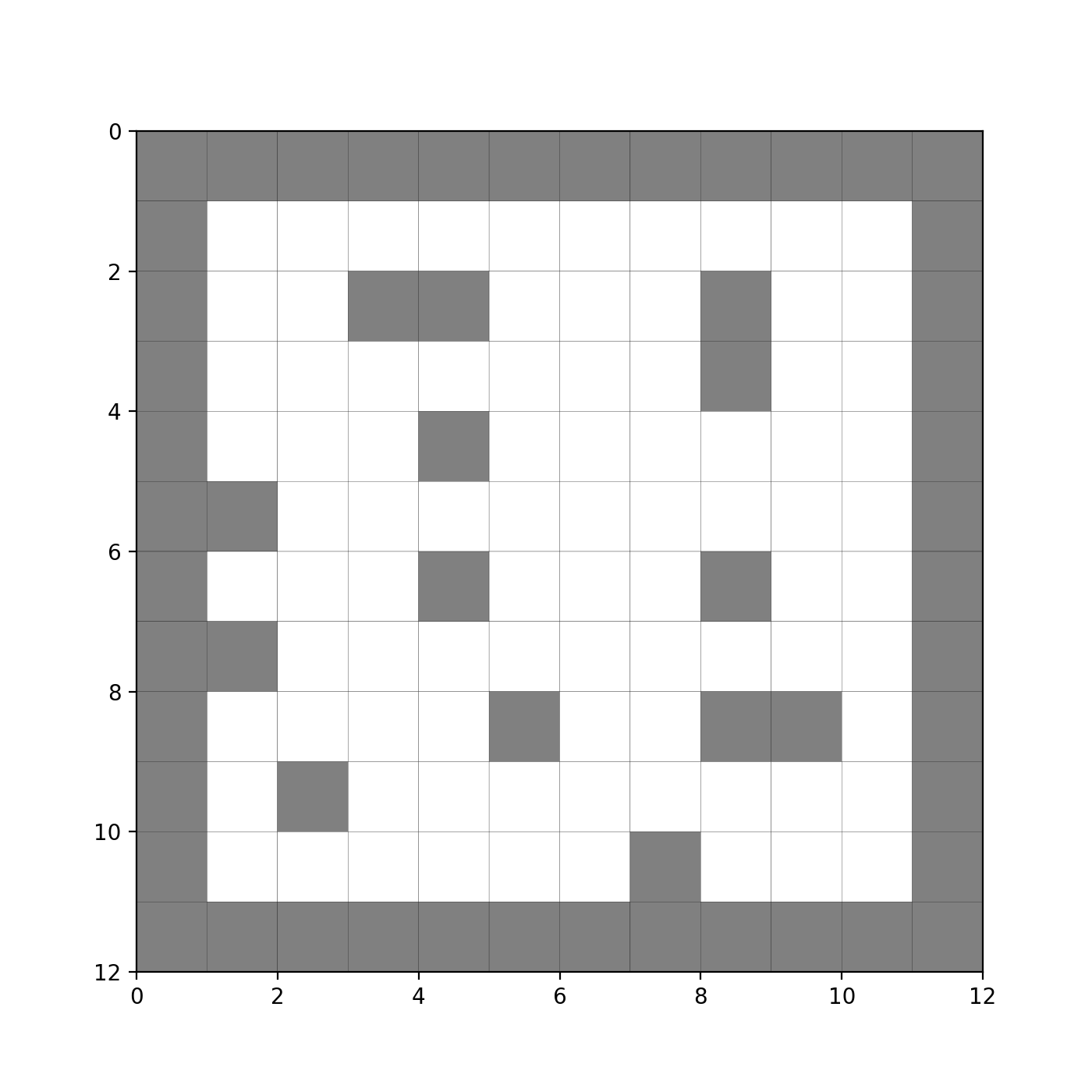}
	\hspace{-3mm}
	\includegraphics[height=35mm]{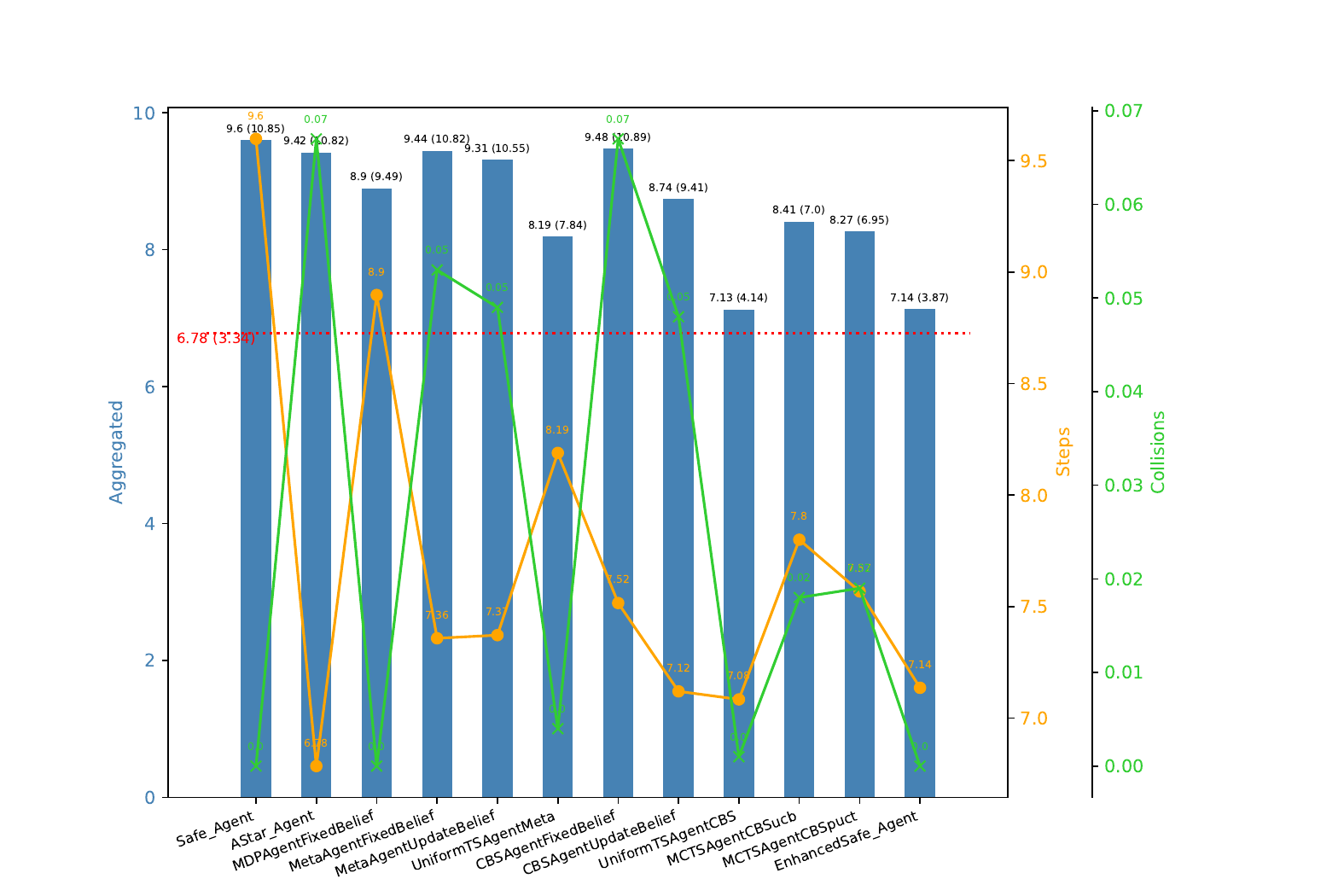}
	\hspace{-8mm}
	\includegraphics[height=35mm]{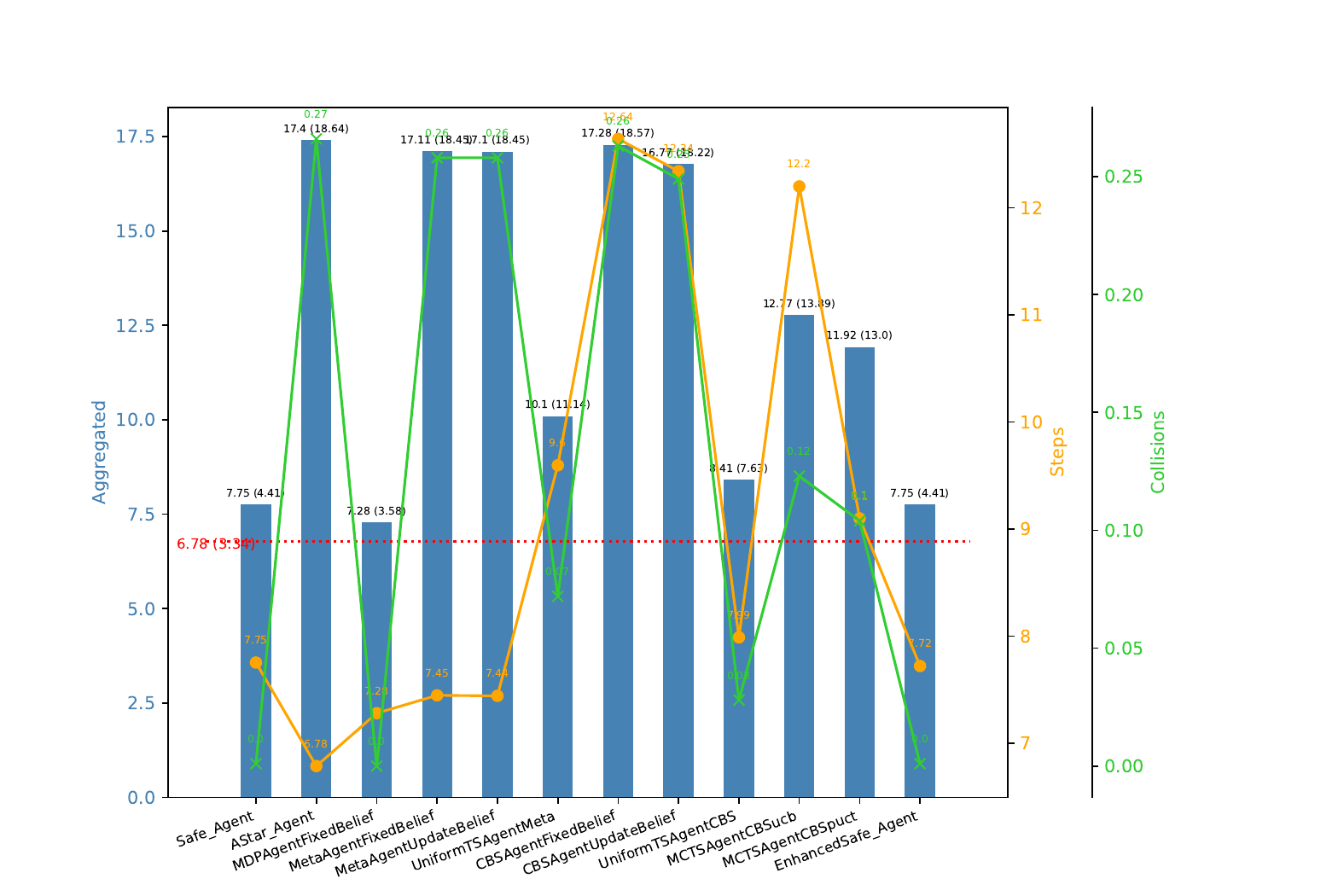}
	\hspace{-8mm}
	\includegraphics[height=35mm]{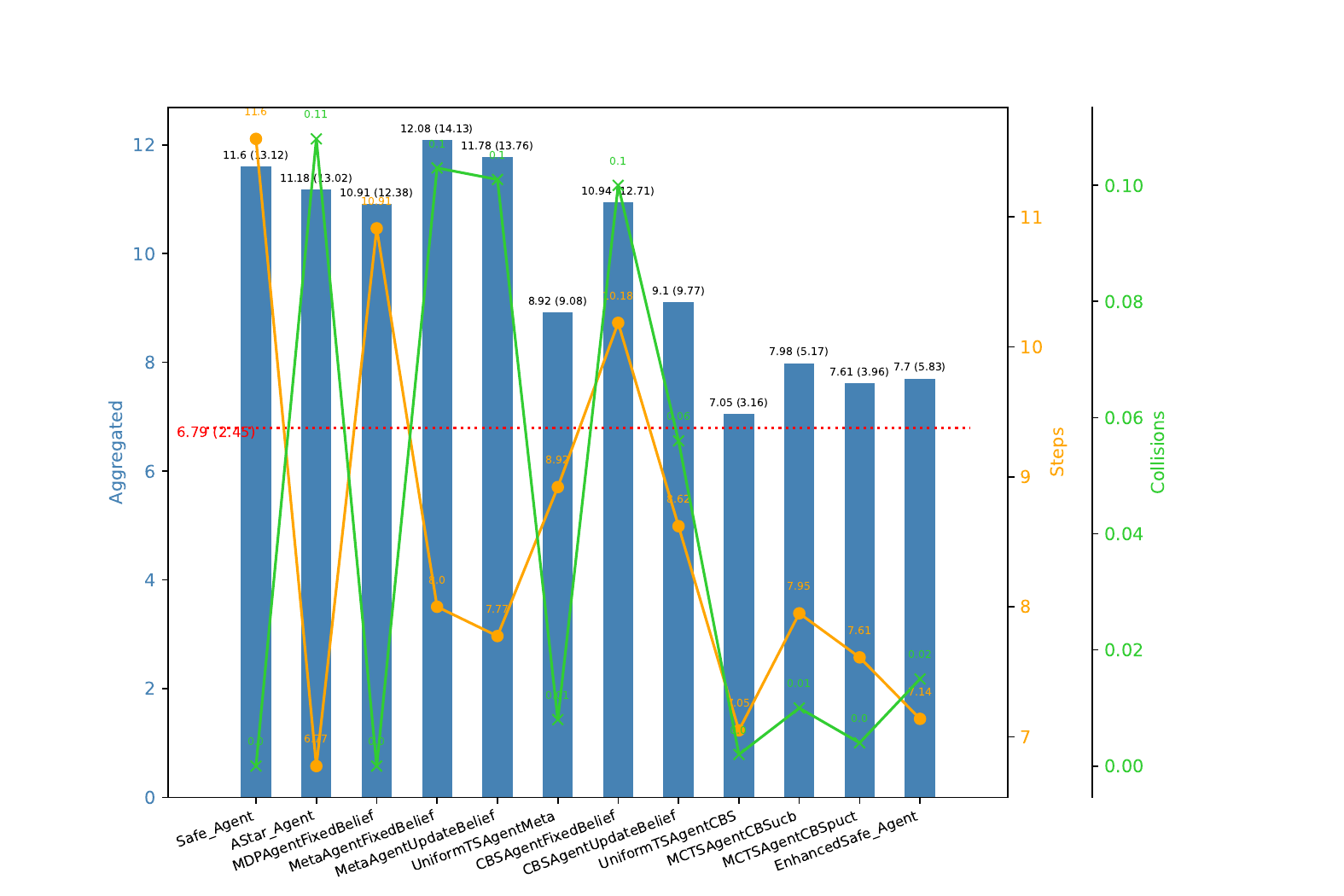}
	\caption{Detailed experiments for ``Square2a'' configurations.}
\end{figure}

\begin{figure}[!ht]
	\flushleft
	\includegraphics[height=35mm]{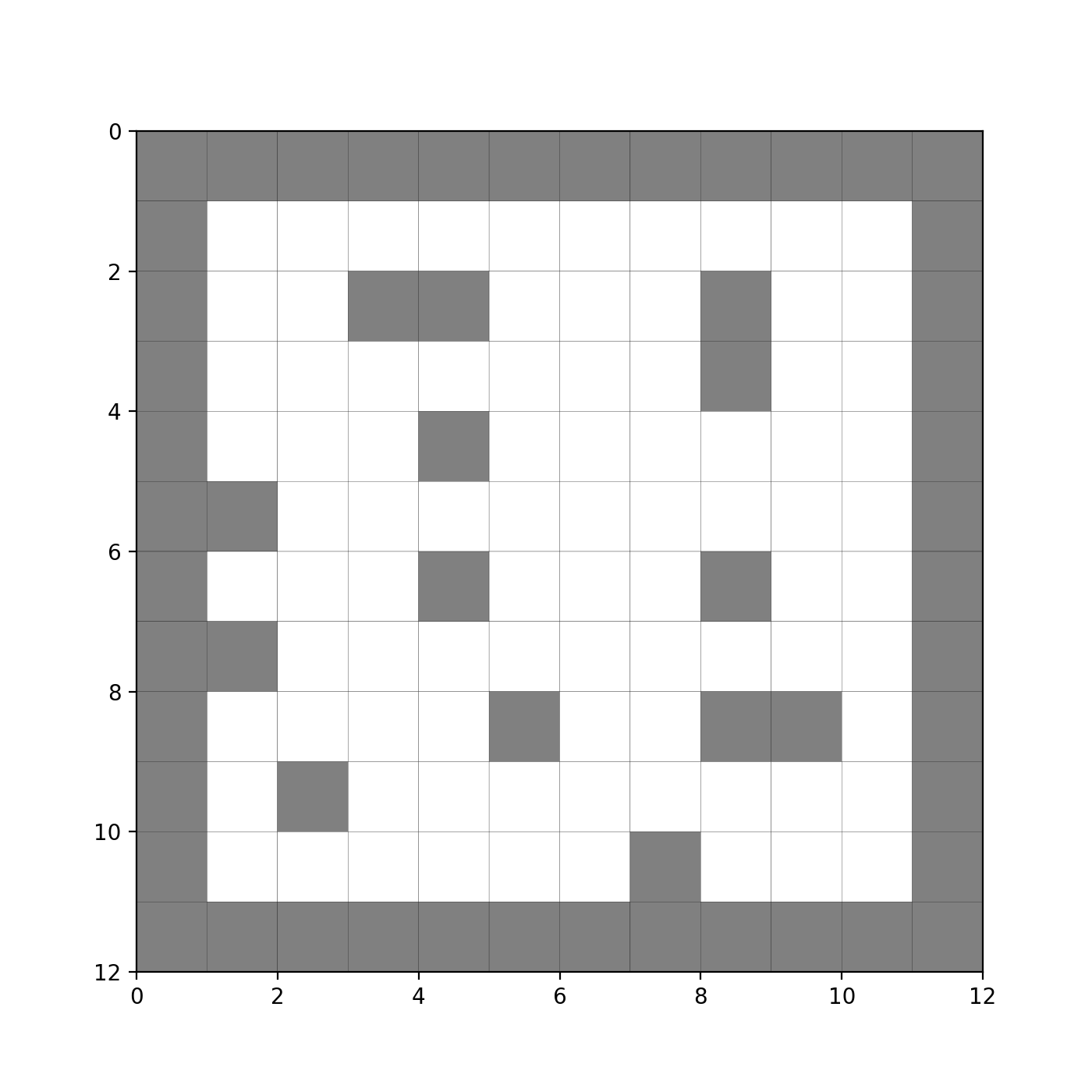}
	\hspace{-3mm}
	\includegraphics[height=35mm]{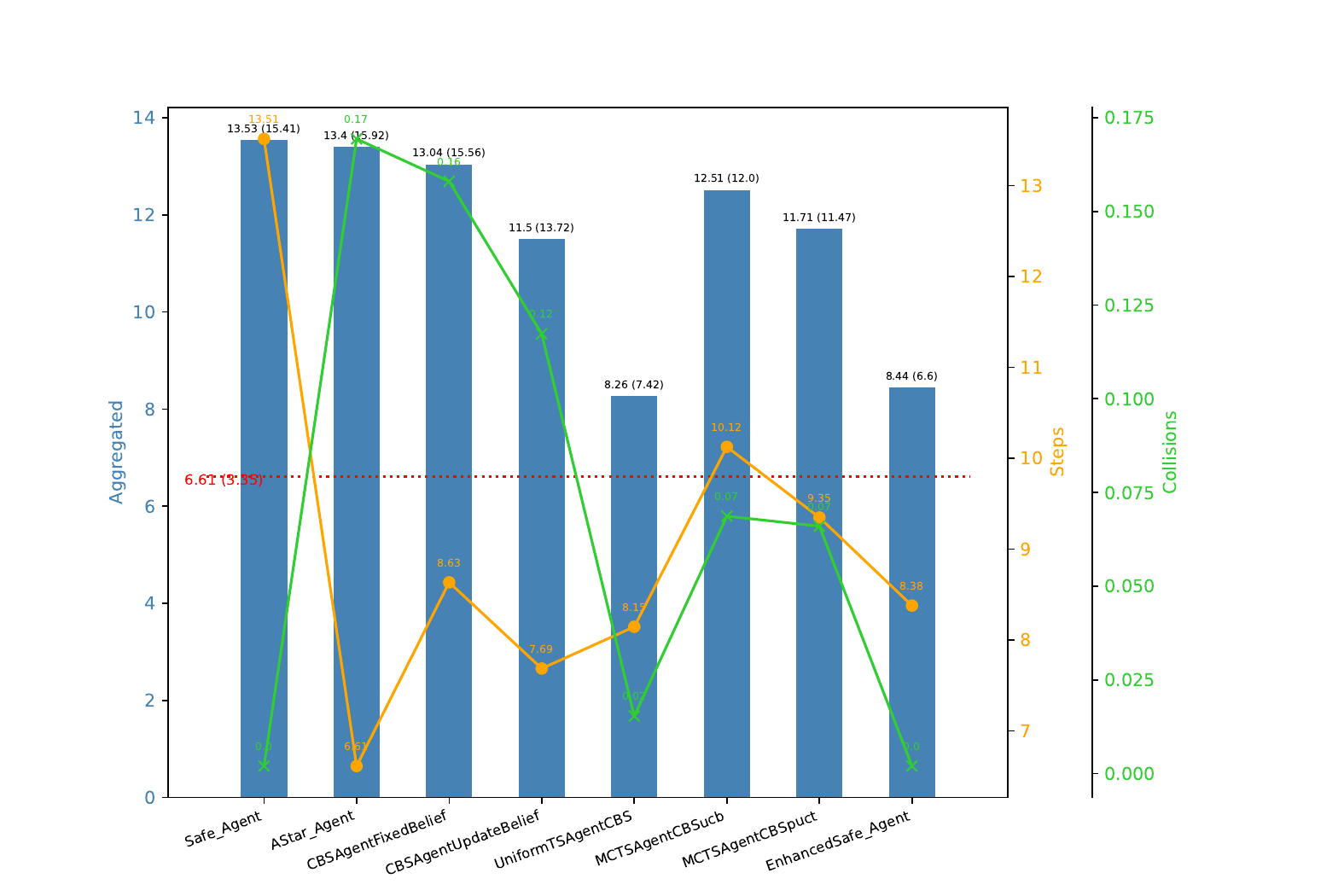}
	\hspace{-8mm}
	\includegraphics[height=35mm]{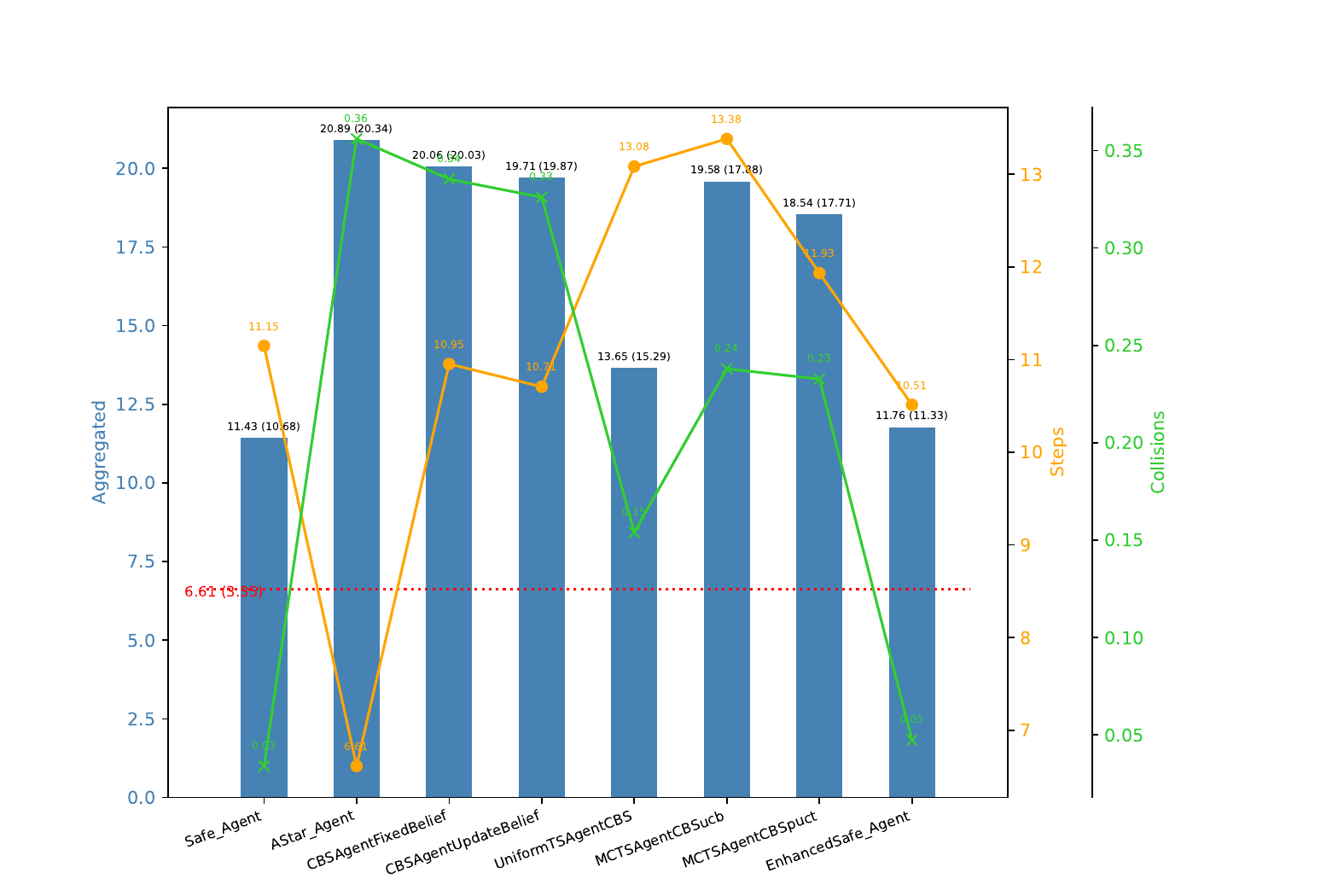}
	\hspace{-8mm}
	\includegraphics[height=35mm]{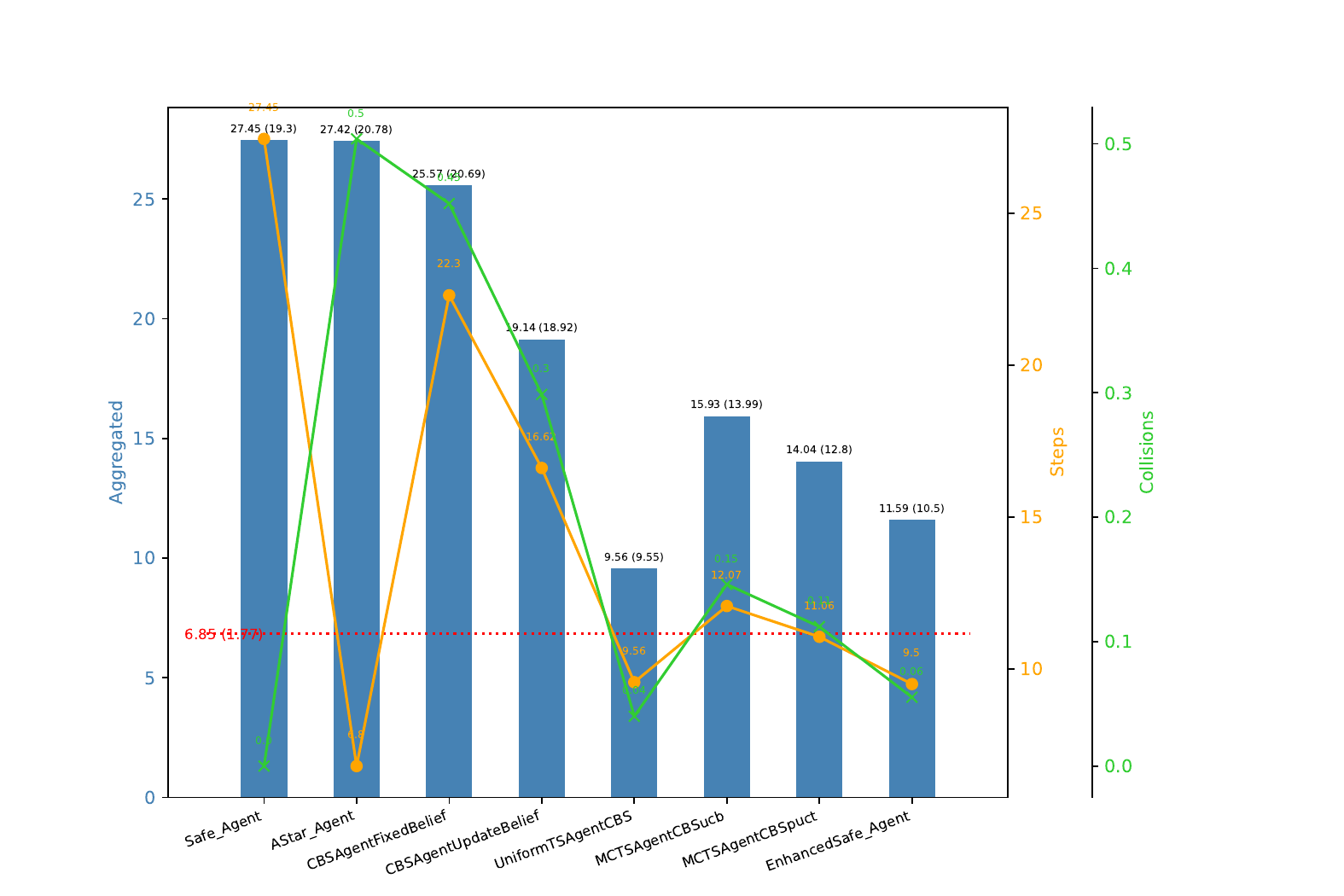}
	\caption{Detailed experiments for ``Sqaure4a'' configurations.}
\end{figure}

\begin{figure}[!ht]
	\flushleft
	\includegraphics[height=35mm]{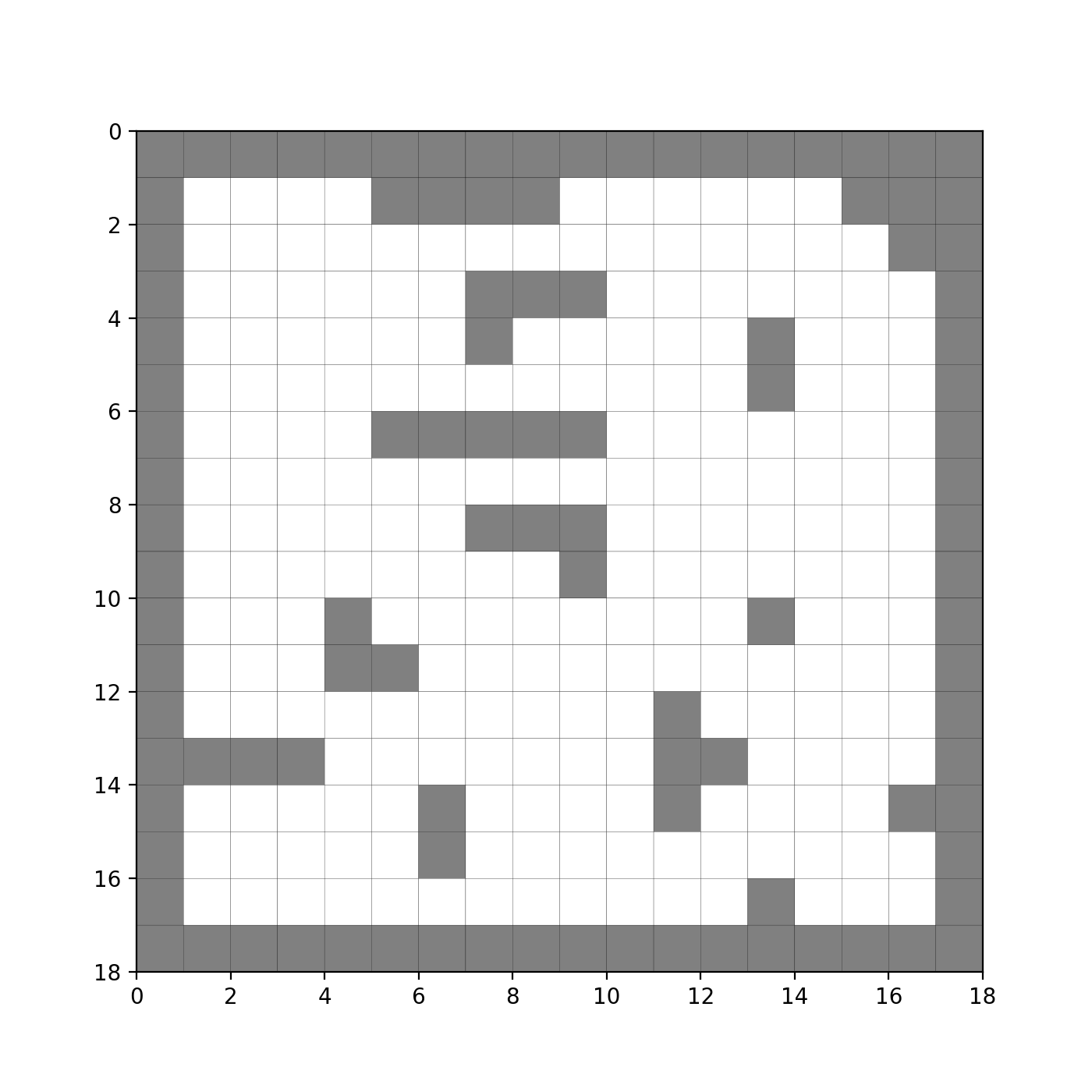}
	\hspace{-3mm}
	\includegraphics[height=35mm]{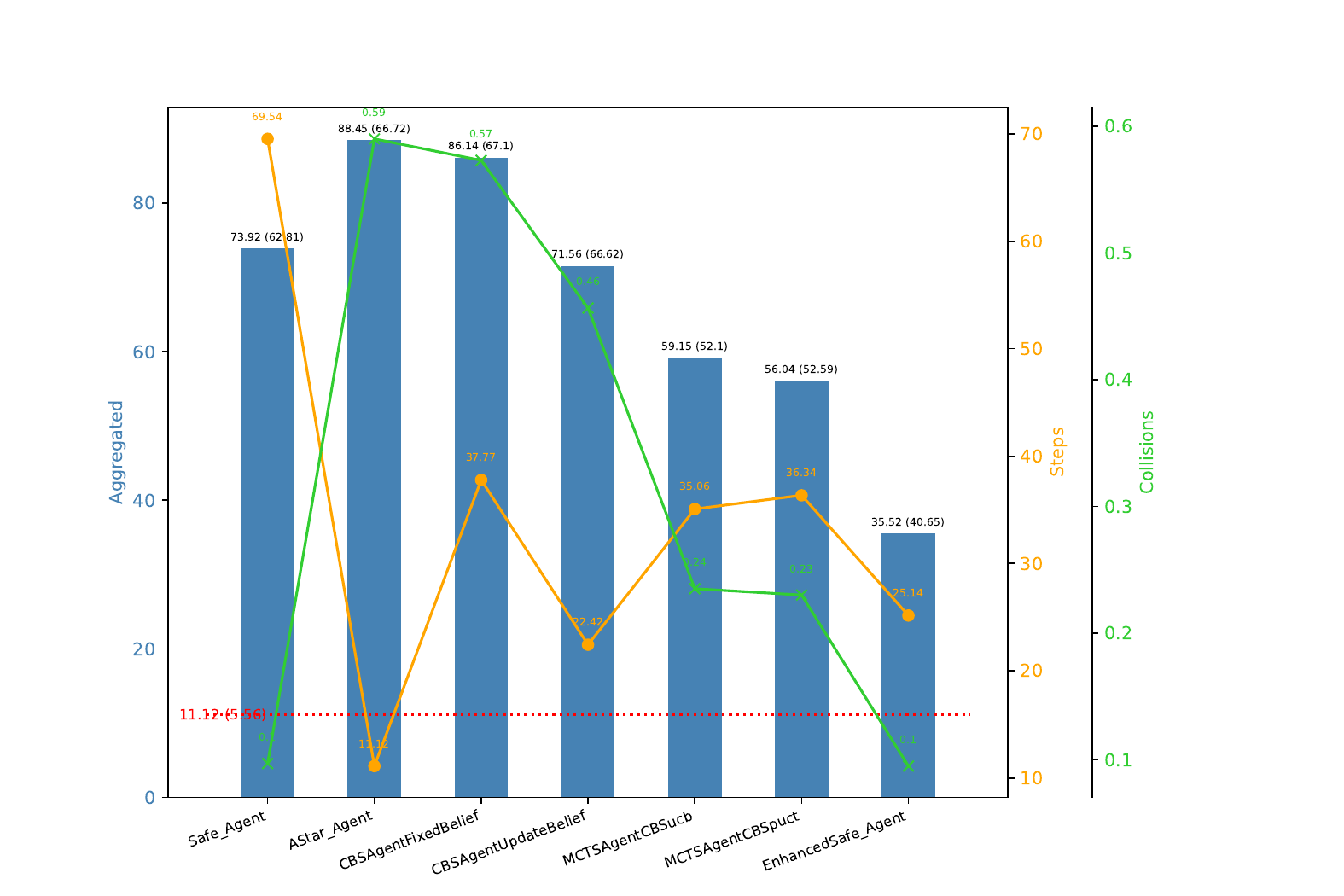}
	\hspace{-8mm}
	\includegraphics[height=35mm]{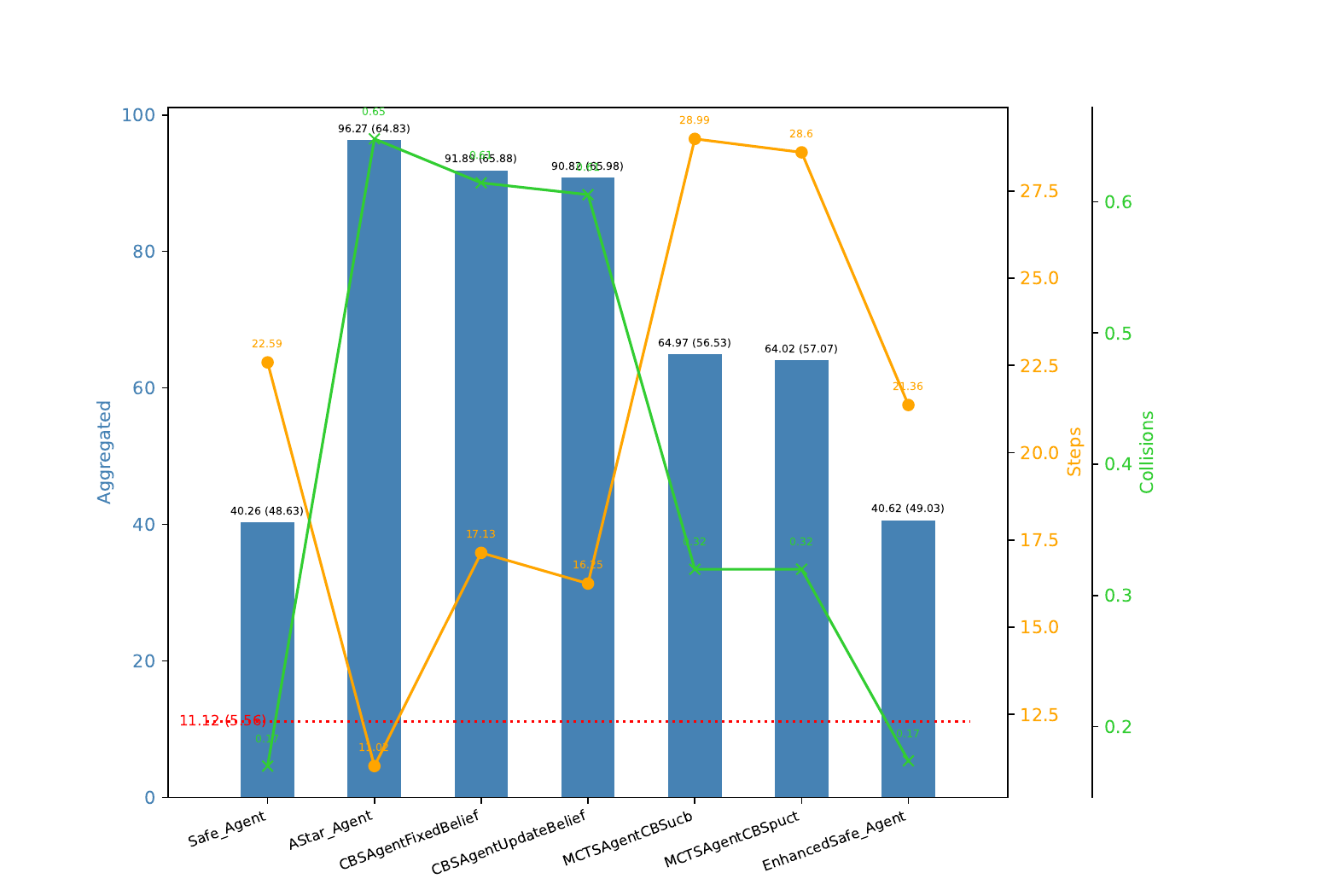}
	\hspace{-8mm}
	\includegraphics[height=35mm]{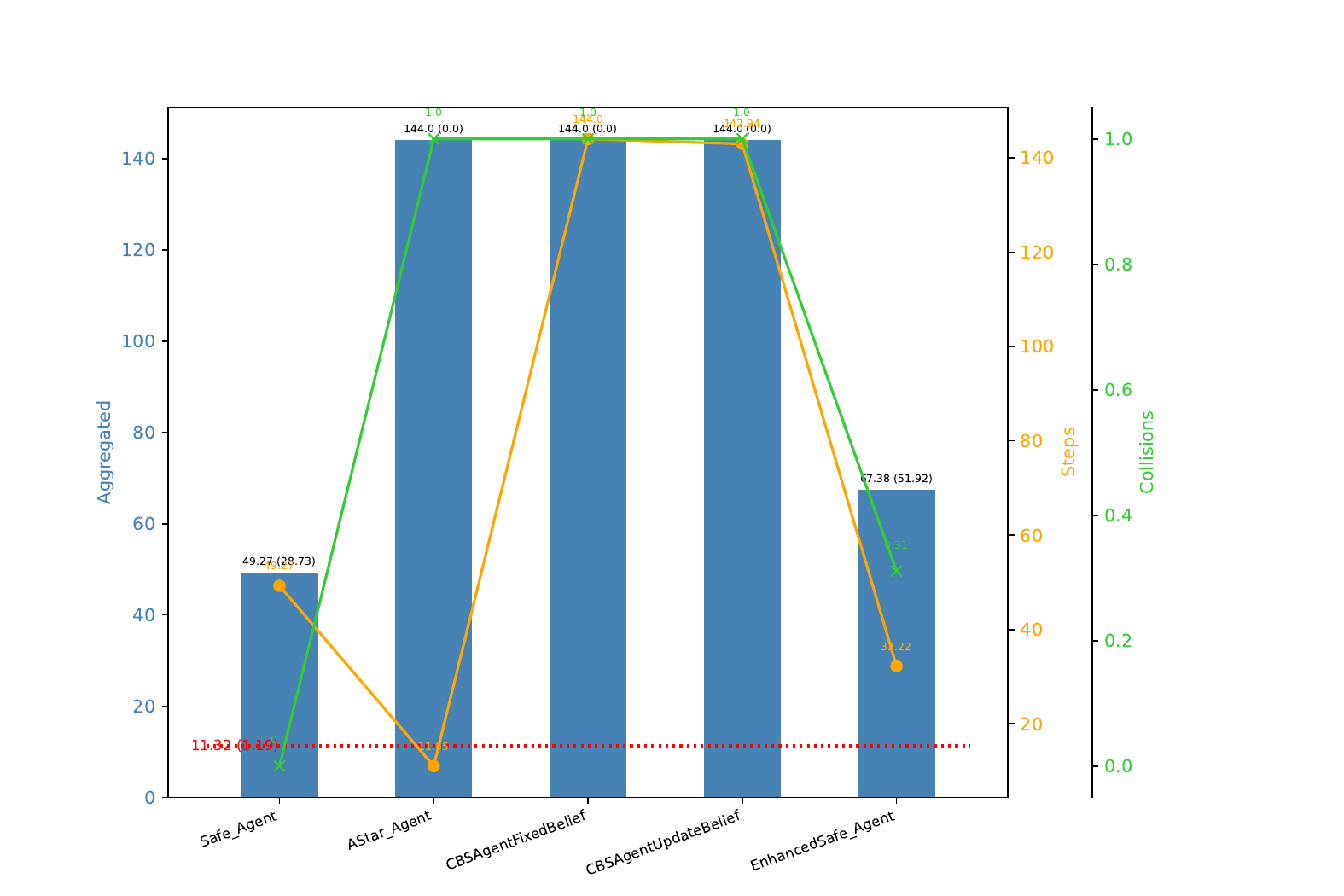}
	\caption{Detailed experiments for ``Medium20a'' configurations.}
\end{figure}

\begin{figure}[!ht]
	\flushleft
	\includegraphics[height=35mm]{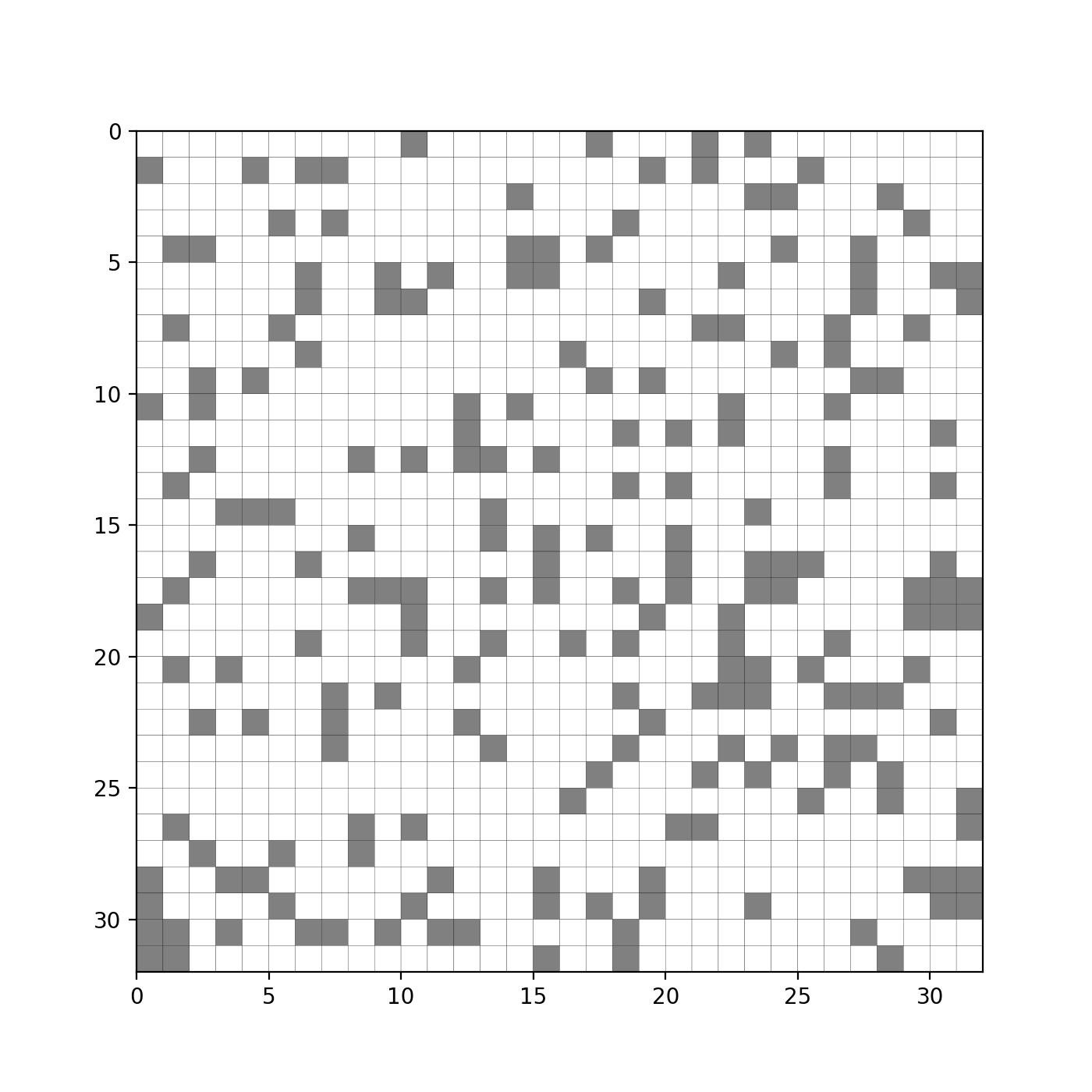}
	\hspace{-3mm}
	\includegraphics[height=35mm]{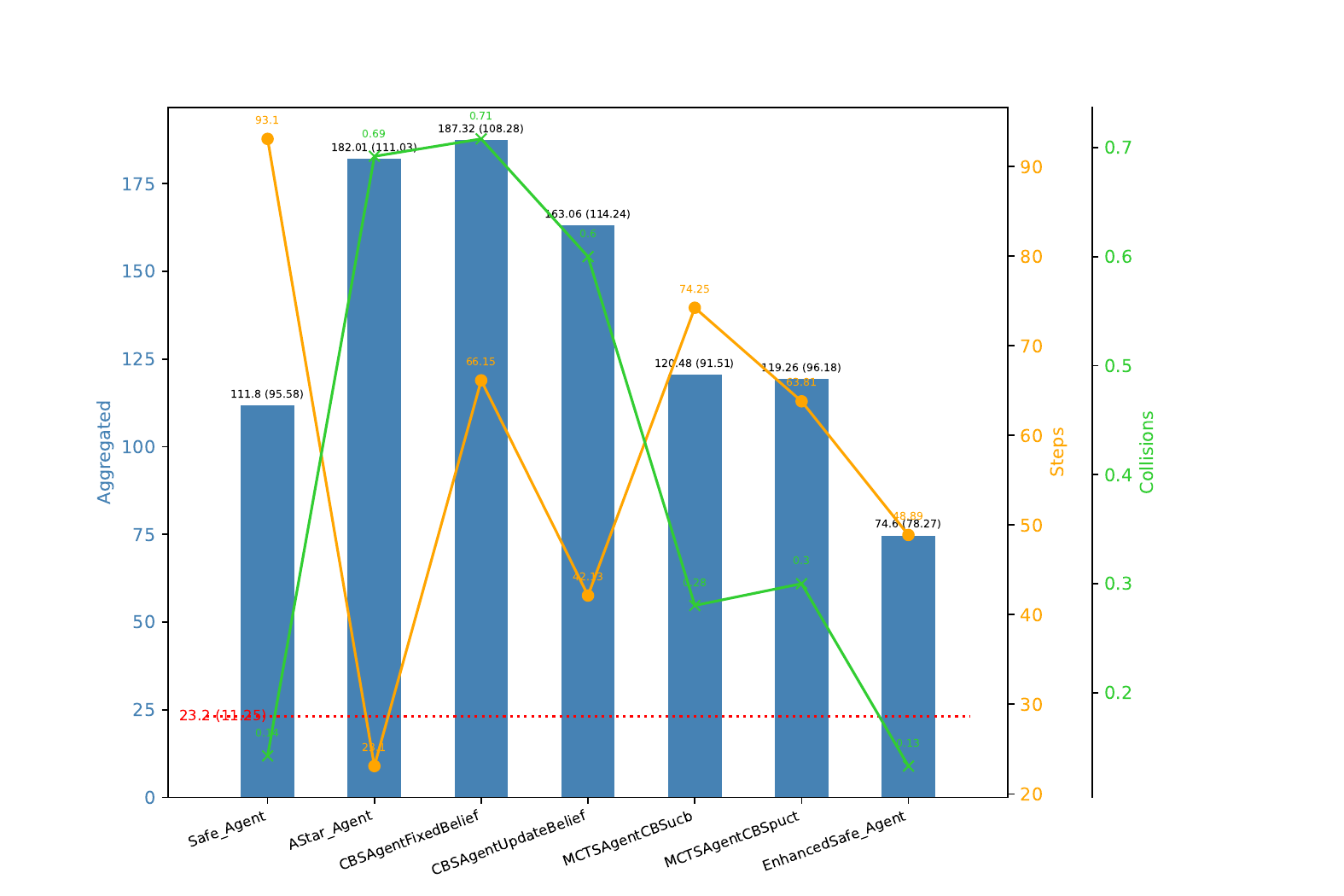}
	\hspace{-8mm}
	\includegraphics[height=35mm]{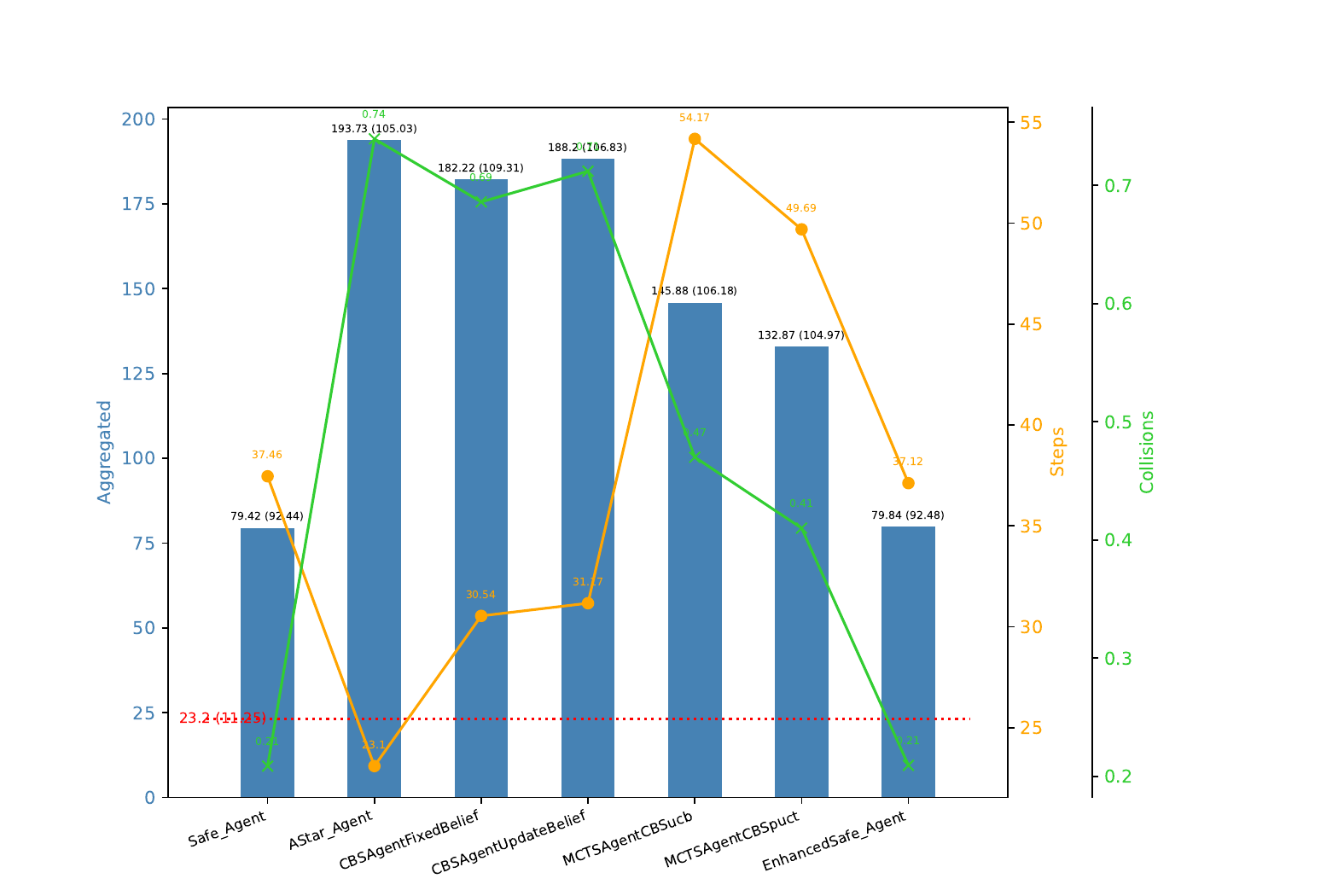}
	\hspace{-8mm}
	\includegraphics[height=35mm]{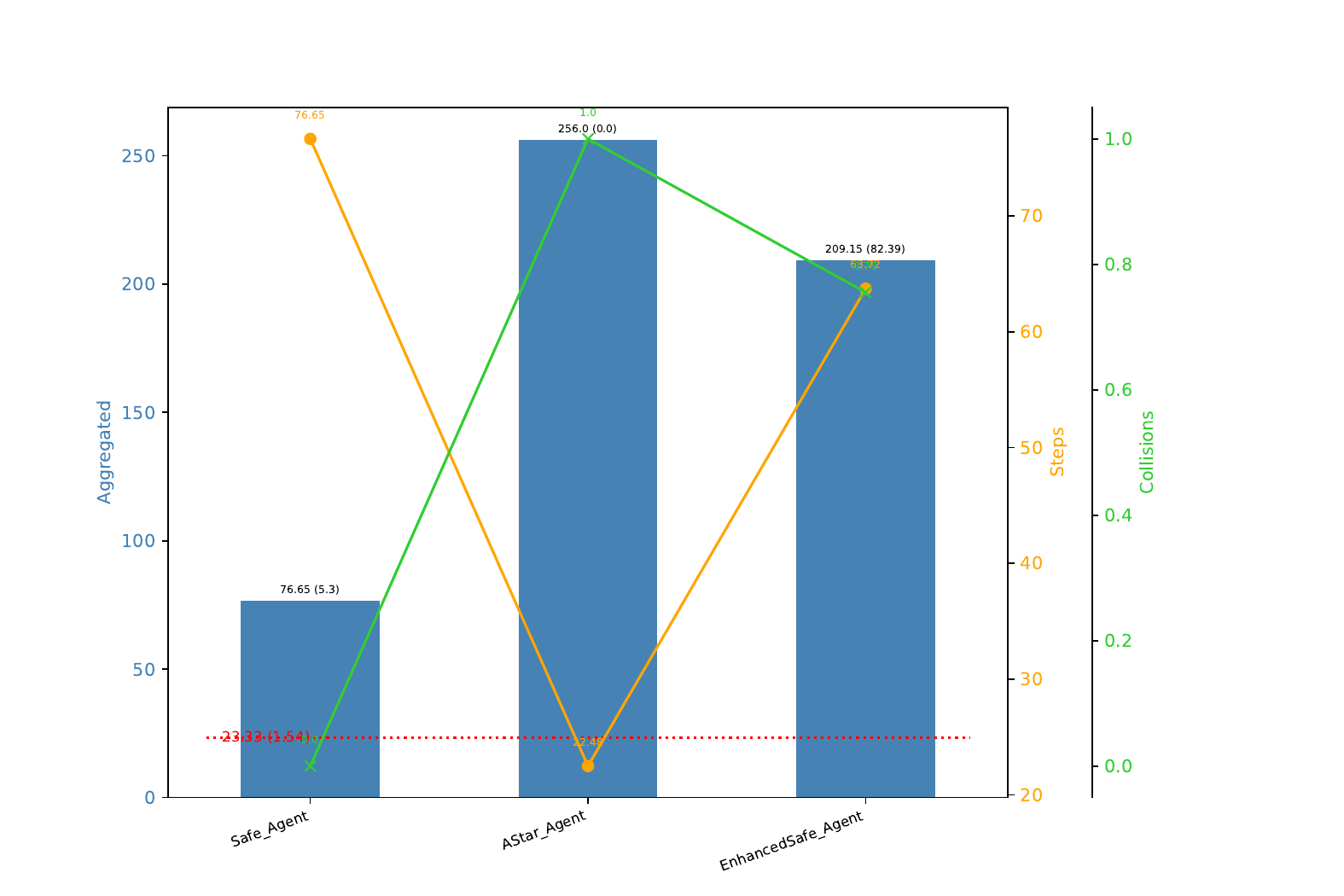}
	\caption{Detailed experiments for ``Large50a'' configurations.}
	\label{fig:detail_large}
\end{figure}

\section{Case Studies}
\label{apd:case_study}

We here present a few additional planners in Table~\ref{tab:additional_planner}, as they are computationally expensive and therefore not suitable for a multi-round testing.
Instead, in this section we show a few case studies to help one understand the capability of each planner.

\begin{table}[!htpb]
\begin{tabular}{@{}c|cccccccc@{}}
\toprule
\textbf{Planners}  & \textbf{n}   & \textbf{m} & $\mathbf\beta$ & \textbf{Collision penalty} & \textbf{\textsc{Eval}$_i$} & \textbf{Backup} & \textbf{Lookahead} & \textbf{Need replanning}  \\ \midrule
POMDP        & $\infty$     &            & 1              & $< \infty$                 &                                             &                 &                    &                           \\
QMDP         & 0            & 0          & 1              & $< \infty$                 & QMDP                                        &                 &                    &                           \\
UniformTS-MDP      & $(0,\infty)$ & $\infty$   & 1              & $< \infty$                 &                                             & exact           & full-width         & \checkmark \\
UniformTS-reactive & $(0,\infty)$ & 0   & 1              & $< \infty$                 & Euclidean distance                          & exact           & full-width         & \checkmark \\ \bottomrule
\end{tabular}
\caption{Additional planners.}
\label{tab:additional_planner}
\end{table}

\subsubsection*{Case 1}
It is predictable that by Bayesian-like belief update, the belief held by the modelling agent will converge to the underlying ground truth, if given many enough rounds.
The challenging case is when the agent is under a rather small grid world and may not have too many steps to go. In this case, the modelling agent has to plan over hypothetical beliefs upfront, instead of updating her belief per move.
As shown in Figure~\ref{fig:small_compare}, the POMDP agent only needs 5 steps, while the QMDP agent needs 6 steps and both MDP agents (with and without belief update) need the longest 7 steps.

\begin{figure}[!ht]
	\flushleft
	\includegraphics[height=100mm]{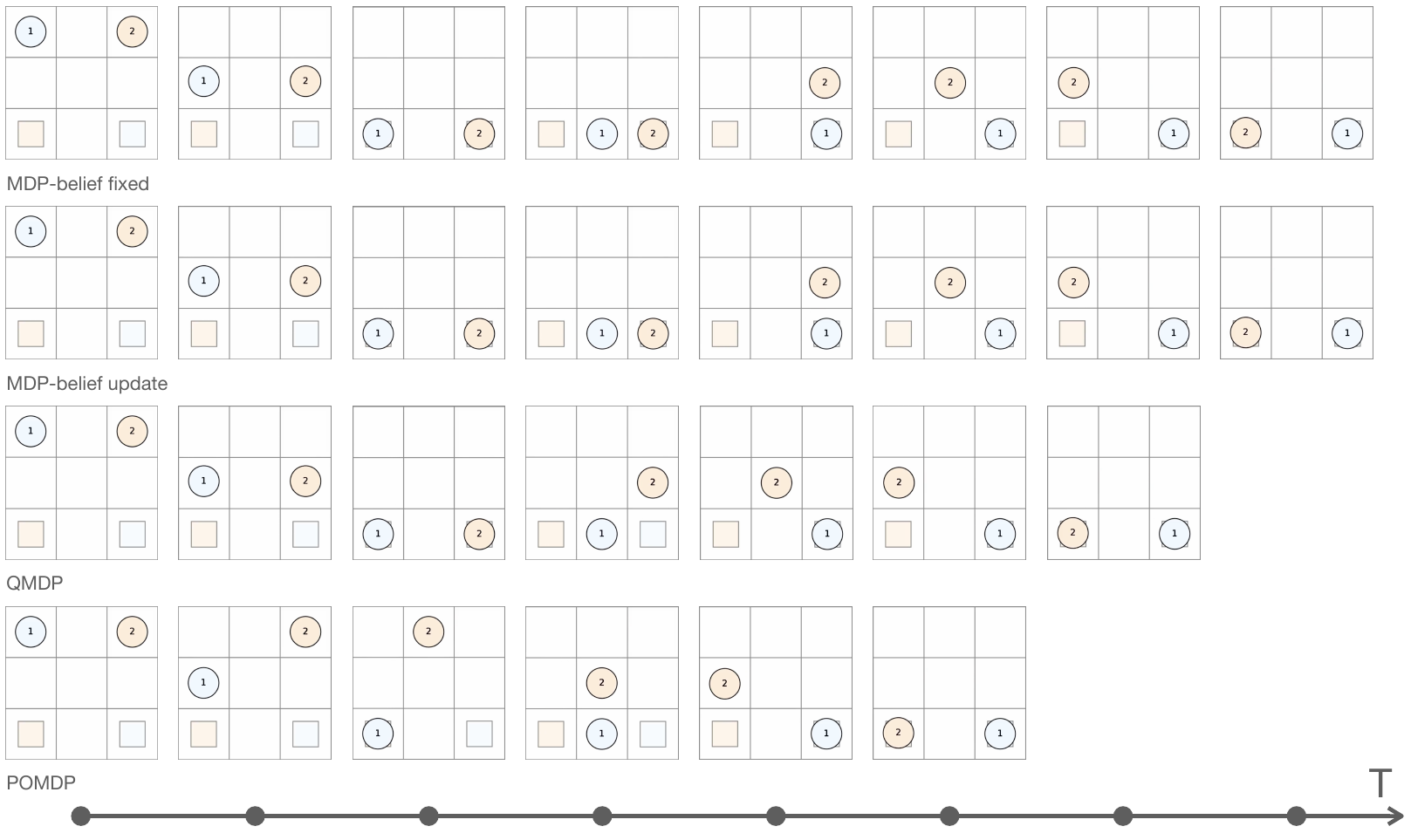}
	\caption{Detailed paths planned in a tiny situation (3x3 map two agents): agent 1 is a shorest-path agent while agent 2 is the modelling agent. Agent 1 is using a different shortest path algorithm from the way how agent 2 does in belief modelling. Squares with the corresponding colors are the respective goals.}
	\label{fig:small_compare}
\end{figure}

\subsubsection*{Case 2}
Here we show the necessity of belief update in Figure~\ref{fig:noupdate_stuck}.
Even if the modelling agent is capably of planning over states for infinite horizons, she will still easily get stuck under such a crowded environment if she does not update her belief.
With belief update, as long as she has observed the opponent has stayed in a position for a sufficiently long time, she can therefore infer that the probability of getting hit becomes negligible.


\begin{figure}[!ht]
	\flushleft
	\includegraphics[height=55mm]{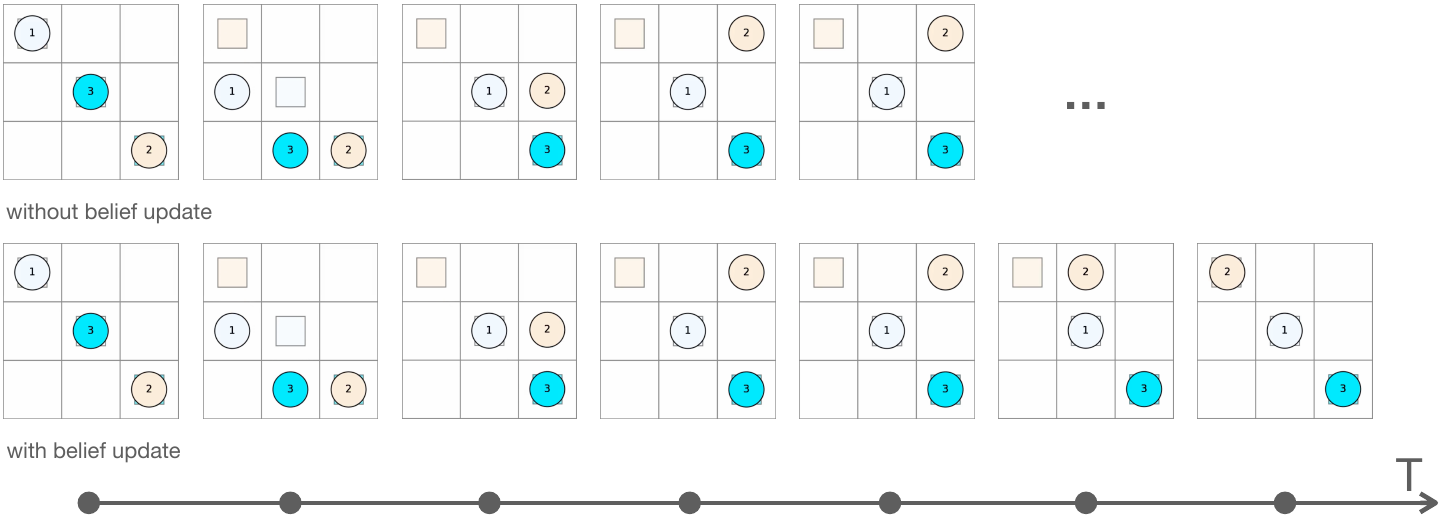}
	\caption{Different results by MDP agents with and without belief update. Agent 2 is the modelling agent using MDP planners, while agent 1 and 3 are two shortest-path agents using different algorithms. Squares with the corresponding colors are the respective goals.}
	\label{fig:noupdate_stuck}
\end{figure}

\subsubsection*{Case 3}
Figure~\ref{fig:probe} illustrate a special type of actions, what we term as \textit{probing actions}.
A probing action is such an action by which the modelling agent gains more information about the opponents without sacrificing her own future return.
In both cases, agent 2 is at C5 when agent 1 is at D1.
For agent 2, going down is the only action that leads to shortest paths, but at the same time of her taking this action, the action made by agent 1 will reveal more information about her underlying goal. More specifically, in the right figure, when agent 2 proceeds to D5 by this probing action, she realizes that agent 1's goal must be below row E and hence decides to go left for the next step.
In contrast, the planners used in the left figure are not able to utilize the probed information. 

\begin{figure}[!ht]
	\centering
	\includegraphics[width=105mm]{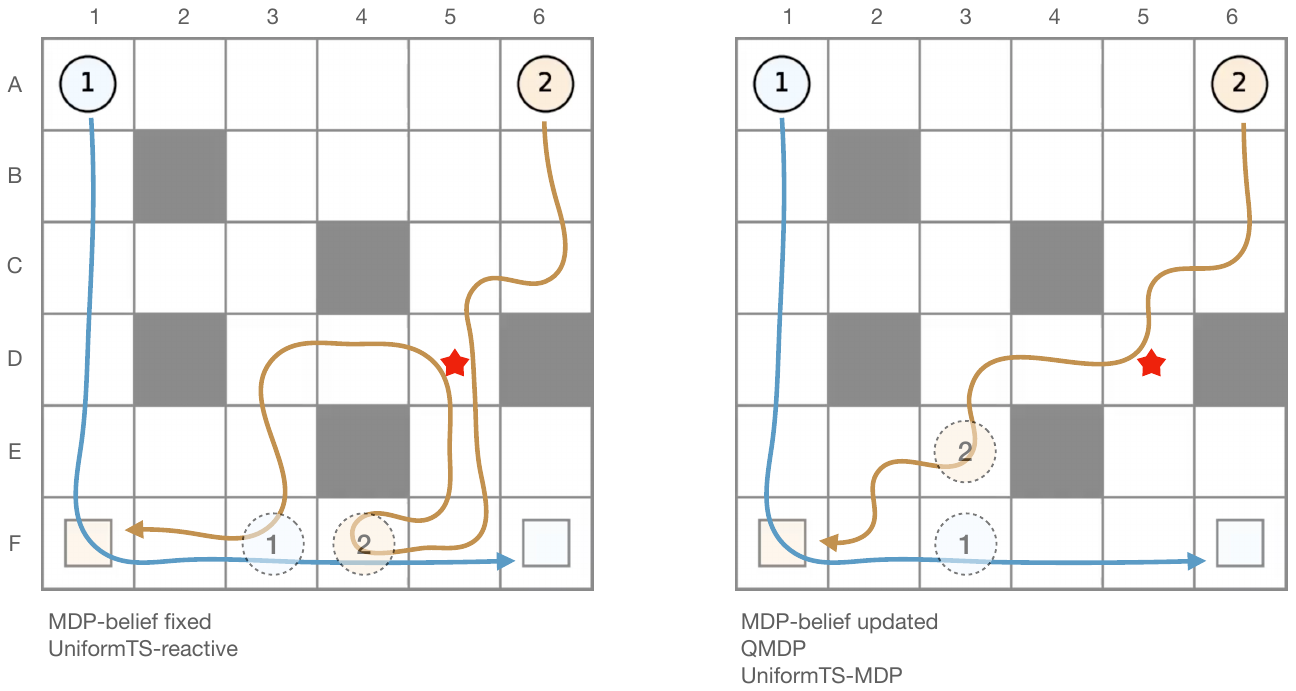}
	\caption{Agent 2 is the modelling agent using any of the planners noted at the lower left corner of each sub-figure, while agent 1 is a shortest-path agent. Squares with the corresponding colors are the respective goals.}
	\label{fig:probe}
\end{figure}

\subsubsection*{Case 4} Figure~\ref{fig:probe} presents a situation where the modelling agent will sooner or later get stuck if she only does myopic planning like \textit{safe-agents}. 
The enhanced version can later resolve it but has to wait until the others stop moving.
The full-width tree search agent here does only one-depth lookahead search, and therefore, also takes a late turn,
while the MCTS agent gets around the opponents in a much earlier time primarily because she manages to lookahead for more steps with the help of heuristic node selection within the given budget.

\begin{figure}[!ht]
	\centering
	\includegraphics[width=170mm]{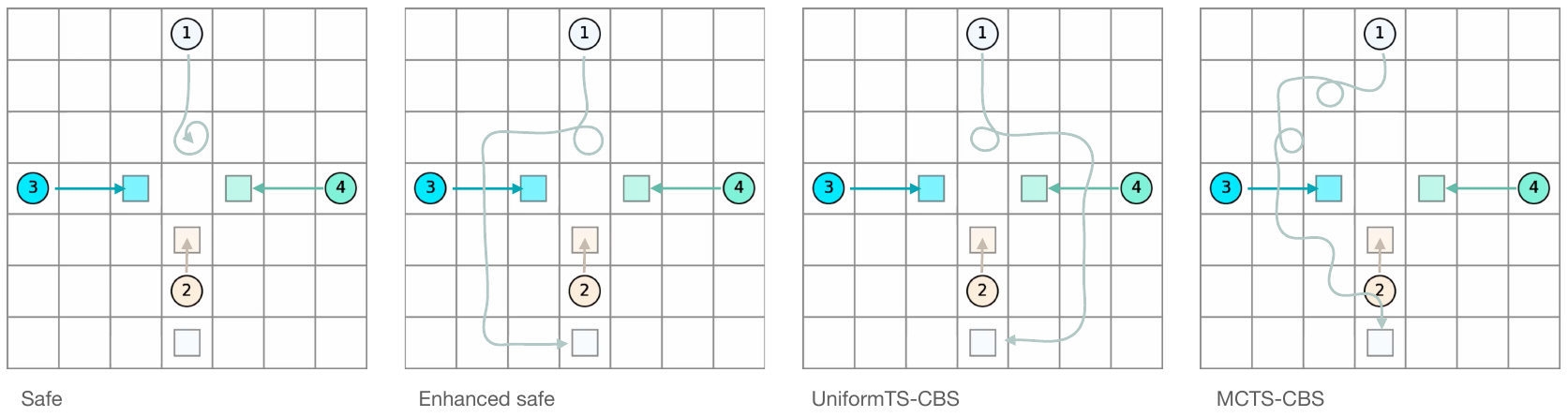}
	\caption{Agent 1 is the modelling agent using the respective planners noted at the lower left corner of each sub-figure, while agent 2, 3, and 4 are shortest-path agents. Squares with the corresponding colors are the respective goals.}
	\label{fig:probe}
\end{figure}

\subsubsection*{Case 5}
For a large map with a dense population of agents, we attach two videos to show the capability of our planners, especially the potential of tree search.
In both videos, we control agent 1 as the modelling agent, while all the others are naive A$^*$ agents ignoring others, and therefore, may collide into each other. Initially agents are randomly spawn on the map. In \textbf{large50a\_cbs.mp4} agent 1 is born at the upper left corner, while in \textbf{large50a\_mctscbs.mp4} agent 1 is born at the middle right part. In the former one  agent 1 directly enquires CBS plans at each move, while in the latter one she instead uses CBS plans as node heuristics and applies MCTS with pUCT to improve them in real-time. As one can see the vanilla version encounters many intermediate collisions, while the MCTS improved one perfectly avoids all collisions and finds a fairly short path.

%

\section{Computing Time}
Besides a comparison over the performances of our proposed planners, we here attach two additional tables for computing times, in order to help readers make more informed decisions on selecting suitable planners. 
Table~\ref{tab:comp_time_main} shows computing times for the planners in Table~\ref{tab:performance},
while Table~\ref{tab:comp_time_additional} are for the extra planners mentioned in Table~\ref{tab:additional_planner} in Appendix~\ref{apd:case_study}, with the following elaboration:
\begin{enumerate}
	\item Each cell records the average computing time in seconds, except for the POMDP planner where we explicitly write roughly 2 hours. We run 10 times, as the computing time is quite stable, for each one  to calculate the average (just two runs for the POMDP planner), with the standard deviation omitted to make the table more concise.
	\item The additional Tiny2a configuration means the one (3-by-3 maps with 2 agents) we present in Figure~\ref{fig:small_compare} in Appendix~\ref{apd:case_study}.
	\item Some planners need no replanning, therefore, we report the time they take for the initial plans.
	\item ``/'' means the planner is not feasible in that scenario.
\end{enumerate}

\begin{table}[!ht]
\small
\centering
\begin{tabular}{@{}lrrrrrrrrr@{}}
\toprule
                  & \textbf{Astar} & \textbf{Safe} & \textbf{MDP} & \textbf{RL} & \textbf{UniformTSRL} & \textbf{CBS} & \textbf{UniformTSCBS} & \textbf{MCTSCBSucb} & \textbf{MCTSCBSpuct} \\ \midrule
\textbf{Small2a}   & 6.40E-04                           & 3.96E-04                                    & 3.38E+00                                       & 1.67E-06                                        & 1.85E-01                                   & 1.24E-01                                       & 7.65E+01                                  & 3.87E+00                                & 3.93E+00                                 \\
\textbf{Square2a}  & 1.54E-03                           & 1.64E-03                                    & 1.11E+02                                       & 1.57E-06                                        & 2.57E-01                                   & 1.27E-01                                       & 7.80E+01                                  & 6.53E+00                                & 6.12E+00                                 \\
\textbf{Square4a}  & 1.55E-03                           & 1.62E-03                                    & /                                              & /                                               & /                                          & 6.47E-02                                       & 3.67E+01                                  & 4.12E+00                                & 4.34E+00                                 \\
\textbf{Medium20a} & 2.44E-03                           & 4.96E-03                                    & /                                              & /                                               & /                                          & 7.17E-02                                       & /                                         & 1.73E+01                                & 3.16E+01                                 \\
\textbf{Random}    & 1.29E-02                           & 1.52E-02                                    & /                                              & /                                               & /                                          & 6.04E-02                                       & /                                         & 1.81E+02                                & 2.19E+02                                 \\ \bottomrule
\end{tabular}
\caption{For the planners in Table 3}
\label{tab:comp_time_main}

\small
\begin{tabular}{@{}lrrrr@{}}
\toprule
                 & \textbf{POMDP*} & \textbf{QMDP*} & \textbf{UnifTSMDP} & \textbf{UnifTS-reactive} \\ \midrule
\textbf{Tiny2a}  & $\sim$2hr       & 4.16E-01       & 5.52E+01           & 4.48E-02                 \\
\textbf{Small2a} & /               & 5.10E+01       & 8.67E+01           & 7.70E-02                 \\ \bottomrule
\end{tabular}
\caption{For the planners in Table 5 in Appendix D}
\label{tab:comp_time_additional}
\end{table}


\end{document}